\documentclass[10pts]{amsart}
\usepackage{geometry} 
\usepackage[latin1]{inputenc} 
\usepackage[T1]{fontenc} 

\usepackage{amsmath,amsthm,amssymb,amsfonts,amsbsy,amscd,amstext}
\usepackage{dsfont,color}
\usepackage{graphics,graphicx}
\usepackage{epsfig,epstopdf,psfrag}
\usepackage{subfigure}
\usepackage{wrapfig}
\usepackage{url}
\usepackage{verbatim}
\usepackage{algorithm,algorithmic}
\usepackage{framed} 
\usepackage{hyperref}
\usepackage{pgf,tikz}
\usepackage{enumerate}

 \newtheorem{theorem}{Theorem}[section]
 \newtheorem{corollary}[theorem]{Corollary}
 \newtheorem{lemma}[theorem]{Lemma}
 \newtheorem{proposition}[theorem]{Proposition}
 \newtheorem{definition}[theorem]{Definition}

\newcommand{\J}{\mathbf{J}}
\newcommand{\DZ}{\Delta\mathbf{Z}}
\newcommand{\Z}{\mathbf{Z}}
\newcommand{\Om}{\Omega}
\newcommand{\bnu}{\boldsymbol{\nu}}
\newcommand{\Tau}{\boldsymbol{\tau}}
\newcommand{\p}{\boldsymbol{\,.\,}}
\newcommand{\dv}{\,\delta v}
\newcommand{\ds}{\,\delta s}

\newcommand{\s}{s\searrow0}

\newcommand{\A}{\mathbf{A}}

\newcommand{\V}{\nabla\textbf{V}_c}
\newcommand{\Vv}{\textbf{V}_c}
\newcommand{\As}{\mathbf{A}_{\tau}}
\newcommand{\PhiS}{\Phi_{\tau}}
\newcommand{\Vs}{\nabla_{\tau}\textbf{V}_c}

\newcommand{\E}{\mathbf{E}}
\renewcommand{\H}{\mathbf{H}}
\newcommand{\curl}{\textbf{curl\hspace{.01in}}}
\newcommand{\curlS}{\textbf{curl}_{\tau}}
\newcommand{\dvg}{\textbf{div}}
\newcommand{\dvgS}{\textbf{div}_{\tau}\hspace{.01in}}
\newcommand{\grad}{\nabla}
\newcommand{\gradS}{\nabla_{\tau}}
\newcommand{\tta}{\boldsymbol{\theta}}
\newcommand{\ttaS}{\boldsymbol{\theta}_{\tau}}
\author{Houssem Haddar}
\address{INRIA-Saclay/CMAP , Ecole Polytechnique, Route de Saclay, 91128 Palaiseau Cedex FRANCE}%
\author{Mohamed Kamel RIAHI}
\address{INRIA-Saclay/CMAP , Ecole Polytechnique, Route de Saclay, 91128 Palaiseau Cedex FRANCE}%
\title[Eddy current direct and inverse solvers]{3D direct and inverse solvers for eddy current testing of deposits in steam generator}

\begin{document}

\maketitle


 \begin{abstract}
	We consider the inverse problem of estimating the shape profile of an unknown deposit from a set of eddy current impedance measurements. The measurements are acquired with an axial probe, which is modeled by a set of coils that generate a magnetic field inside the tube. For the direct problem, we validate the method that takes into account the tube support plates, highly conductive part, by a surface impedance condition. For the inverse problem, finite element and shape sensitivity analysis related to the eddy current problem are provided in order to determine the explicit formula of the gradient of a least square misfit functional. A  geometrical-parametric shape inversion algorithm based on cylindrical coordinates is designed to improve the robustness and the quality of the reconstruction. Several numerical results are given in the experimental part. Numerical experiments on synthetic deposits, nearby or far away from the tube, with different shapes are considered in the axisymmetric configuration.
\end{abstract}
\section{Introduction}
	Eddy current testing simulation for the detection of cracks, default and deposit is a challenging research problem in non destructive evaluation, which has a major interest in many industrial applications. 
The induced eddy current created on the surface of a test-piece due to the presence of an electromagnetic field, generated by a moving probe, enables impedance measurement on the test-piece. 
The vector magnetic potential and electric scalar potential formulation is of common practice for the approximation of the eddy current solution in 3d-configuration. The finite element approximation of the eddy current direct problem can be achieved in this case with Lagrange finite element~\cite{MR2680968} instead of $\curl$-conform finite elements, while conserving continuity of the tangential component of the solution. 
	
	This work is focused on the shape reconstruction of deposits: conducting materials using time-harmonic eddy current measurements. We propose two approaches based on an efficient numerical model of the probe-defect interaction; the first consider the test-piece as a pileup of several layers and is parametrized by the cylindrical coordinates, and the second one proposes a regularization approach that smooth the descent direction of the inverse algorithm.  Both methods are potentially capable of treating  clogging on of the TSP.
	
	In this framework, a main issue is the ability to assess the 3D conductivity profile of the sample under test. The major difficulty encountered to achieve this aim is the non-linearity and ill-posedness of eddy currents inverse scattering models.
	The numerical method here proposed has been developed in order to treat efficiently
the typical situation encountered in eddy current testing where clogging are present in
an a priori known sub-region. The anomalies perturb locally the induced eddy current and therefore it is efficient to assume as
unknown the impedance measurement as the difference between the healthy signal and the one produced with the flowed part. 

\section{The industrial Problem}

Steam generators (SGs) are critical components in nuclear power plants. Heat produced in a nuclear reactor core is transferred as pressurized water of high temperature via the primary coolant loop into a SG, consisting of tubes in U-shape, and boils coolant water in the secondary circuit on the shell side of the tubes into steam. This steam is then delivered to the turbine generating electrical power. The SG tubes are hold by the broached quatrefoil tube support plates (TSP) with flow paths between tubes and plates for the coolant circuit, see Fig~\ref{quadrifolTSP}. 
\begin{figure}[!htbp]\centering
  \begin{minipage}[c]{0.45\linewidth}\includegraphics[scale=.3]{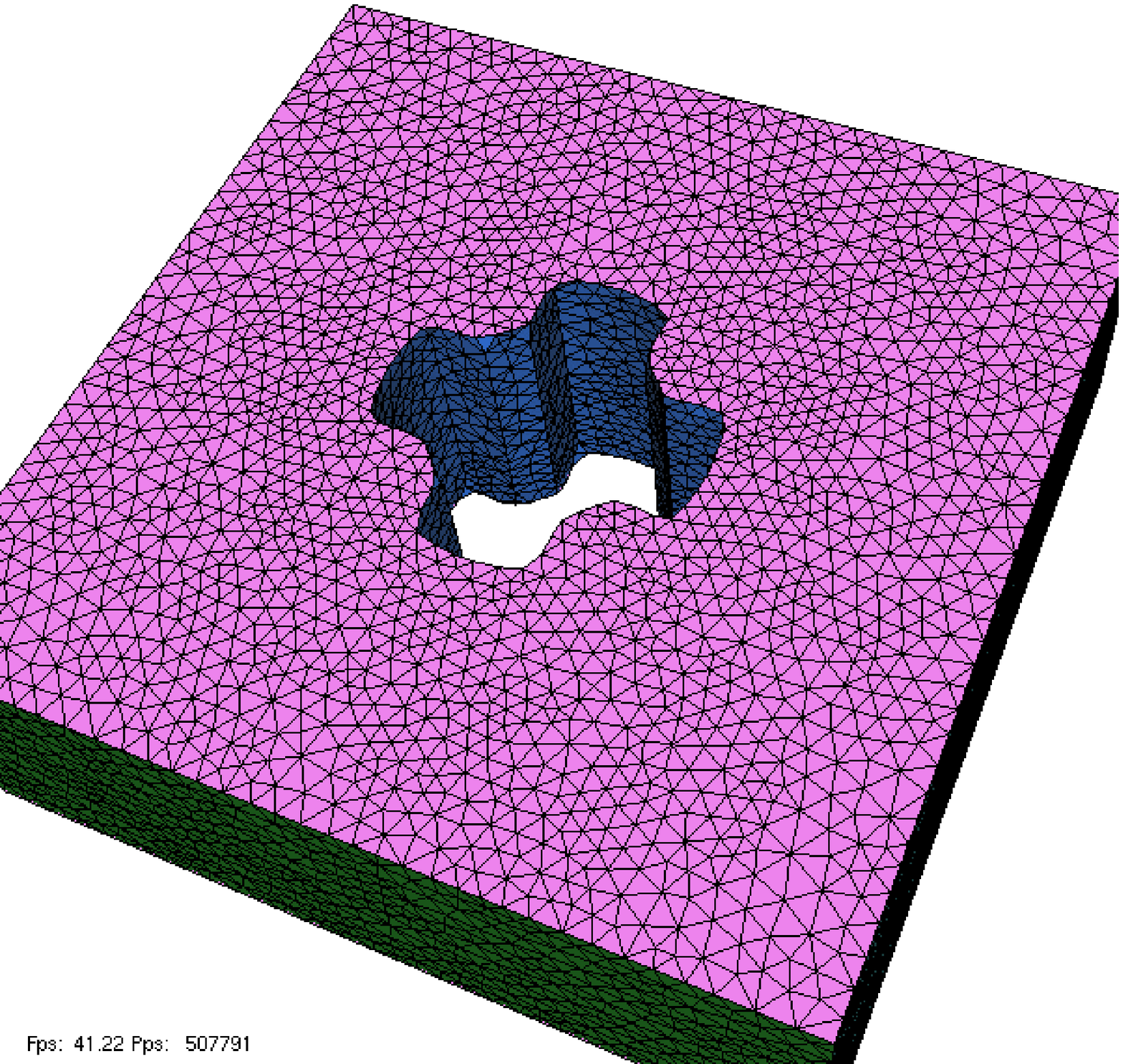}  \end{minipage}
  \begin{minipage}[c]{0.45\linewidth}\includegraphics[scale=.3]{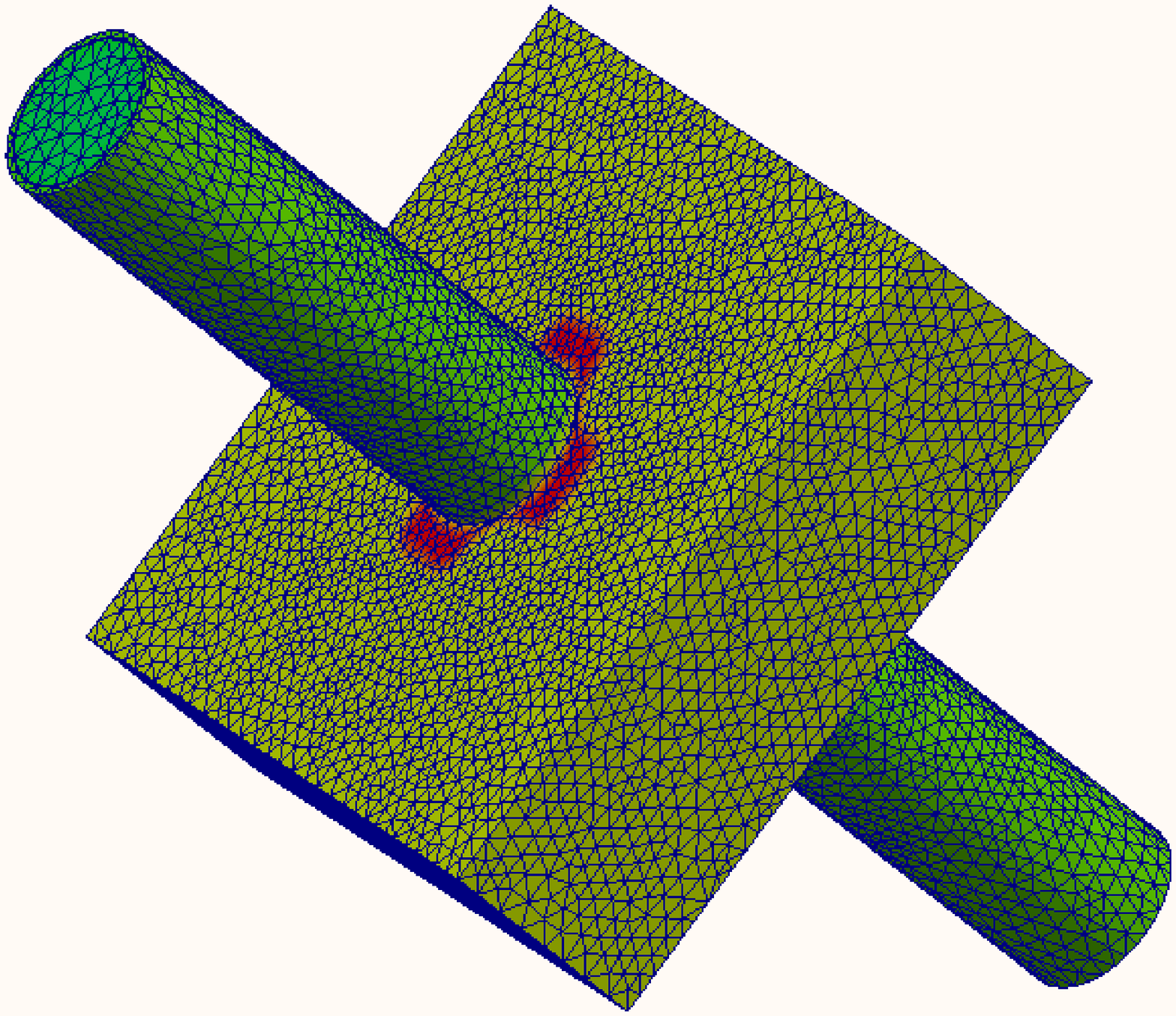}  \end{minipage}
\caption{The quatrefoil TSP is a conductive material (left) with the tube  also conductor material (right), which is 100\% clogged by magnetite (presented with red color).}\label{quadrifolTSP}
\end{figure}
	Due to the impurity of the coolant water in the secondary circuit, conductive magnetic deposits are observed on the shell side of the U-tubes, usually at the level of the quatrefoil TSP after a long-term exploitation of the SGs. Theses deposits could, by clogging the flow paths of coolant circuit between the tubes and the support plates, reduce the power productivity and even harm the structure safety. Without disassembling the SG, the lower part of the tubes -- which is very long -- is inaccessible for normal inspections. Therefore, a non-destructive testing procedure, called eddy current testing (ECT), is widely practiced in industry to detect the presence of defects, such as cracks, flaws, inclusions and deposits.
	
	  Eddy currents are created through a process called electromagnetic induction. When alternating current is applied to the conductor, such as copper wire, a magnetic field develops in and around the conductor (depending on its conductivity). This magnetic field expands as the alternating current rises to maximum and collapses as the current is reduced to zero. If another electrical conductor is brought into the close proximity to this changing magnetic field, current will be induced in this second conductor. Eddy currents are induced electrical currents that flow in a circular path. They get their name from ``eddies'' that are formed when a liquid or gas flows in a circular path around obstacles when conditions are right.
\begin{figure}[!htbp]\centering
\includegraphics[scale=.2]{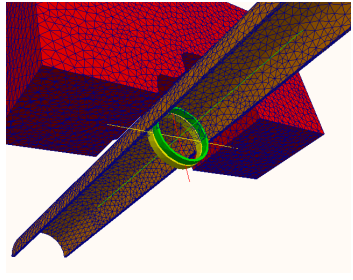}
\caption{Probing the TSP in steam generator.}\label{probingTSP}
\end{figure}

	In the ECT of steam generator, one introduces a probe consisting of two copper wire coils in the tube see Fig.~\ref{probingTSP}. Each of these coils is connected to a current generator producing an alternating current and to a voltmeter measuring the voltage change across the coil. One of the coils is excited by its current generator to create a primary electromagnetic field which in turn induces a current flow -- the eddy current -- in the conductive material nearby, such as the tube and the conducting support plates. Given the deposit-free case as background information, the presence of conducting deposits distorts the eddy current flow and leads to a current change in the two coils, which is measured by the linked voltmeters in terms of impedance. This measurement is called ECT signal that we use to identify the deposits. 
\begin{equation}\label{industparams}
\left.\begin{array}{|c|c|c|c|}
  \hline\text{Paramertes}\slash \text{Domain} & \text{TSP} & \text{Tube} & \text{Deposit}
\\\hline\text{Conductivity} \sigma & 1.75.e+06 & 0.97.e+06 & 61 
\\\hline\text{Relative permeability}\mu_r & 70 & 1.01 & 1.64
\\\hline\end{array}\right.
\end{equation}

\section{Direct and inverse numerical models and schemes}

Time harmonic Maxwell's equations for the electric field $\E$ and the magnetic field $\H$ on a domain $\Om=\cup_\ell\Om_\ell$ ($i=0$ for vacuum, $\ell=t$ for tube, $\ell=d$ for deposit and/or flaw and $\ell=p$ for plates) reads:

\begin{equation}\label{maxwell}
\begin{array}{ll}
\curl\H+(i\omega\epsilon_\ell-\sigma_\ell)\E=\J&\text{on } \cup_\ell\Omega_\ell:\text{ Maxwell-Amp\`ere},\\
\curl\E - i\omega\mu_\ell\H = 0.&\text{on } \cup_\ell\Omega_\ell:\text{ Maxwell-Faraday}.
\end{array}
\end{equation}

The Eddy Current model assume: $$\epsilon_i+\dfrac{i\sigma}{\omega}\cong\frac{i\sigma}{\omega},$$
which is the case with low frequency model e.g $\omega\epsilon\approx \text{55.78e-7}<<\sigma_\ell$. We are therefore concerned with the following equations:
\begin{equation}\label{MAMF}
\begin{cases}
\curl\H-\sigma_\ell \E=\J&\text{on } \cup_\ell\Omega_\ell:\text{ Maxwell-Amp\`ere},\\
\curl\E - i\omega\mu_\ell\H = 0.&\text{on } \cup_\ell\Omega_\ell:\text{ Maxwell-Faraday}.
\end{cases}
\end{equation}
 
From Eq.~\eqref{MAMF}$_1$, with the fact that $\sigma_0=0$ in $\Om_I$ we have :

$$\curl\H = \J\quad\text{on }\Om_I.$$
The current density $\J$ is required to be divergence free and uniformly distributed on a support included in $\Om_I$ (principally it models the solenoid source coil of the present problem). So, let $\J\in (L^2(\Om))^3$  with $\grad\J=0$ in $\Om_I$ as $\J:=:J_0[-r^{-1}y_{|_{\Om_I}},r^{-1}x_{|_{\Om_I}},0]_{\pm}$ ($J_0=|\J|$), where $r=\sqrt{x^2+y^2}$. 

Applying the divergence operator $\nabla\p$ on the Maxwell-Faraday equation Eq.~\eqref{MAMF}$_2$ we obtain:
\begin{equation}
\grad(\mu_\ell\H) = 0\text{ on }\cup_\ell\Om_\ell.
\end{equation}
For computational reasons we are obliged to limit the computational domain $\Om$ with an artificial surface $\partial\Om$, where we have to take into account the continuity of the tangential component of the magnetic fields $\H$ solution of the above equations. It is thus necessary to impose the boundary condition:
\begin{equation}
\H\times\bnu = 0 \text{ on } \partial\Om,
\end{equation}
in order to describe a perfect conductivity. 
In addition, we point out the fact that we have to consider the behavior of the magnetic field at the interface   limiting the insulator part $\Om_I$ and the conductor part $\Om_c$:
\begin{equation}\label{Hcontinuity}
\begin{cases}
\mu_0\H\,.\,\bnu_I  + \mu_c\H\p\bnu_c =0 &\text{ on } \partial\Om_I\cap\partial\Om_c, \\
\mu_0\H\times\bnu_I  + \mu_c\H\times\bnu_c =0 &\text{ on } \partial\Om_I\cap\partial\Om_c. 
\end{cases}
\end{equation}
In the above $\bnu_I$ and $\bnu_c$ are outward normal from the insulator part and the conductor part respectively. 
\subsection{Setting up the direct solver}
In this subsection, we will develop the setting of our direct solver that is based on the $\A\!-\!\V$ mixed formulation as magnetic vector potential $\A$ and a scalar electric potential $\V$ such that:
\begin{equation}\label{defA}
\begin{cases}
\E = i\omega\A +  \V &\text{ on } \Om,\\
\mu_\ell\H= \curl \A &\text{ on }\cup_\ell\Om_\ell,
\end{cases}
\end{equation}
where $\V$ is uniquely defined on the conductive $\Om_c$. Following the aboves equations, it is immediate that $\curl \E=i\omega\curl\A=i\omega\mu_\ell\H$ in $\Om$, hence Eq.~\eqref{MAMF}$_2$ is satisfied. Furthermore, $\mu\H$ is a solenoidal vector field in all the computational domain $\Om$ thus from the divergence theorem we have:
$$\int_\Om \grad\p(\mu_l\H) \dv=\int_{\Om_c\cup\Om_c} \grad\p(\curl\A) \dv= 0.$$
It yields that: 
\begin{eqnarray*}
\int_{\Om_c} \grad\p(\mu\H) \dv &=& \int_{\Om_I}\grad\p(\mu\H)\dv.\\
	\int_{\partial\Omega\cup\big(\partial\Om_I\cap\partial\Om_c\big)} \mu_c\H\p\bnu_c \,ds&=&	\int_{\partial\Omega\cup\big(\partial\Om_I\cap\partial\Om_c\big)}\mu_I\H\p\bnu_c \,ds.
\end{eqnarray*}
We obtain finally satisfaction of the Eq.~\eqref{Hcontinuity}$_1$.

In order to avoid singular system and make well-posed problem in the sense of the magnetic potential vector $\A$ is unique, it is classical and necessary to impose additional conditions, known as Coulomb gauge conditions
\begin{equation}\label{CGCdiv}\grad\p\A=0\text{ in } \Om,
\end{equation}
with the boundary condition $\A\p\bnu = 0$ on $\partial\Om$.

Let us go back to Maxwell-Amp\`ere equation~\eqref{MAMF}$_1$, by applying the divergence we obtain 
$$
-\grad\p(\sigma_c\E) = \grad\p\J \quad\text{ on } \Om_c,
$$
Where in the weak formulation, after an integration by part ,$\forall \varphi\in H^1(\Omega_c)$ we obtain:
 $$\int_{\Om_c}\sigma_c\E\p\grad\varphi-\int_{\partial\Om\cup\big(\partial\Om_I\cap\partial\Om_c\big)} \sigma_c\bnu\p\E\varphi\,ds=-\int_{\Om_c}\J\p\grad\varphi\,\dv +\int_{\partial\Om\cup\big(\partial\Om_I\cap\partial\Om_c\big)}\bnu\p\J\,.\varphi ds.$$
 By identification of integrals and using the fact that $\E=i\omega\A+\V$ in $\Om_c$ we obtain:
 \begin{eqnarray*}
 \int_{\Om_c}\sigma_c\big(i\omega\A+\V\big)\p\,\grad\varphi \,\dv&=&-\int_{\Om_c}\J\p\grad\varphi\,\dv\\
-\int_{\partial\Om\cup\big(\partial\Om_I\cap\partial\Om_c\big)} \sigma_c\big(i\omega\A+\V\big)\p\bnu\,.\varphi\,ds&=&\int_{\partial\Om\cup\big(\partial\Om_I\cap\partial\Om_c\big)}\J\p\bnu\,.\varphi ds.
 \end{eqnarray*}
 Hence, it is necessary to include the equations below as an additional constraints to the new problem that has as unknown the magnetic vector potential $\A$ and the scalar electric potential $\V$. The introduction of a gauge on the vector magnetic field $\A$ leads to a differential constraint; $\dvg\A=0$ in $\Omega$. A classical technique incorporates this constraint using a penalization term 
 $$
- \dfrac 1{\tilde\mu}\grad\dvg\A,
 $$
 in the Amp\`ere equation, where $\tilde\mu$ is a suitable average of $\mu$ in $\Om$.

 We are now in a good position to introduce the complete ($\A,\V$) strong formulation of our problem. It reads:
 \begin{equation}\label{strongPbAV}
 \begin{cases}
 \curl\big( \dfrac 1 \mu\curl \A\big) -\dfrac 1 {\tilde\mu}\grad\dvg\A-\sigma i\omega\A - \sigma\V = \J &\text{ on }\Om,\\
\dvg \big( i\omega\sigma\A+\sigma\V \big) = \dvg\J =0&\text{ on }\Om_c,\\
\big( \sigma i\omega\A+\sigma\V\big)\p\bnu = \J\p\bnu&\text{ on }\partial\Om_I\cap\partial\Om_c,\\
\A\p\bnu = 0 &\text{ on }\partial\Om,\\
\big(\dfrac{1}{\mu}\curl\A\big)\times\bnu = 0 &\text{ on }\partial\Om,
 \end{cases}
 \end{equation}
where $\V$ is determined up to an additive constant. We may thus (numerically) make a supplement condition such that
$$
\int_{\Om_{C_i}} \V \dv = 0 , \Om_{C_i} \text{ is any connex component subset of } \Om_c.
$$
This also could be incorporated under the global problem by penalization 
$$
\delta\tilde\sigma\V, \text{ in } \Om_c, \text{ with a small } \delta<<1.
$$  
Consider the space $H(\textbf{curl},\Omega)\cap H_0(\textbf{div},\Omega)$.
where 
$$H(\textbf{curl};\Omega):=\{u\in (L^2(\Omega)\big)^3 \,|\, \curl u\in(L^2(\Omega)\big)^3 \},$$
and 
$$
H(\textbf{div};\Omega):=\{ u\in(L^2(\Omega)\big)^3\,|\, \nabla\p u\in L^2(\Omega) \},
$$
also we have 
$$
H_0(\textbf{div};\Omega):=\{u\in H(\textbf{div};\Omega)\, |\, u\p\bnu_{|\partial\Omega}=0\}.
$$
Let us take test functions $\Phi\in H(\textbf{curl};\Omega)\cap H_0(\textbf{div};\Omega)$ and $\varphi\in H^1(\Omega_c)$ for the Eq.~\eqref{strongPbAV}$_1$ and the Eq.~\eqref{strongPbAV}$_2$ respectively. After integration by part we obtain the following weak formulations:
\begin{equation}\label{varf}
\begin{cases}
\displaystyle\int_{\Omega}\dfrac{1}{\mu}\curl\A\p\curl\overline{\Phi} \dv +\dfrac{1}{\tilde\mu}\int_{\Omega}\dvg\A\dvg\overline{\Phi} \dv
	- \int_{\Omega_c}\sigma(i\omega\A+\V)\p\overline{\Phi} \dv
	&= \int_{\Omega}\J\p\overline{\Phi} \dv.\\
\displaystyle\int_{\Omega_c}\sigma\big(i\omega\A+\V\big)\p\overline{\nabla\varphi} \dv&= \int_{\Omega_c}\J\p\overline{\nabla\varphi} \dv.
\end{cases}
\end{equation}

 Let us denote by $\tilde\varphi$ the solution of the Neumann problem:
\begin{equation*}\begin{cases}
 \Delta \tilde\varphi = \dvg\A \quad &\text{ on }\Omega,\\
 \nabla\tilde\varphi.\bnu = 0 \quad &\text{ on }\partial\Omega.
 \end{cases}\end{equation*} 
 
We have first $\nabla\tilde\varphi\in H(\textbf{curl};\Omega)\cap H_0(\textbf{div};\Omega)$ and $\tilde\varphi_{|_{\Omega_c}}\in H^{1}(\Omega_c)$. Therefore using $(\nabla\tilde\varphi,\tilde\varphi_{|_{\Omega_c}})$ as a test functions for Eq.~\eqref{strongPbAV}$_1$ and the Eq.~\eqref{strongPbAV}$_2$, we obtain immediately: 
\begin{equation*}
\int_{\Omega}\nabla\dvg\A\p\overline{\nabla\tilde\varphi} \dv = 0 =-\int_{\Omega}\dvg\A\p\overline{\dvg\nabla\tilde\varphi} \dv - \int_{\partial\Omega}\dfrac{\partial\tilde\varphi}{\partial \bnu}\overline{\dvg\A} \ds.
\end{equation*}
Thus
\begin{equation*}
\int_{\Omega}\dvg\A\p \overline{\Delta\tilde\varphi} \dv = \int_{\Omega}|\dvg\A|^2=0.
\end{equation*}
The solution of Eq.~\eqref{strongPbAV} satisfies the gauge condition~\eqref{CGCdiv}.

Obviously $(C^\infty_0(\Omega))^3\subset H(\textbf{curl};\Omega)\cap H_0(\textbf{div};\Omega)$, hence let $\Phi_1$ be a solution belonging to $(C^\infty_0(\Omega))^3$ taken as a test function. It yelds after integration by part the following variational formulation:
\begin{equation*}
\int_{\Omega}\dfrac{1}{\mu}\curl\A\p \overline{\curl\Phi}_1 \dv -\int_{\Omega_c}\big( \sigma i\omega\A+\sigma\V \big)\p\overline{\Phi}_1\dv= \int_{\Omega}\J\p\overline{\Phi}_1\dv .\\
\end{equation*}
We apply the same procedure taken as test function; $\Phi_2$ as any function that  belongs to $H(\textbf{curl};\Omega)$, we obtain:
\begin{equation*}
\int_{\Omega}\dfrac{1}{\mu}\curl\A\p\curl\overline{\Phi}_2 \dv -\int_{\Omega_c} \sigma\big( i\omega\A+\V\big)\p\overline{\Phi}_2\dv+\int_{\partial\Omega}\big(\bnu\times\dfrac{1}{\mu}\curl\A\big)\p\overline{\Phi}_2 \dv= \int_{\Omega}\J\p\overline{\Phi}_2 \dv.\\
\end{equation*}
When subtracting the tow equations above we conclude that 
\begin{equation*}
 \dfrac{1}{\mu}\curl\A \times\bnu = 0, \text{ on }\partial\Om   \text{ weakly in } H(\textbf{curl};\Omega)\cap H_0(\textbf{div};\Omega).
\end{equation*}
Finally the boundary condition~\eqref{strongPbAV}$_5$ is satisfied. 
The strong formulation Eqs.~\eqref{strongPbAV} of the eddy-current problem is well defined.
	
	In the following we will give a suitable well-posed weak variational formulation~\cite{MR2680968} that link its solution to the solution of the strong formulation Eq.~\eqref{strongPbAV}.   
We multiply Eq.~\eqref{varf}$_{2}$ by $\dfrac{-1}{i\omega}$ to obtain : 
\begin{equation*}
\displaystyle\dfrac{-1}{i\omega}\int_{\Omega_c}\sigma\big(i\omega\A+\V\big)\p\overline{\nabla\varphi} \dv= \dfrac{-1}{i\omega}\int_{\Omega_c}\J\p\overline{\nabla\varphi} \dv.
\end{equation*}	
and couple this with Eq.~\eqref{varf}$_1$ in a single mixed weak variational formulation, which writes:
\begin{align*}
&\int_{\Omega}\dfrac{1}{\mu}\curl\A\p\overline{\curl\Phi} \dv +\dfrac{1}{\tilde\mu}\int_{\Omega}\dvg\A\overline{\dvg\Phi} \dv
	- \dfrac{1}{i\omega}\int_{\Omega_c}\sigma(i\omega\A+\V)\p (i\omega\overline\Phi+\overline{\grad\varphi}) \dv \\
&=\int_{\Omega}\J\p\overline{\Phi} \dv - \dfrac{1}{i\omega}\int_{\Omega_c}\J\p\overline{\nabla\varphi} \dv.
\end{align*}
	For reason of simplicity and abbreviation, we define the sesquilinear form $\mathcal{L}(\A,\Vv,\Phi,\varphi)$ as the right-hand side of the above, which writes:
\begin{equation}\label{sesquil}\displaystyle
\mathcal{L}\big(\A,\Vv;\Phi,\varphi\big):=\int_{\Omega}\dfrac{1}{\mu}\curl\A\p\overline{\curl\Phi} \dv +\dfrac{1}{\tilde\mu}\int_{\Omega}\dvg\A\overline{\dvg\Phi} \dv
	- \dfrac{1}{i\omega}\int_{\Omega_c}\sigma(i\omega\A+\V)\p (i\omega\overline{\Phi}+\overline{\grad\varphi}) \dv .
\end{equation}

\subsection{Models for highly conductive parts}
	The TSP have a very high conductivity $\sigma_\text{p}$ as compered with the tube .i.e $\sigma_\text{t}$ and the corresponding skin depth is then very small. Taking into account the effect of TSP using the 3D model, described above, which requires a very thin mesh size (proportional to the skin depth) inside TSP and leads to a huge size of the discrete 3D problem. We hereafter explain how one can avoid meshing the volume of TSP by imposing appropriate impedance boundary condition (IBC) on its boundary. More precisely it is shown in~\cite{Durufle2006533} that electromagnetic field satisfies 
 \begin{equation}\label{IBC}
\bnu\times\H = -\dfrac{1}{\mathcal{Z}_{\Gamma_{p}}}\E_T,  \quad\text{ on } \Gamma_{p}.
\end{equation}
(up to $O(\delta^2)$) on $\Gamma_{p}$. In the above equation; $\mathcal{Z}_{\Gamma_{p}}:=\dfrac{1-i}{\delta\sigma_p}$ with the skin depth $\delta:=\sqrt{\frac{2}{\omega\mu\sigma_p}}$ and the tangential component of the electric field $\E_T=\bnu\times\big( \E\times\bnu\big)$.
 Therefore if $\delta$ is sufficiently small i.e. $\omega\sigma_\text{p}\mu$ is sufficiently large Eq.~\eqref{IBC} is a very good approximation.

 Taking into account the definition given at Eq.~\eqref{defA} we can express the above boundary condition with the magnetic vector potential and the scalar electric potential as  :

\begin{equation}\label{ROTA}
\bnu\times(\dfrac{1}{\mu_p}\curl\A )=  -\dfrac{1}{ \mathcal{Z}_{\Gamma_{p}}}\big(i\omega\As+\nabla_{\Gamma}\V \big), \quad\text{ on } \Gamma_{p}.
\end{equation}

where, $\As:=\bnu\times\big(\A\times\bnu\big)$ and $\dvgS:=\bnu\times\big(\V|_{\Gamma}\times\bnu\big)$ are the surface traces of the magnetic vector potential and the gradient of the scalar potential over the $\Gamma_p$ manifold.

	 In addition, the normal component given in the Eq.~\eqref{MAMF}$_1$ at the impedance surface reads: $\curl\H\p\bnu = \sigma_p\E\p\bnu + \J\p\bnu$, which we reformulate taking into account the IBC condition at Eq.~\eqref{IBC}. We have therefore on $\Gamma_p$
\begin{eqnarray*}
 \sigma_p\E\p\bnu + \J\p\bnu&=&-\dvg(\bnu\times\H),\\ 
&=&\dvgS\Big(\dfrac{1}{\mathcal{Z}_{\Gamma_{p}}}\big( \bnu\times\big( \E\times\bnu\big) \Big),\\
&=&\dfrac{1}{\mathcal{Z}_{\Gamma_{p}}} \dvgS\E_T.
\end{eqnarray*}
Consequently:
\begin{equation}\label{NORM}
 \sigma_p\big(i\omega\A+\V \big)\p\bnu  = \dfrac{1}{\mathcal{Z}_{\Gamma_{p}}} \dvgS\big(i\omega\As+\Vs),\quad\text{ on } \Gamma_{p}.
\end{equation}
having used Eq.~\eqref{defA} and the fact that $\J=0$ is not supported in the TSP.

	The impedance surface term of the weak formulation of Eq.~\eqref{strongPbAV} at the interface $\Gamma$ of the TSP writes:
\begin{eqnarray}\label{ibcROTA}
\int_{\Gamma_{p}} \big(\bnu\times(\frac{1}{\mu_p}\curl\A)\big)\p\overline{\PhiS} \ds_{p}&=& -\dfrac{1}{\mathcal{Z}_{\Gamma_{p}}}\int_{\Gamma_{p}}(i\omega\As+\Vs)\p\overline{\PhiS} \ds_p,\\
\label{ibcNORM}
\int_{\Gamma_{p}}\sigma_p(i\omega\A+\V)\p\bnu\,\overline{\varphi} \ds_{p} &=&-\dfrac{1}{\mathcal{Z}_{\Gamma_{p}}}\int_{\Gamma_{p}}(i\omega\As+\Vs)\p\overline{\gradS\varphi} \ds_{p},
\end{eqnarray}
where~Eq.~\eqref{ibcROTA} is a direct consequence of Eq.~\eqref{ROTA} and Eq.~\eqref{ibcNORM} is consequence of Eq.~\eqref{NORM} having used integration by part.

\subsection{Impedance measurements of 3D deposit}
\begin{definition}\label{definitionDZ}
$$\DZ_{kl}=\int_{\Gamma} (\E_{l}^{0}\times\H_{k} - \E_{k}^{}\times\H_{l}^{0})\p\bnu\, d\Gamma$$
\end{definition}
\begin{lemma}
The volume impedance measured with the coil $k$ in the electromagnetic field induced by the coil $l$ writes:
\begin{eqnarray}\label{impedkl}
\displaystyle
\DZ_{kl}:=\frac{1}{i\omega|\J|}\frac{\mu_0-\mu_d}{\mu_d\mu_0}\int_{\Om_{d}}\big( \curl\E_k\p\curl\E_l^0\big)\dv 
+\frac{\sigma_d-\sigma_0}{|\J|}\int_{\Om_d}\E_k\p\E_l^0 \dv,
\end{eqnarray}
where $\E_l^0$ refers to the electric field propagating in vacuum.  
\end{lemma}

\begin{proof}
Applying the divergence theorem we have,
\begin{eqnarray}\label{dvgappendix}
\int_{\Gamma} (\E_{l}^{0}\times\H_{k} - \E_{k}^{}\times\H_{l}^{0})\p\bnu\, d\Gamma
&=&\int_{\Om_{d}} \nabla\p(\E_{l}^{0}\times\H_{k} - \E_{k}^{}\times\H_{l}^{0})\dv\\
&=&\int_{\Om_{d}} \nabla\p(\E_{l}^{0}\times\H_{k} )\,\dv-\int_{\Om_{d}}\nabla\p( \E_{k}^{}\times\H_{l}^{0})\dv.\nonumber
\end{eqnarray}

Using Eq.~\eqref{maxwell}, in one hand we have:
\begin{eqnarray}\label{El0Hk}
\int_{\Om_{d}} \nabla\p(\E_{l}^{0}\times\H_{k} )\,\dv
&=&\int_{\Om_{d}} \big(\H_{k}\p\curl\E_{l}^{0}-\E_{l}^{0}\p\curl\H_{k}\big)\p\dv\\
&=&\frac{\mu_0}{i\omega\mu_{d}\mu_0}\int_{\Om_{d}} \curl\E_k\p\curl\E_l^0\dv + (i\omega\epsilon-\sigma_d)\int_{\Om_{d}}\E_l^0\p\E_k\dv\nonumber
\end{eqnarray}
and in the other hand:
\begin{eqnarray}\label{EkHl0}
\int_{\Om_{d}} \nabla\p(\E_{k}\times\H_{l}^0 )\,\dv
&=&\int_{\Om_{d}} \big(\H_{l}^0\p\curl\E_{k}^{}-\E_{k}^{}\p\curl\H_{l}^{0}\big)\p\dv\\
&=&\frac{\mu_d}{i\omega\mu_d\mu_0}\int_{\Om_{d}} \curl\E_l^0\p\curl\E_k\dv  + (i\omega\epsilon-\sigma_0)\int_{\Om_{d}}\E_k\p\E_l^0 \dv.\nonumber
\end{eqnarray}
We substitute Eqs~\eqref{El0Hk}-\eqref{EkHl0} in Eq.~\eqref{dvgappendix}, we obtain:
\begin{equation}\label{DZkl}
|\J|\DZ_{kl}= \frac{\mu_0-\mu_d}{i\omega\mu_d\mu_0}\int_{\Om_{d}}\big( \curl\E_k\p\curl\E_l^0\big)\dv +(\sigma_0-\sigma_d)\int_{\Om_d}\E_k\p\E_l^0 \dv.
\end{equation}
\end{proof}
Let us now gives the involving formulation of impedance measurements signals:
\begin{equation}\label{impedmod}
\begin{cases}
{\bf Z}_{FA} = \frac{i}{2}\big(\DZ_{11}+\DZ_{12}\big) \quad\text{ absolute mode},\\
{\bf Z}_{F3} =\frac{i}{2} ((\DZ_{11}-\DZ_{22}\big) \quad\text{ differential mode}.
\end{cases}
\end{equation}
Taking into account the definition of the impedance given at Eq.~\eqref{impedkl},  Therefore formulas at Eq.~\eqref{impedmod} read:
\begin{eqnarray*}
{\bf Z}_{FA}&=&\frac{1}{2\omega|\J|}\frac{\mu_0-\mu_d}{\mu_d\mu_0}\int_{\Om_d}\curl\E_1 \p\curl\big(\E_1^0+\E_2^0\big)\dv \\
&&+ \frac i2\frac {\sigma_0-\sigma_d}{|\J|}\int_{\Om_d}\E_1 \p\big(\E_1^0+E_2^0\big)\dv.
\end{eqnarray*}
and 
\begin{eqnarray*}
{\bf Z}_{F3}&=&\frac{1}{2\omega|\J|}\frac{\mu_0-\mu_d}{\mu_d\mu_0}\int_{\Om_d}\curl\E_1\p\curl\E_1^0-\curl\E_2\p\curl\E_2^0\dv\\
&&+\frac i2\frac {\sigma_0-\sigma_d}{|\J|}\int_{\Om_d}\E_1\p\E_1^0-\E_2\p\E_2^0\dv.
\end{eqnarray*}
\subsection{Impedance signals for IBC models}

\begin{lemma}The surface impedance measure taking into account the impedance boundary condition at Eq.~\eqref{IBC} writes: 
\begin{equation}\label{}
\DZ_{kl}=\int_{\Gamma}(\H_k\times\bnu)\p\Big(\mathcal{Z}_{\Gamma}\H^0_l\times\bnu - \E^0_{l,T} \Big) \,d\Gamma,
\end{equation}
as well
 \begin{equation}\label{}
\DZ_{kl}=\int_{\Gamma_d}\E_{k,T}\p\big( \H_{l}^{0}\times\bnu-\dfrac{1}{\mathcal{Z}_{\Gamma}} \E_{l,T}^{0}\big)\,d\Gamma.
\end{equation}

\end{lemma}
\begin{proof}
Using the surface integral definition of the impedance $\DZ$ at Eq. \eqref{definitionDZ} we have
\begin{eqnarray}
\int_{\Gamma} (\E_{l}^{0}\times\H_{k} - \E_{k}^{}\times\H_{l}^{0})\p\bnu\, d\Gamma &=&\int_{\Gamma}\bnu\p(\E_k\times\H^0_l)-\bnu\p(\E^0_l\times\H_k)\,d\Gamma\nonumber
\\&=&\int_{\Gamma} \H^0_l\p(\bnu\times\E_k)-\H_k\p(\bnu\times\E_l^0)\,d\Gamma\nonumber
\\&=&\int_{\Gamma}\E_{k,T}\p(\H^0_l\times\bnu) - \E^0_{l,T}\p(\H_k\times\bnu)\label{ApassB}
\\&=&\int_{\Gamma}(\mathcal{Z}_{\Gamma}\H_k\times\bnu)\p(\H^0_l\times\bnu)-\E^0_{l,T}\p(\H_k\times\bnu)\,d\Gamma\label{passB}
\\&=&\int_{\Gamma}(\H_k\times\bnu)\p\Big(\mathcal{Z}_{\Gamma}\H^0_l\times\bnu - \E^0_{l,T} \Big) \,d\Gamma\nonumber.
\end{eqnarray}
 Having used Eq.~\eqref{IBC} to write Eq.~\eqref{passB} from Eq.~\eqref{ApassB}. It is worth noting that different interpretation gives en equivalent formulation of the impedance 
$
\DZ_{kl}=\int_{\Gamma_d}\big( \bnu\times\H_{l}^{0}-\dfrac{1}{\mathcal{Z}_{\Gamma}} \E_{l,T}^{0}\big)\p\E_{k,T} \,d\Gamma.
$
 when at Eq \eqref{ApassB} we replace $\H_k\times\bnu$ instead of $\E_{k,T}$ according to the IBC at Eq.~\eqref{IBC}.
\end{proof}

\section{The inverse problem by shape sensitivity analysis}
The inverse problem is a shape optimization problem that aims at finding the shape of a conductive domain addressed by the electromagnetic field due to the presence of probes. The optimization problem consists in minimizing a least squared impedance signal gaps, which is the difference between a computed signal and a measured one. Obviously the minimization is based on a gradient evaluation of the shape function. Because of the non linearity of the signal in regards to the shaped domain (where the impedance signal is measured), we use adjoint state to explicitly evaluate the shape gradient. 

\begin{proposition}
Consider a continuously differentiable, bijective map $T_s:\mathbb{R}^3\longrightarrow\mathbb{R}^3$ that transform $\Omega$ to $T_s(\Omega):=\Omega^s$. For any vector field $\mathcal{B}$ belongs to $H(\text{curl},\Omega)$ the following equalities hold on $H(\text{curl},\Omega)$
\begin{equation}
\big(\curl \mathcal{B}\big)\circ T_{s} =\dfrac{DT_{s}}{det(DT_{s})} \curl\big(DT_{s}^t\mathcal{B}\circ T_{s}\big),
\end{equation}
Whereas for $\mathcal{B}$ belonging to $H(\dvg,\Omega)$ following equality holds on $H(\dvg,\Omega)$
\begin{equation}
(\dvg\mathcal{B})\circ T_{s} = \dfrac{1}{det(DT_{s})}\dvg\big(det(DT_{s})DT_{s}^{-1}\mathcal{B}\circ T_{s}\big).
\end{equation}
And for all $\varphi\in H^1(\Omega)$
\begin{equation}
(\grad\varphi)\circ T_s = DT_s^{-t}\grad(\varphi\circ T_s).
\end{equation}
\end{proposition}
Let us denote by $j_{s}=det(DT_{s})$ and by $w_{s}=j_{s}\|DT^{-1}_{s}\bnu\|_{\mathbb{R}^3}$ 
\begin{proposition}
\begin{eqnarray*}
\dfrac{\partial}{\partial s}\Big(j_{s}\Big)_{s=0} &=&\dvg \tta,
\\\dfrac{\partial}{\partial s}\Big(\dfrac{1}{j_{s}}\Big)_{s=0} &=& -\dvg \tta,
\\ \dfrac{\partial}{\partial s}\big(w_{s}\big)_{s=0}&=& \dvgS\tta,
\\ \dfrac{\partial}{\partial s}\Big( \dfrac{DT_{s}^tDT_{s}}{j_s}\Big)_{s=0}&=&-\dvg(\tta)I + D\tta + D\tta^t
\\ \dfrac{\partial}{\partial s}\Big( \dfrac{DT_{s}}{j_s}\Big)_{s=0} &=& -\dvg(\tta) I + D\tta,
\end{eqnarray*}
\end{proposition}
\begin{proof} The results could be proven with the taylor first order expansion :
\begin{eqnarray*}
j_{s} = 1 + s\dvg(\tta) + o(\|s\tta\|)
\\ DT_{s} = I + sD\tta
\end{eqnarray*}
\end{proof}

We recall hereafter some useful vector identities for the $\curl, \dvg$ and $\grad$ operator applied to complex or real valued vectors ${A}$, ${B}$ and ${C}$ that belong to $\mathbb{C}^3$.
\begin{eqnarray}
{A}\p{B}\times{C} &=& {B}\p{C}\times{A}={C}\p{A}\times{B}\label{vecid1}.
\\{A}\times({B}\times{C}) &=& {C}\times({B}\times{A}) = ({A}\p{C}){B}-({A}\p{B}){C}\label{vecid2}.
\\\dvg(a{A}) &=& a\dvg({A}) + {A}\p\grad(a)\label{vecid3}.
\\\curl({A}\times{B}) &=& \dvg({B}){A} - \dvg({A}){B} + D{A}{B} - D{B}{A}\label{vecid4}.
\\{A}\times\curl({B})&=& D{B}^t{A} - D{B}{A}\label{vecid5}.
\\\grad({A}\p{B}) &=& {A}\times\curl{B} + {B}\times\curl{A} + D{B}{A} + D{A}{B}\label{vecid6}.
\\\dvg({A}\times{B}) &=& {B}\p\curl{A} - {A}\p\curl{B}.\label{vecid6}.
\end{eqnarray}
The shape derivative of a scalar complex valued function as $\Vv$ is given by 
\begin{equation}
\Vv'(\Omega,\tta):= \dot{\bf V}_{c} - \grad\Vv(\Omega)\p\tta.
\end{equation}
	We define the shape and material derivative for a given function $\E\in \big(H^1(\Om)\big)^3$ as
\begin{definition}
We denote by $\dot{\E}(\Omega)$ the material derivative of $\E$ that verifies:
$$
 \dot{\E}(\Omega)  = \lim_{\s}\dfrac{1}{s}\big(\E(\Omega^s)\circ T_{s} - \E(\Omega) \big).
$$
Consider an extention of $\E$ from $\mathcal H$ to $\mathcal H$ denoted $\E$ as well. We have in the weak sense on $H(curl,\Omega)$
$$\lim_{\s}\dfrac{1}{s}(\E(\Omega^s)-\E) \rightharpoonup \dot{\E} - D\E.\tta$$
\begin{equation}\label{shapeE2shapeA}
\E'=\dot{\E} - D\E.\tta
\end{equation}
As $\E=i\omega\A+\V$ we have by linearity $\E^\prime=i\omega\A^\prime+\V^\prime$.
\end{definition}
\begin{proposition}
The shape difference quaution $\mathcal W := \lim_{\s} \dfrac{1}{s}(DT_s^t\E(\Omega^s)\circ T_s - \E(\Omega))$ verifes:
$$
\mathcal W = D\tta^t \E + D\E \tta + \E'
$$
\end{proposition}
\begin{proof}
\begin{eqnarray*}
\dfrac{1}{s}( DT_s^t\E(\Omega^s)\circ T_s - \E )&=&\dfrac{1}{s}\Big( (DT_s^t\E(\Omega^s)\circ T_s -\E(\Omega^s)\circ T_s) + (\E(\Omega^s)\circ T_s- \E(\Omega) \Big)\\
&=&\dfrac{1}{s}\Big( (DT_s^t\E(\Omega^s)\circ T_s -\E(\Omega^s)\circ T_s) \Big) + \dfrac{1}{s}\Big(\E(\Omega^s)\circ T_s- \E(\Omega^s) \Big)\\
&&+\dfrac{1}{s}\Big(\E(\Omega^s)-\E(\Omega) \Big),
\end{eqnarray*}
which converge weakly in $H(curl,\Omega)$ to $ D\tta^t \E + D\E \tta + \E'$.
\end{proof}
\subsection{Shape sensitivity for impedance signal response}
$$
\mathcal{J}_{\ast}(\Omega_d):=\dfrac{1}{2}\int_{z_{min}}^{z_{max}} | \Z_{\ast}(\Omega_d;\xi)- \Z_{\ast}(\Omega_d^\star;\xi)|^2 d\xi
$$
where $\ast$ is either $FA$  or $F3$ type signal measurement.
Let us consider the following perturbation of the identity transformation: 
\begin{equation*}
\left.\begin{array}{ll}
T_{\tta}:&L^2(\mathbb{R})^3 \mapsto  L^2(\mathbb{R})^3\\
 &\Omega_d^0\longrightarrow\Omega_d=(Id+\tta)(\Omega_d^0).
 \end{array}\right.
\end{equation*}
where $\Om_d:=\{x\in\mathbb{R}^3\backslash x= \tilde x + \tta(\tilde x),\quad \forall \tilde x\in\Om_d^0\}.$

The shape derivative of the cost functional $\mathcal{J}(\Omega_d)$ is defined as follows:

$$
\mathcal{J}^\prime(\Omega_d^0,\tta)=\lim_{\s} \dfrac{\mathcal{J}(\Omega_d)-\mathcal{J}(\Omega_d^0)}{s}.
$$ 

Analogously we define the shape derivative of the impedance signal type measurement $\Z_{\ast}^\prime$ and $\Z_{kl}^\prime$.

$$
\mathcal{J}^\prime(\Omega^0_d,\tta) = \int_{z_\text{min}}^{z_\text{max}} \Re\Big\langle \Z_{\ast}'(\Omega;\xi)(\tta),\overline{\big(\Z_{\ast}(\Omega^0_d;\xi)-\Z_{\ast}(\Omega^\star_d;\xi)  \big)}\Big\rangle \,d\xi
$$
In order to obtain the governing equation of a shape function, it is common practice to through the material derivative. We develop here after some preliminaries calculus that helps us to burn several lines in our proofs.  

Recall that $\bnu$ denote the outward directed unit normal to $\Gamma$. We define the jump $\Big[.\Big]_\pm$  
$$
\Big[F\Big]_\pm := \lim_{t\rightarrow 0} F(x+t\bnu) - \lim_{t\rightarrow 0} F(x-t\bnu) \text{ for } x\in\Gamma.
$$
of the continuous extension of a function $F$ from the exterior and the interior of $\Omega$, respectively.

\subsubsection{Preliminaries calculus}
Let us define the shape functionals $\mathcal{C}_\text{curl}(\Omega)$, $\mathcal{C}_\dvg(\Omega)$ and $\mathcal{C}_\text{mix}(\Omega)$ as follows
\begin{eqnarray*}
\mathcal{C}_\text{curl}(\Omega) &=& \int_{\Om}\dfrac{1}{\mu}\curl\A(\Omega)\p\overline{\curl\Phi} \dv\\
\mathcal{C}_\dvg(\Omega) &=& \dfrac{1}{\tilde\mu}\int_{\Om}\dvg\A(\Omega)\overline{\dvg\Phi} \dv\\
\mathcal{C}_\text{mix}(\Omega)&=&\int_{\Om_c}\sigma\big( i\omega\A(\Omega)+\V(\Omega)\big)\p\overline{\big(i\omega\Phi+\grad\varphi\big)}\, dv
\end{eqnarray*}
\begin{lemma}\label{shape-derivatives}
The Eulerian derivative of the above shape functions exist on $\Omega$ and are given by 
 \begin{align*}
 \displaystyle\mathcal{C}_\text{curl}'(\Omega,\tta)=&   \int_{\Omega}\dfrac{1}{\mu} \curl\A'\p \overline{\curl\Phi} \dv
 -\displaystyle\int_{\Omega}  (\tta^t\bnu)\Big[\dfrac{1}{\mu}\Big]_\pm \bnu^t\curl\A\p\overline{\bnu^t\curl\Phi}\ds
\\ &  -\displaystyle\int_{\Omega}  \Big[\mu\Big]_\pm \frac{1}{\mu}\curlS\A\p\frac{1}{\mu}\overline{\curlS\Phi} \ds
\\\mathcal{C}_\dvg'(\Omega;\tta)   =&\dfrac{1}{\tilde\mu}\int_{\Omega}\dvg\A' \dvg\Phi\,dv.
\\\mathcal{C}_\text{mix}'(\Omega_c,\tta)=&\dfrac{1}{i\omega}\int_{\Om_c} \sigma\big( i\omega\A'+\V'\big) \p\big(i\omega\overline{\Phi}+\overline{\grad\varphi}\big) \dv
\\&-\dfrac{1}{i\omega}\int_{\Gamma} (\tta^t\bnu) \Big[\sigma\Big]_\pm \big(i\omega\As+\Vs\big)\p\big(i\omega\overline{\PhiS}+\overline{\gradS\varphi}\big) \ds
\end{align*}
\end{lemma}
\begin{proof}
\begin{enumerate}[i)]
\item 
\begin{eqnarray*}
\mathcal{C}_\text{curl}(\Omega^s)&=&\int_{\Omega}\big( \curl\A(\Omega_{d}^{s})\big)\circ T_{s}\p \overline{\big(\curl\Phi\big)\circ T_{s}}\,j_{s}\dv
\\&=&\int_{\Omega}\big(\dfrac{1}{\mu}\dfrac{DT_{s}}{j_{s}}\curl(DT_{s}^t\A(\Omega_{d}^{s})\circ T_{s}) \big)
	\p \overline{\big( \dfrac{DT_{s}}{j_{s}}\curl(DT_{s}^t\Phi\circ T_{s})\big)} \, j_s\dv
\\&=&\int_{\Omega}\big(\dfrac{1}{\mu}\dfrac{DT_s^tDT_{s}}{j_{s}}\curl(DT_{s}^t\A(\Omega_{d}^{s})\circ T_{s})\big)\p \overline{\big(\curl(DT_{s}^t\Phi\circ T_{s}) \big)} \dv	
\\&=&\int_{\Omega}\big(\dfrac{1}{\mu}\dfrac{DT_s^tDT_{s}}{j_{s}}\curl(DT_{s}^t\A(\Omega_{d}^{s})\circ T_{s})\big)\p \overline{\big(\curl(DT_{s}^t\Phi\circ T_{s}) \big)} \dv	
\end{eqnarray*}
So:
\begin{eqnarray*}
\dfrac{1}{s}\big( \mathcal{C}_\text{curl}(\Omega^s) -\mathcal{C}_\text{curl}(\Omega^0) \big)
&=&\dfrac{1}{s}\displaystyle\int_{\Omega} \dfrac{1}{\mu}\Big(\big(\dfrac{DT_s^tDT_{s}}{j_{s}}-I \big)\p \curl(DT_{s}^t\A(\Omega_{d}^{s})\circ T_{s})\Big)\p \overline{\big(\curl(DT_{s}^t\Phi\circ T_{s}) \big)} \dv
\\&&+\dfrac{1}{s}\displaystyle\int_{\Omega} \dfrac{1}{\mu}\Big(\curl\big(DT_s^t\A(\Omega_{d}^{s})\circ T_{s}\big)-\curl\A \Big)\p \overline{\big(\curl(DT_{s}^t\Phi\circ T_{s}) \big)} \dv
\\&&+\dfrac{1}{s}\displaystyle\int_{\Omega} \dfrac{1}{\mu}\curl\A \p \overline{\big(\curl((DT_{s}^t-I)\Phi\circ T_{s}) \big)} \dv.
\end{eqnarray*}
Thus the Eulerian derivative $\mathcal{C}_\text{curl}'(\Omega,\tta)=\lim_{\s}\dfrac{1}{s}\big( \mathcal{C}_\text{curl}(\Omega^s) -\mathcal{C}_\text{curl}(\Omega^0) \big)$ writes
\begin{eqnarray}
\mathcal{C}_\text{curl}'(\Omega,\tta)&=
& \nonumber
\int_{\Omega}\dfrac{1}{\mu} \big(D\tta + D\tta^t  -\dvg\tta I \big)\curl\A\p \overline{\curl\Phi}\dv +\int_{\Omega}\dfrac{1}{\mu} \curl(\mathcal W)\p \overline{\curl\Phi} \dv\nonumber
\\&=&\int_{\Omega}\dfrac{1}{\mu} \big(D\tta + D\tta^t  -\dvg\tta I \big)\curl\A \p \overline{\curl\Phi} \dv +\int_{\Omega}\dfrac{1}{\mu} \curl\A'\p \overline{\curl\Phi} \dv \nonumber
\\&&+\int_{\Omega}\dfrac{1}{\mu} \curl(D\tta^t\A+D\A\tta)\p\overline{\curl\Phi} \dv + \displaystyle\int_{\Omega} \dfrac{1}{\mu}\curl\A \p \overline{\curl(D\tta^t\Phi + D\Phi\tta)}\nonumber
\\&=&\int_{\Omega}\dfrac{1}{\mu} \curl\A'\p \overline{\curl\Phi} \dv \nonumber
\\&&+\int_{\Omega}\dfrac{1}{\mu} \bigg(\curl(D\tta^t\A+D\A\tta) + \big(D\tta + D\tta^t  -\dvg\tta I \big)\curl\A \bigg)\p \overline{\curl\Phi} \dv  \label{ddd1}
\\&&+ \displaystyle\int_{\Omega} \dfrac{1}{\mu}\curl\A \p \overline{\curl(D\tta^t\Phi + D\Phi\tta)} \dv\label{ddd2}.
\end{eqnarray}
Remark first the following identities:
\begin{align}
( D\tta^t\curl\A ) +& ( D\tta^t\curl\A - \dvg(\tta)\curl\A ) \nonumber
\\&=(\curl\A\times\curl\tta +D\tta\curl\A) - \Big(\curl(\curl\A\times\tta)-D(\curl\A)\tta \Big)\nonumber
\\&= \curl\A\times\curl\tta  + D\tta\curl\A+D(\curl\A)\tta - \curl(\curl\A\times\tta)\nonumber
\\&= \grad\big(\curl\A\p\tta\big)-\tta\times\big( \curl\curl\A\big) - \curl\big(\curl\A\times\tta\big)\nonumber
\\&= \grad\big(\curl\A\p\tta \big)-\curl\big(\curl\A\times\tta\big)+ \mu\sigma(i\omega\A+\V + \J)\times\tta\label{ida}
\end{align}
having replaced $\curl\curl\A$ in Eq.~\eqref{strongPbAV} for the last line. In addition we have:
\begin{align}
\curl\big(D\tta^t\A+D\A\tta \big) &= \curl\Big(\A\times\curl\tta + D\tta\A+ D\A\tta\Big)\nonumber
\\& = \curl\Big(\grad\big(\A\p\tta\big) - \tta\times\curl\A \Big)\nonumber
\\& = \curl\big(\curl\A\times\tta\big)\label{idb}.
\end{align}
 Hence, in regard to Eq.\eqref{idb} the domain integral.~\eqref{ddd2} leads to : 
\begin{align}
  \int_{\Omega} \dfrac{1}{\mu}&\curl\A \p \overline{\curl(D\tta^t\Phi + D\Phi\tta)}=  \int_{\Omega} \dfrac{1}{\mu}\curl\A \p \overline{\curl(\curl \Phi\times\tta)}\nonumber
 \\&=\displaystyle\int_{\Omega} \curl(\dfrac{1}{\mu}\curl\A) \p \overline{\curl \Phi\times\tta} \dv + \int_{\Gamma}\Big[\bnu\times\dfrac{1}{\mu} \curl\A\p\overline{\curl\Phi\times\tta}  \Big]_\pm\ds \nonumber
 \\&=\displaystyle\int_{\Omega}  (\sigma(i\omega\A+\V+\J) \p \overline{\curl \Phi\times\tta} \dv 
 	+ \int_{\Gamma}\Big[
					(\tta^t\curl\A)(\bnu^t\overline{\curl\Phi}) - (\tta^t\bnu)\dfrac{1}{\mu}\curl\A\p\overline{\curl\Phi}
					\Big]_\pm\ds  \nonumber
 \\&= - \int_{\Omega}  (\sigma(i\omega\A+\V+\J)\times\tta \p \overline{\curl \Phi} \dv 
 	+ \int_{\Gamma}\Big[
					(\tta^t\dfrac{1}{\mu}\curl\A)(\bnu^t\overline{\curl\Phi}) - (\tta^t\bnu)\dfrac{1}{\mu}\curl\A\p\overline{\curl\Phi}
					\Big]_\pm \ds \nonumber.
\end{align}
Remark also that, taking into account the identities Eq.~\eqref{ida},\eqref{idb}, the domaine integral \eqref{ddd1} remains: 
\begin{align*}
\int_{\Omega}\dfrac{1}{\mu}& \bigg(\curl(D\tta^t\A+D\A\tta) + \big(D\tta + D\tta^t  -\dvg\tta I \big)\curl\A \bigg)\p \overline{\curl\Phi} \dv
\\&=\int_{\Omega}   \dfrac{1}{\mu} \nabla\big(\tta^t\curl\A\big)\p \overline{\curl\Phi} \dv + \int_{\Omega}\sigma\big(i\omega\A+\V +\J \big)\times\tta \p \overline{\curl\Phi} \dv.  
\end{align*}

Furthermore, using the identity $\dfrac{1}{\mu}\nabla\big(\tta^t\curl\A\big)\p\overline{\curl\Phi}=\dvg\Big(\dfrac{1}{\mu} \big(\tta^t\curl\A\big)\p\overline{\curl\Phi}\Big)$, and the divergence theorem, we obtain 
 \begin{align*}
\int_{\Omega}   \dfrac{1}{\mu}& \nabla\big(\tta^t\curl\A\big)\p \overline{\curl\Phi} \dv 
\\&=\int_{\Omega}   \dvg\Big(\dfrac{1}{\mu} \big(\tta^t\curl\A\big)\p\overline{\curl\Phi}\Big) \dv,
\\&=-\int_{\Gamma}\Big[\dfrac{1}{\mu}(\tta^t\curl\A)\, (\bnu^t\overline{\curl\Phi})\Big]_{\pm} \ds
\\&= - \int_{\Gamma} \Big [ \dfrac{1}{\mu}   (\tta^t\bnu) (\curl\A^t\bnu) + \ttaS^t\curlS\A  \big)\Big ]_{\pm}\,  \overline{\curl\Phi} \ds.
\\&= - \int_{\Gamma} \Big [ \dfrac{1}{\mu}   (\tta^t\bnu) (\curl\A^t\bnu)\big)\Big ]_{\pm}\, \overline{\curl\Phi}\ds 
        -\int_{\Gamma} \Big [ \dfrac{1}{\mu}   \big(\ttaS^t\curlS\A \Big ]_{\pm}\, \overline{\curl\Phi}\ds.
\\&= - \int_{\Gamma}   (\tta^t\bnu)\Big [ \dfrac{1}{\mu} \bnu^t\curl\A \Big ]_{\pm}\, \overline{\curl\Phi}\ds 
        -\int_{\Gamma}   \ttaS^t\Big [ \dfrac{1}{\mu}\curlS\A \Big ]_{\pm}\, \overline{\curl\Phi}\ds.
\\&= - \int_{\Gamma}   (\tta^t\bnu)\Big [ \dfrac{1}{\mu} \bnu^t\curl\A \Big ]_{\pm}\, \overline{\curl\Phi}\ds. 
\end{align*}
where we have used the fact that 
\begin{align*}
\tta^t\curl\A &= \Big((\tta^t\bnu)\bnu + \bnu\times(\tta\times\bnu) \Big)^t \Big((\curl\A^t\bnu)\bnu + \bnu\times(\curl\A\times\bnu\Big),
                  \\&=  (\tta^t\bnu) (\curl\A^t\bnu) + \ttaS^t\curlS\A.
\end{align*}
	Therefore the eulerian derivative $\mathcal{C}_\curl' (\Omega,\tta)$ of the shape function $\mathcal{C}_\curl (\Omega)$ remains
	\begin{equation}
\mathcal{C}_\curl' (\Omega,\tta)=   \int_{\Omega}\dfrac{1}{\mu} \curl\A'\p \overline{\curl\Phi} \dv
 -\int_{\Gamma}  (\tta^t\bnu)\Big[\dfrac{1}{\mu}\curl\A\p\overline{\curl\Phi} \Big]_\pm \ds \nonumber.
\end{equation}
Finally, considering $\frac{1}{\mu}\curl\A\p\overline{\curl\Phi}=\mu\frac{1}{\mu}\curlS\A\p\frac{1}{\mu}\overline{\curlS\Phi} + \dfrac{1}{\mu}\bnu^t\curl\A\p\overline{\bnu^t\curl\Phi}$ we obtain the shape derivative $\mathcal{C}_\curl' (\Omega,\tta)$ expressed as 
\begin{framed}
\begin{equation}\begin{array}{lll}
\displaystyle\mathcal{C}_\curl' (\Omega,\tta)=&   \displaystyle\int_{\Omega}\dfrac{1}{\mu} \curl\A'\p \overline{\curl\Phi} \dv\vspace{.01in}
\\& -\displaystyle\int_{\Gamma}  (\tta^t\bnu)\Big[\dfrac{1}{\mu}\Big]_\pm \bnu^t\curl\A\p\overline{\bnu^t\curl\Phi}\ds\vspace{.01in}
\\&  -\displaystyle\int_{\Gamma}  \Big[\mu\Big]_\pm \frac{1}{\mu}\curlS\A\p\frac{1}{\mu}\overline{\curlS\Phi} \ds.
\end{array}\end{equation}
\end{framed}
\item $\mathcal{C}_\dvg(\Omega^s)$ writes 
\begin{eqnarray*}
\mathcal{C}_\dvg(\Omega^s) &=&\dfrac{1}{\tilde\mu} \int_{\Om}\dvg\A(\Omega^s)\circ T_{s}\p \overline{(\dvg\Phi)\circ T_{s}} j_{s} \dv
\\&=&\dfrac{1}{\tilde\mu}\int_{\Gamma}\dfrac{1}{j_{s}}\dvg(j_s DT_{s}^{-1}\A(\Omega^s)\circ T_{s}) \overline{\dvg(j_s DT_{s}^{-1}\Phi\circ T_{s})} \dv.
\end{eqnarray*}
Doing similar calculus as before, $\mathcal{C}_\dvg'(\Omega,T_{s})=\lim_{\s}\dfrac{1}{s}(\mathcal{C}_\dvg(\Omega^s)-\mathcal{C}_\dvg(\Omega))$ writes 
\begin{eqnarray*}
\mathcal{C}_\dvg'(\Omega,\tta)&=&-\dfrac{1}{\tilde\mu}\int_{\Omega} \dvg(\tta)\dvg\A \overline{\dvg\Phi} \dv+\dfrac{1}{\tilde\mu}\int_{\Omega} \dvg(\dvg(\tta)\A-D\tta\A + D\A\tta+\A') \overline{\dvg\Phi}\dv
\\&=&\dfrac{1}{\tilde\mu}\int_{\Omega} \dvg\big(\curl(\A\times\tta)\big) \overline{\dvg\Phi}\dv+\dfrac{1}{\tilde\mu}\int_{\Omega}\dvg\A' \overline{\dvg\Phi}\dv.
\end{eqnarray*}
	Finally we have 
\begin{framed}
\begin{align}
\mathcal{C}_\dvg'(\Omega,\tta)=\dfrac{1}{\tilde\mu}\int_{\Omega}\dvg\A' \overline{\dvg\Phi}\dv.
\end{align}
\end{framed}
\item $\displaystyle\mathcal{C}_\text{mix}(\Omega^s)$ writes 
\begin{eqnarray*}
\displaystyle\mathcal{C}_\text{mix}(\Omega^s_c)
&=&\dfrac{1}{i\omega}\int_{\Om_c^s}\sigma\big( i\omega\A(\Om_c^s)+\V(\Om_c^s))\p\big(i\omega\overline{\Phi}+\overline{\grad\varphi}\big) \dv,
\\&=&\dfrac{1}{i\omega}\int_{\Om_c} \Big(\sigma\big( i\omega\A(\Om_c^s)+\V(\Om_c^s)\big) \circ T_s \Big)\p \Big(\big(i\omega\overline{\Phi}+\overline{\grad\varphi}\big) \circ T_s \Big)\,j_s \dv,
\\&=&\dfrac{1}{i\omega}\int_{\Om_c} \Big(\sigma\big( i\omega\A(\Om_c^s)\circ T_s+\big(\V(\Om_c^s)\big)\circ T_s \Big)\p \Big(\big(i\omega\overline{\Phi}+\overline{\grad\varphi}\big) \circ T_s \Big)\,j_s \dv,
\\&=&\dfrac{1}{i\omega}\int_{\Om_c} \Big(\sigma\big( i\omega\A(\Om_c^s)\circ T_s +DT_s^{-t}\grad\big(\Vv(\Om_c^s) \circ T_s\big) \Big)\p \big(i\omega\overline{\Phi}\circ T_s+DT_s^{-t}\overline{\grad\varphi}\circ T_s\big)\,j_s \dv,
\end{eqnarray*}
Thus we can calculate the shape derivative $\mathcal{C}_\text{mix}'(\Omega_c,\tta)$, which writes:
\begin{eqnarray*}
\mathcal{C}_\text{mix}'(\Omega_c,\tta)&=&
\dfrac{1}{i\omega}\int_{\Omega_c} \dvg(\tta)\sigma\big(i\omega\A+\V\big)\p\big(i\omega\overline{\Phi}+\overline{\grad\varphi}\big) \dv
\\&&+\dfrac{1}{i\omega}\int_{\Om_c} \sigma\big( i\omega\A'+\V'\big) \p\big(i\omega\overline{\Phi}+\overline{\grad\varphi}\big) \dv
\\&&+\dfrac{1}{i\omega}\int_{\Om_c}\sigma\Big( D (i\omega\A)\tta - D\tta^t \V + \grad(\tta\p\V) \Big)\p\big(i\omega\overline{\Phi}+\overline{\grad\varphi}\big) \dv
\\&&+\dfrac{1}{i\omega}\int_{\Om_c}\sigma\big( i\omega\A + \V \big)\p\Big(i\omega\overline{D\Phi\tta}-\overline{D\tta^t\grad\varphi} + \overline{\grad(\tta\p\grad\varphi)}\Big) \dv
\end{eqnarray*}
We can easily proof for any complex valued function ${\bf V}$
\begin{align*}
\grad(\tta\p\V)&=(\grad\tta)\V + D(\V)\tta,\\&= D\tta^t\V + D(\V)\tta.
\end{align*}
	We use the above identity to obtain:
\begin{eqnarray*}
\mathcal{C}_\text{mix}'(\Omega_c,\tta)&=&
\dfrac{1}{i\omega}\int_{\Om_c} \sigma\big( i\omega\A'+\V'\big) \p\big(i\omega\overline{\Phi}+\overline{\grad\varphi}\big) \dv
\\&&+\dfrac{1}{i\omega}\int_{\Omega_c} \tta^t\grad\bigg(\sigma\big(i\omega\A+\V\big)\p\big(i\omega\overline{\Phi}+\overline{\grad\varphi}\big)\bigg) \dv
\\&&-\dfrac{1}{i\omega}\int_{\Gamma} (\tta^t\bnu)\sigma\big(i\omega\A+\V\big)\p\big(i\omega\overline{\Phi}+\overline{\grad\varphi}\big) \ds
\\&&+\dfrac{1}{i\omega}\int_{\Om_c}\sigma D\Big(i\omega\A+\V\Big)\tta\p\big(i\omega\overline{\Phi}+\overline{\grad\varphi}\big) \dv
\\&&+\dfrac{1}{i\omega}\int_{\Om_c}\sigma\big( i\omega\A + \V \big)\p D\Big(i\omega\overline{\Phi}+\overline{\grad\varphi}\Big)\tta \dv
\end{eqnarray*}
	Finally, we have 
	\begin{align*}
\mathcal{C}_\text{mix}'(\Omega_c,\tta)=\nonumber
\dfrac{1}{i\omega}\int_{\Om_c} \sigma\big( i\omega\A'+\V'\big) \p\big(i\omega\overline{\Phi}+\overline{\grad\varphi}\big) \dv
-\dfrac{1}{i\omega}\int_{\Gamma} (\tta^t\bnu) \Big[
								\sigma\big(i\omega\A+\V\big)\p\big(i\omega\overline{\Phi}+\overline{\grad\varphi}\big)
								\Big]_\pm \ds
\end{align*} 
Remark that $\big(i\omega\A+\V\big)\p\big(i\omega\overline{\Phi}+\overline{\grad\varphi}\big)$ is equivalent to $\big(i\omega\As+\Vs\big)\p\big(i\omega\overline{\PhiS}+\overline{\nabla_{\tau}\varphi}\big) + \big(i\omega\bnu^t\A+\bnu^t\V\big)\p\big(i\omega\overline{\bnu^t\Phi}+\overline{\bnu^t\grad\varphi}\big)$. Thus we conclude by the following 
\begin{framed}
\begin{align}
\mathcal{C}_\text{mix}'(\Omega_c,\tta)=&\nonumber
\dfrac{1}{i\omega}\int_{\Om_c} \sigma\big( i\omega\A'+\V'\big) \p\big(i\omega\overline{\Phi}+\overline{\grad\varphi}\big) \dv
\\&-\dfrac{1}{i\omega}\int_{\Gamma} (\tta^t\bnu) \Big[\sigma\Big]_\pm \big(i\omega\As+\Vs\big)\p\big(i\omega\overline{\PhiS}+\overline{\gradS\varphi}\big) \ds
\end{align}
 \end{framed}

\end{enumerate}
since the tangential components of $(i\omega\A+\V)$ and tests functions are continuous across the surface $\Gamma$ while the normal component; $\bnu^t(i\omega\A+\V)=0$ is vanishing on $\Gamma$.
\end{proof}


\subsubsection{The governing equation of the shape functions}

Since the shape deformation $\tta$ concerns the deposit or the flawed part, we use thus the notation $\A(\Om_d)$ and $\Vv(\Om_d)$ to emphasize the shape deformation dependance. Let $(\Phi,\varphi)$ be the test functions that belong to $\big(\mathcal{D}(\Om)\big)^3\times\mathcal{D}(\Om)$. 
Recall that the solution of the weak form 
\begin{equation}\label{varfsesqui}
\mathcal{L}\big(\A,\Vv;\Phi,\varphi\big)= \int_{\Omega}\J\p\overline{\Phi} \dv - \dfrac{1}{i\omega}\int_{\Omega_c}\J\p\overline{\nabla\varphi} \dv.
\end{equation}
is a solution of the strong problem Eqs.~\eqref{strongPbAV}. Here the sesquilinear form $\mathcal{L}$ is defined at Eq.~\eqref{sesquil}.
Remark that the sesquilinear form can be written using the shape functions defined on Lemma.~\ref{shape-derivatives}. In fact we have 
$$
\mathcal{L}\big(\A,\Vv;\Phi,\varphi\big)=\mathcal{C}_\text{curl}(\Omega) + \mathcal{C}_\dvg(\Omega) - \mathcal{C}_\text{mix}(\Omega_c).
$$
We calculate the shape derivative on both sides of the integral identities on Eq.~\eqref{varfsesqui}, we obtain:
\begin{align*}
\mathcal{C}_\text{curl}'(\Omega,\tta)  &+\mathcal{C}_\dvg'(\Omega,\tta) -\mathcal{C}_\text{mix}'(\Omega_c,\tta) =
\\& \displaystyle\int_{\Omega}\dfrac{1}{\mu} \curl\A'\p \overline{\curl\Phi} \dv	+\dfrac{1}{\tilde\mu}\int_{\Omega}\dvg\A' \dvg\Phi\,dv
-\dfrac{1}{i\omega}\int_{\Om_c} \sigma\big( i\omega\A'+\V'\big) \p\big(i\omega\overline{\Phi}+\overline{\grad\varphi}\big) \dv
\\& -\displaystyle\int_{\Omega}  (\tta^t\bnu)\Big[\dfrac{1}{\mu}\Big]_\pm \bnu^t\curl\A\p\overline{\bnu^t\curl\Phi}\ds
  -\displaystyle\int_{\Omega}  \Big[\mu\Big]_\pm \frac{1}{\mu}\curlS\A\p\frac{1}{\mu}\overline{\curlS\Phi} \ds.
\\& +\dfrac{1}{i\omega}\int_{\Gamma} (\tta^t\bnu) \Big[\sigma\Big]_\pm \big(i\omega\As+\Vs\big)\p\big(i\omega\overline{\PhiS}+\overline{\gradS\varphi}\big) \ds	    
\end{align*}
Therefore, the sesquilinear form on the shape derivative $(\A',\Vv')$ satisfies:
\begin{framed}
\begin{equation}\begin{array}{ll}\displaystyle\mathcal{L}(\A',\V';\Phi,\varphi)=
& \displaystyle\int_{\Omega}  (\tta^t\bnu)\Big[\dfrac{1}{\mu}\Big]_\pm \bnu^t\curl\A\p\overline{\bnu^t\curl\Phi}\ds
   +\displaystyle\int_{\Omega}  \Big[\mu\Big]_\pm \frac{1}{\mu}\curlS\A\p\frac{1}{\mu}\overline{\curlS\Phi} \ds.
\\& -\dfrac{1}{i\omega}\displaystyle\int_{\Gamma} (\tta^t\bnu) \Big[\sigma\Big]_\pm \big(i\omega\As+\Vs\big)\p\big(i\omega\overline{\PhiS}+\overline{\gradS\varphi}\big) \ds	    
\end{array}\end{equation}\end{framed}

\subsubsection{Impedance shape gradient}
We demonstrate in the sequel the shape gradient of the volume impedance signal measurement $\DZ_{kl}$ where we recall its formula taking into account Eq.~\eqref{defA} it follows
\begin{align*}
|\J|\DZ_{kl}(\Omega_d)=& {i\omega}\big(\dfrac{1}{\mu}-\dfrac{1}{\mu_0} \big)\int_{\Om_{d}}\curl\A_k(\Omega_d)\p{\curl\A_l^0 }\dv 
\\&+(\sigma_0-\sigma_d)\int_{\Om_d}(i\omega\A_k(\Omega_d)+\V(\Omega_d))\p{(i\omega\A_l^0+\V^0)} \dv.
\end{align*}
Thus 
\begin{align*}
|\J|\DZ_{kl}(\Omega_d^s)=& {i\omega}\big(\dfrac{1}{\mu}-\dfrac{1}{\mu_0} \big)\int_{\Om_{d}}\Big(\curl\A_k(\Omega_d)\Big)\circ T_s\p{\Big(\curl\A_l^0 \Big)\circ T_s} j_s \dv 
\\&+(\sigma_0-\sigma_d)\int_{\Om_d}\Big(i\omega\A_k(\Omega_d)\circ T_s+\big(\V(\Omega_d)\big)\circ T_s\Big)\p{\Big(i\omega\A_l^0\circ T_s+\big(\V^0\big)\circ T_s\Big)}j_s \dv.
\\=& {i\omega}\big(\dfrac{1}{\mu}-\dfrac{1}{\mu_0} \big)\int_{\Om_{d}}\dfrac{DT_s^tDT_s}{j_s}\curl\Big(DT_s^t\A_k(\Omega_d)\circ T_s\Big)\p{\curl\Big(DT_s^t\A_l^0 \circ T_s\Big)} \dv 
\\&+(\sigma_0-\sigma_d)\int_{\Om_d}\Big(i\omega\A_k(\Omega_d)\circ T_s+DT_s^{-t}\grad\big(\Vv(\Omega_d)\circ T_s\big)\Big)\p{\Big(i\omega\A_l^0\circ T_s+DT_s^{-t}\grad\big(\Vv^0\circ T_s\Big)\Big)} j_s \dv.
\end{align*}
The shape derivative is therefore written as
\begin{align*}
|\J|\DZ_{kl}'&(\Omega_d^s,\tta)=
\\& {i\omega}\big(\dfrac{1}{\mu}-\dfrac{1}{\mu_0} \big)\int_{\Om_{d}}\Big(-\dvg(\tta)I+D\tta+D\tta^t\Big)\curl\A_k(\Omega_d)\p{\curl\A_l^0} \dv 
\\&+ {i\omega}\big(\dfrac{1}{\mu}-\dfrac{1}{\mu_0} \big)\int_{\Om_{d}}\Big( \curl\A_k'(\Omega_d) + \curl\big(D\tta^t\A_k+D\A_k\tta\big) \Big)\p{\curl\A_l^0} \dv 
\\&+ {i\omega}\big(\dfrac{1}{\mu}-\dfrac{1}{\mu_0} \big)\int_{\Om_{d}} \curl\A_k(\Omega_d)\p{\curl\big(D\tta^t\A_l^0 + D\A_l^0\tta\big)} \dv
\\&+(\sigma_0-\sigma_d)\int_{\Om_d}\dvg(\tta)\Big(i\omega\A_k(\Omega_d)+\grad{\bf V}_{c,k}(\Omega_d)\Big)\p{\Big(i\omega\A_l^0 + \grad {\bf V}_{c,l}^0\Big)} \dv.
\\&+(\sigma_0-\sigma_d)\int_{\Om_d}\Big(i\omega\A_k'(\Omega_d) + i\omega D\A_k\tta -D\tta^t \grad{\bf V}_{c,k} + \grad{\bf V}'_{c,k}(\Omega_d) +\grad\big(\tta\p\grad {\bf V}_{c,k}\big)\Big)\p{\Big(i\omega\A_l^0 + \grad \Vv^0\Big)} \dv.
\\&+(\sigma_0-\sigma_d)\int_{\Om_d}\Big(i\omega\A_k(\Omega_d) + \grad{\bf V}_{c,k}(\Omega_d)\Big)\p{\Big(i\omega D\A_l^0\tta -D\tta^t\grad{\bf V}_{c,l}^0+ \grad\big(\tta\p\grad {\bf V}_{c,l}^0 \big) \Big)} \dv.
\end{align*}
It reads also 
\begin{align}
|\J|\DZ_{kl}'&(\Omega_d,\tta)=\nonumber
\\=& {i\omega}\big(\dfrac{1}{\mu}-\dfrac{1}{\mu_0} \big)\int_{\Om_{d}} \curl\A_k'(\Omega_d) \p{\curl\A_l^0} \dv \nonumber
\\&+(\sigma_0-\sigma_d)\int_{\Om_d}\Big(i\omega\A_k'(\Omega_d)+ \grad{\bf V}'_{c,k}(\Omega_d)  \Big)\p{\Big(i\omega\A_l^0 + \grad \Vv^0\Big)} \dv \nonumber
\\&+(\sigma_0-\sigma_d)\int_{\Gamma_d}\Big[(\tta^t\bnu)\Big( i\omega\A_k(\Omega_d)+\grad{\bf V}_{c,k}(\Omega_d)\big)\p{(i\omega\A_l^0 + \grad {\bf V}_{c,l}^0)}\Big)\Big]_{\pm} \dv\label{la-6}
\\&-(\sigma_0-\sigma_d)\int_{\Om_d}\tta^t\grad\Big( \big(i\omega\A_k(\Omega_d)+\grad{\bf V}_{c,k}(\Omega_d)\big)\p{(i\omega\A_l^0 + \grad {\bf V}_{c,l}^0)}\Big) \dv\label{la-5}
\\&+ {i\omega}\big(\dfrac{1}{\mu}-\dfrac{1}{\mu_0} \big)\int_{\Om_{d}}\Big(\grad\big(\curl\A_k(\Omega)\p\tta \big)  + \mu\sigma(i\omega\A_k+{\bf V}_{c,k})\times\tta\big)\Big)\p{\curl\A_l^0} \dv\label{la-4} 
\\&+ {i\omega}\big(\dfrac{1}{\mu}-\dfrac{1}{\mu_0} \big)\int_{\Om_{d}} \curl\A_k(\Omega_d)\p{\curl\big(D\tta^t\A_l^0 + D\A_l^0\tta\big)} \dv\label{la-3}
\\&+(\sigma_0-\sigma_d)\int_{\Om_d}\Big(D(i\omega\A_k(\Omega_d)+ \grad{\bf V}_{c,k}(\Omega_d) )\tta \Big)\p{\Big(i\omega\A_l^0 + \grad \Vv^0\Big)} \dv\label{la-2}
\\&+(\sigma_0-\sigma_d)\int_{\Om_d}\Big(i\omega\A_k(\Omega_d) + \grad{\bf V}_{c,k}(\Omega_d)\Big)\p{\Big( D\big(i\omega\A_l^0 + \grad {\bf V}_{c,l}^0 \big)\tta \Big)} \dv\label{la-1}.
\end{align}
Remark that the sum of the domain integrals.~\eqref{la-5},\eqref{la-2} and \eqref{la-1} vanishes, due to the gradient of a scalar product. Remark also that domain integral.~\eqref{la-4} could be written, using divergence theorem, as 
\begin{align*}
{i\omega}&\big(\dfrac{1}{\mu}-\dfrac{1}{\mu_0} \big)\int_{\Om_{d}}\Big(\grad\big(\curl\A_k(\Omega)\p\tta\big)  + \mu\sigma(i\omega\A_k+{\bf V}_{c,k})\times\tta\big)\Big)\p{\curl\A_l^0} \dv
\\=&{i\omega}\int_{\Gamma}\Big[\dfrac{1}{\mu}\big(\tta^t\curl\A_k)\,\bnu\p{\curl\A_l^0}\Big]_\pm\dv
-{i\omega}\big(\dfrac{1}{\mu}-\dfrac{1}{\mu_0} \big)\int_{\Om_{d}}\mu\sigma\Big(i\omega\A_k(\Omega_d)+{\bf V}_{c,k}\Big)\p{\curl\A_l^0\times\tta} \dv.
\end{align*}
where we used the formula $\dvg(\big(\curl\A_k(\Omega)\p\tta\p{\curl\A_l^0})=\grad\big(\tta^t\curl\A_k(\Omega)\big)\p{\curl\A_l^0}$. The first domain integral in the above line, could be written
\begin{align*}
{i\omega}\int_{\Gamma}\Big[\dfrac{1}{\mu}\big(\tta^t\curl\A_k)\,\bnu\p{\curl\A_l^0}\Big]_\pm\dv=&
{i\omega}\int_{\Gamma} (\tta^t\bnu) \Big[ \dfrac{1}{\mu}\bnu^t\curl\A_k\Big]_{\pm}\bnu^t\curl\A_l^0 \ds
\end{align*}

For the domain integral.~\eqref{la-3}, we use the fact that $\curl\big(D\tta^t\A_l^0 + D\A_l^0\tta\big)=\curl\big(\curl\A\times\tta\big)$ and after integration by part we obtain
\begin{align*}
{i\omega}&\big(\dfrac{1}{\mu}-\dfrac{1}{\mu_0} \big)\int_{\Om_{d}} \curl\A_k(\Omega_d)\p{\curl\big(D\tta^t\A_l^0 + D\A_l^0\tta\big)} \dv 
\\&= {i\omega}\big(\dfrac{1}{\mu}-\dfrac{1}{\mu_0} \big)\bigg(\int_{\Om_{d}}\curl\curl\A_k(\Omega_d)\p{\curl\A\times\tta} \dv-\int_{\Gamma}\bnu\times\curl\A_k\p{\curl\A_l^0\times\tta}\ds \bigg)
\\&= {i\omega}\big(\dfrac{1}{\mu}-\dfrac{1}{\mu_0} \big) \int_{\Om_{d}}\mu\sigma\Big(i\omega\A_k(\Omega_d)+{\bf V}_{c,k}\Big)\p{\curl\A_l^0\times\tta} \dv
\\&\quad -{i\omega}\big(\dfrac{1}{\mu}-\dfrac{1}{\mu_0} \big)  \int_{\Gamma} (\bnu^t\overline{\curl\A^0_l}) (\tta^t\curl\A_k) \ds 
	      + {i\omega}\big(\dfrac{1}{\mu}-\dfrac{1}{\mu_0} \big) \int_{\Gamma}(\tta^t\bnu) \curl\A_k\p\overline{\curl\A^0_l}.
\end{align*}
Remak that the jump on the domain integral. ~\eqref{la-6} could be written using the tangential and the normal component of both $i\omega\A_k+{\bf V}_{c,k}$ and $i\omega\A_k^0+{\bf V}_{c,l}^0$, which remains taking into account the vanishing normal component of both (see~\eqref{strongPbAV}$_3$), as 
\begin{align*}
\int_{\Gamma_d}&\Big[\sigma\big(\tta^t\bnu\big)\Big( \big(i\omega\A_k(\Omega_d)+\grad{\bf V}_{c,k}(\Omega_d)\big)\p{(i\omega\A_l^0 + \grad {\bf V}_{c,l}^0)}\Big)\Big]_{\pm} \dv
\\&=\int_{\Gamma_d}\Big[\sigma\Big]_{\pm}\big(\tta^t\bnu\big)\Big( \big(i\omega\A_{k,\tau}(\Omega_d)+\grad{\bf V}_{c,k,\tau}(\Omega_d)\big)\p{(i\omega\A_{l,\tau}^0 + \grad {\bf V}_{c,l,\Tau}^0)}\Big) \dv.
\end{align*}

After incorporating all the above external calculus in the impedance shape derivative $\DZ'(\Omega,\tta)$ we obtain:
\begin{corollary}
The impedance shape derivative has the following form
\begin{framed}
\begin{equation}\begin{array}{lll}
|\J|\DZ_{kl}'(\Omega_d,\tta)=&{i\omega}\big(\dfrac{1}{\mu}-\dfrac{1}{\mu_0} \big)\displaystyle\int_{\Om_{d}} \curl\A_k'(\Omega_d) \p{\curl\A_l^0} \dv\vspace{.1in}
\\&+(\sigma_0-\sigma_d)\displaystyle\int_{\Om_d}\Big(i\omega\A_k'(\Omega_d)+ \grad{\bf V}'_{c,k}(\Omega_d)  \Big)\p{\Big(i\omega\A_l^0 + \grad \Vv^0\Big)} \dv\vspace{.1in}
\\&  \displaystyle\int_{\Gamma}  (\tta^t\bnu)\Big[\dfrac{1}{\mu}\Big]_\pm \bnu^t\curl\A_k\p\overline{\bnu^t\curl\A_l^0}\ds
   +\displaystyle\int_{\Gamma}  \Big[\mu\Big]_\pm \frac{1}{\mu}\curlS\A_k\p\frac{1}{\mu}\overline{\curlS\A_l^0} \ds.\vspace{.1in}
\\& -\dfrac{1}{i\omega}\displaystyle\int_{\Gamma} (\tta^t\bnu) \Big[\sigma\Big]_\pm \big(i\omega\A_{k,\tau}+\gradS\Vv\big)\p\big(i\omega\overline{\A_{l,\tau}^0}+\overline{\gradS \Vv^0}\big) \ds	    
\end{array}\end{equation}
\end{framed}
\end{corollary}

\subsubsection{The governing equation of the adjoint state}
Let us define the magnetic vector potential $P$ and the scalar electric potential $W_c$ as the unique solution of the following:
\begin{equation}
\mathcal{L}^*(P,W;\Phi,\Vv) =
\big[\dfrac{i\omega}{\mu}\big]_\pm \displaystyle\int_{\Om_d}\curl\A'(\Omega_d)\p\curl\A_l^0 \dv
- \big[\sigma\big] \displaystyle\int_{\Om_d}\Big(i\omega\A_k'(\Omega_d)+ \grad{\bf V}'_{c,k}(\Omega_d)  \Big)\p{\Big(i\omega\A_l^0 + \grad \Vv^0\Big)} \dv
\end{equation}
where $\mathcal{L}^*$ is the conjugate of the sesquilinear form $\mathcal{L}$ defined at Eq.~\eqref{sesquil}.
we have thus
\begin{eqnarray*}
|\J|\DZ_{kl}'(\Omega_d,\tta) &=& \mathcal{L}^*(P,W;\A'_k,{\bf V'}_{c,k}) + \mathcal{L}(\A'_k,{\bf V}'_{c,k};\A_l^0,{\bf V}_{c,l}^0).
\\ &=& \mathcal{L}(\A'_k,{\bf V'}_{c,k};\overline{P},\overline{W}) + \mathcal{L}(\A'_k,{\bf V}'_{c,k};\A_l^0,{\bf V}_{c,l}^0).
\\ &=& \mathcal{L}(\A'_k,{\bf V}'_{c,k};\A_l^0+\overline{P},{\bf V}_{c,l}^0+\overline{W}).
\end{eqnarray*}
We have therefore the  following theorem. 
\begin{theorem}

\begin{framed}
\begin{equation}\begin{array}{lll}
|\J|\DZ_{kl}'(\Omega_d,\tta)&=
 \displaystyle\int_{\Gamma}  (\tta^t\bnu)\Big[\dfrac{1}{\mu}\Big]_\pm \bnu^t\curl\A_k\p\overline{\bnu^t\curl\A_l^0+\overline{P}}\ds\vspace{.1in}
\\& \displaystyle\int_{\Gamma}  \Big[\mu\Big]_\pm \frac{1}{\mu}\curlS\A_k\p\frac{1}{\mu}\overline{\curlS\A_l^0+\curlS\overline{P}} \ds\vspace{.1in}
\\& -\displaystyle\dfrac{1}{i\omega}\int_{\Gamma}\Big[\sigma\Big]_{\pm}(\tta^t\bnu)(i\omega\A_{k,\tau}+{\bf V}_{c,k,\tau})\p \Big(i\omega\overline{\bnu\times(\A_l^0+\overline{P})} 
 + \overline{\bnu\times(\grad{\bf V}_{c,l}^0+\grad W)}\Big) \ds\end{array}\end{equation}
\end{framed}
\end{theorem}

\subsection{Explicit formula of the shape gradient}
Let us recall the cost functional:
\begin{equation*}
  \mathcal{J}(\Omega_{D}) = \int_{z_{\min}}^{z_{\max}} |\Z(\Omega_{D};\zeta) - \Z_{mes}(\zeta)|^{2} d\zeta,
\end{equation*}
where $\Z$ is either $\Z_{FA}$ or $\Z_{F3}$ according to the measurement mode: 
\begin{align*}
  & \Z_{FA}(\tta) = \frac{\i}{2}( \DZ_{11}(\tta)+  \DZ_{21}(\tta)), \qquad
  & \Z_{F3}(\tta) = \frac{\i}{2}( \DZ_{11}(\tta)-   \DZ_{22}(\tta)), \\
  & \Z'_{FA}(\tta) = \frac{\i}{2}( \DZ'_{11}(\tta)+\DZ'_{21}(\tta)), \qquad
  & \Z'_{F3}(\tta) = \frac{\i}{2}( \DZ'_{11}(\tta)- \DZ'_{22}(\tta)).
\end{align*}
The shape derivative is stated as follows
\begin{equation}
  \mathcal{J}'(\Omega_{d})(\tta) = 2\int_{\Omega} \Re\bigg(\Z'(\Omega_d,\tta)\p\overline{\Z(\Omega_{D};\zeta) - \Z_{mes}(\zeta)}\bigg) \dv
\end{equation}
\begin{eqnarray}
  \mathcal{J}'(\Omega_{d})(\tta) = -\frac{\omega}{I^{2}} \int_{\Gamma_{0}}(\bnu\cdot\tta)g \ds,
\end{eqnarray}
where $g = g_{11}+g_{21}$ in the absolute mode or $g = g_{11}-g_{22}$ in the differential mode, with
\begin{eqnarray*}
  g_{kl} = \int_{z_{\min}}^{z_{\max}} & \Re \bigg( (\overline{Z(\Omega_{D};\zeta) - Z_{mes}(\zeta)}) \bigg\{
    \left[\frac{1}{\mu}\right]_{\pm}(\bnu\cdot\curl A_{k})(\bnu\cdot\overline{P_{l}}-\bnu\cdot\curl A^{0}_{l}) \notag \\
  & - [\mu]_{\pm} \left(\bnu\times(\frac{1}{\mu}\curl A_{k}\times \bnu)\right)
    \cdot \left(\bnu\times(\frac{1}{\mu}\curl\overline{P_{l}}\times \bnu) - \bnu\times(\frac{1}{\mu^{0}}\curl A^{0}_{l}\times \bnu)\right) \notag \\
  & + \frac{1}{\i\omega}[\sigma]_{\pm}(\i\omega A_{k\tau}+\grad_{\tau}V_{k})
    \cdot(\overline{\i\omega P_{l\tau}+\grad_{\tau}W_{l}} + \i\omega A^{0}_{l\tau}+\grad_{\tau}V^{0}_{l})
    \bigg\}\bigg) d \zeta.
\end{eqnarray*}
We choose the shape pertubation $\tta$ such that
\begin{eqnarray*}
  \tta = g \bnu \quad \text{on }\Gamma,
\end{eqnarray*}
which is a minimizing direction since
\begin{eqnarray*}
  \mathcal{J}'(\Omega_{d,0})(\tta) = -\frac{\omega}{I^{2}} \int_{\Gamma_{0}} |g|^{2} d s \leq 0.
\end{eqnarray*}

\subsection{Parametrized inversion}
In the cylindrical coordinates $(\vec l,\vec \rho,\vec z)$, assume $\vec \rho$ is fixed (e.g. $\rho=2\pi$ axisymmetric case). We want to retrieve the shape of a deposit which parameters can be expressed
 on the coordinate $(\vec l,\vec z)$. Since the bound of the variable $z$ can be approximated with a signal processing techniques, we are therefore considered with just one variable .i.e. $l$. In this case, the outward normal can be expressed as $\bnu=\epsilon \vec{l}$


\section{Numerical tests and experiments}

	We present and explain in the sequel some particular techniques to achieve performance of the direct solver, and consequently the inverse solver. 
	
	Probing of deposits is an operation of scan with two probes introduced along the tube. The probes act as cameras and detect abnormal variations of the signal measurements at the presence of a conductor default. Obviously the measurement signal will be confronted with a healthy signal in order to detect change, hence information about deposit. In the mathematical point of view the probes are modeled with coils that generate solenoidal magnetic field that constitute the source term in the equation. Typically at each position of the coils we have a new solution related to the new source term. If we consider a new mesh related to the new position of coils, we are obliged to assemble new matrices and solve new systems, which are typically very huge in terms of memory occupation. It is therefore a good programming practice to create a unique mesh that incorporates all movement of coils along the tube. Hence one only needs to modify the right hand side of the system. Since we use sparse-direct parallel solver we are in good position to factorize the main matrix only once and assemble right hand sides at each new coil position. 
  \begin{figure}[!htbp]
  \begin{minipage}[c]{0.45\linewidth}
   \centering
\includegraphics[width=4cm,height=5cm]{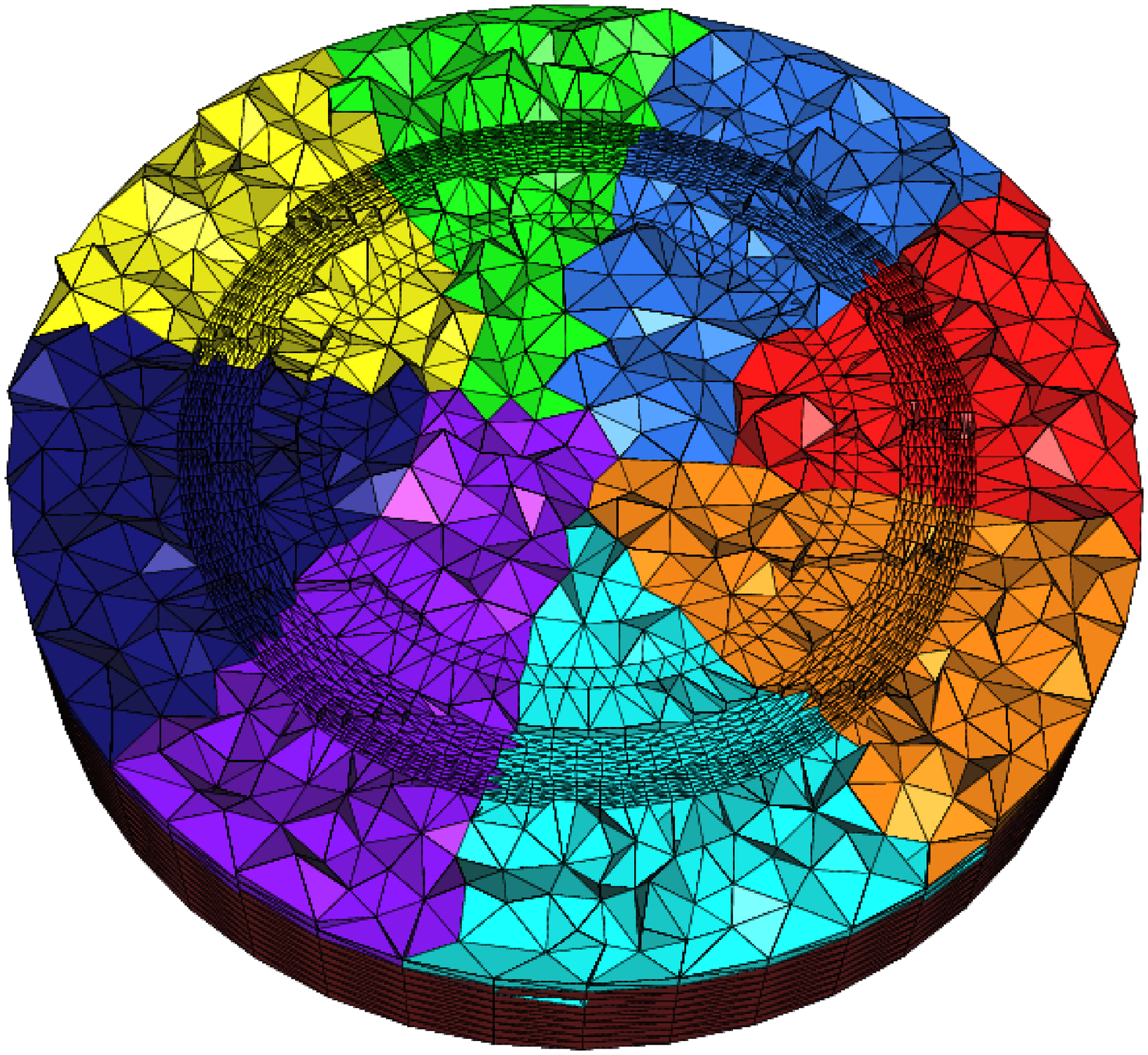}
  \end{minipage}
  \begin{minipage}[c]{0.45\linewidth}
   \centering
\includegraphics[width=4cm,height=5cm]{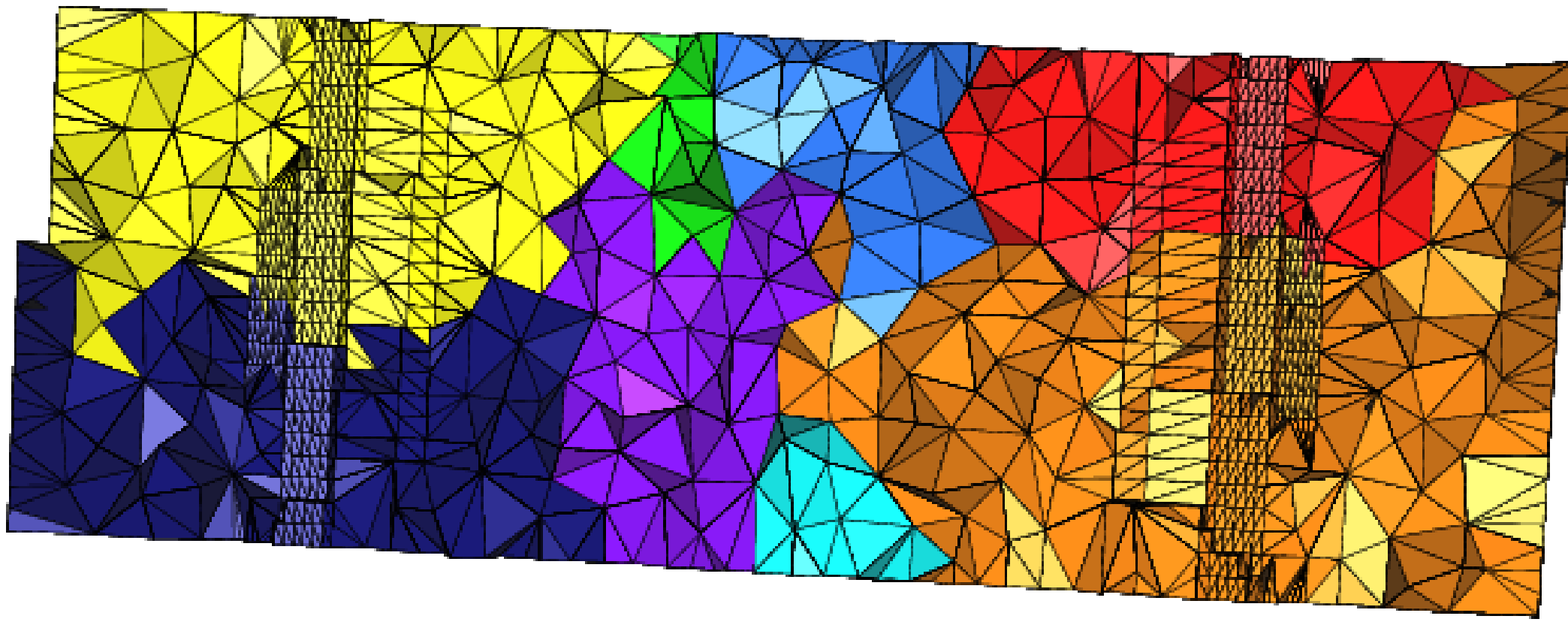}
  \end{minipage}
  \begin{minipage}[c]{0.45\linewidth}
   \centering
\includegraphics[width=4cm,height=5cm]{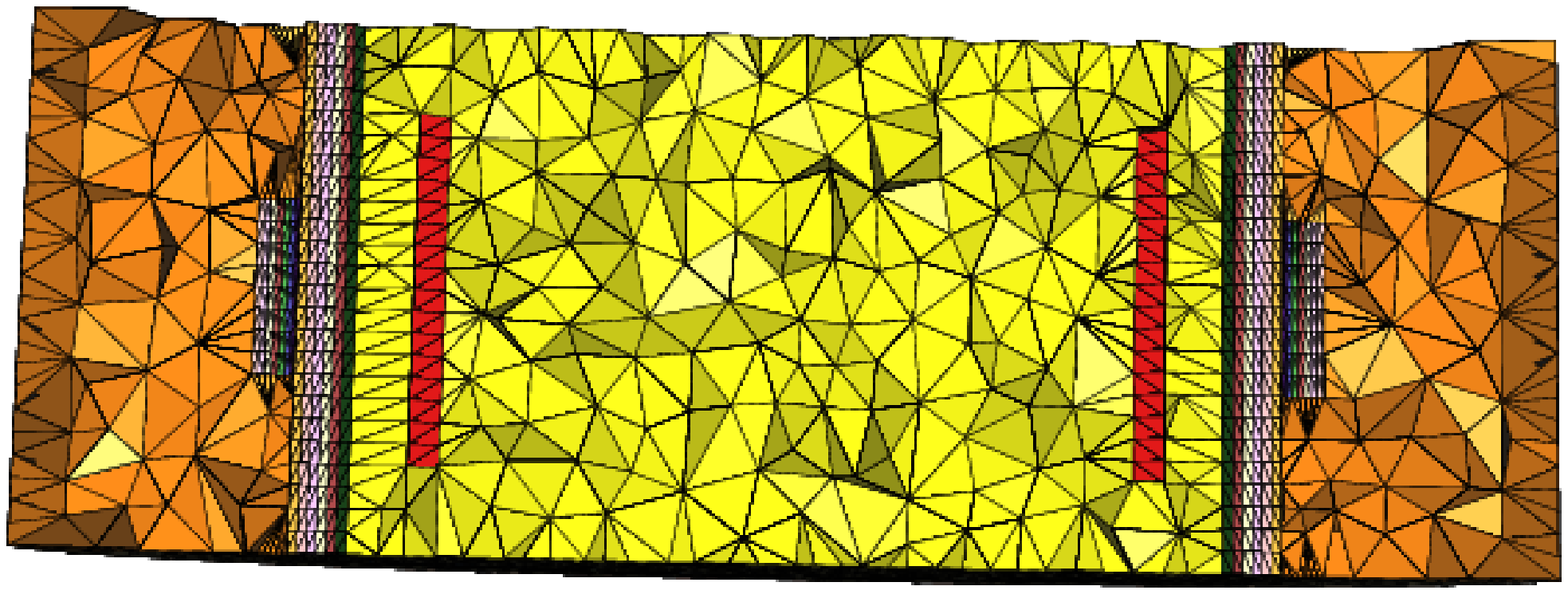}
  \end{minipage}
  \begin{minipage}[c]{0.45\linewidth}
   \centering
\includegraphics[width=4cm,height=5cm]{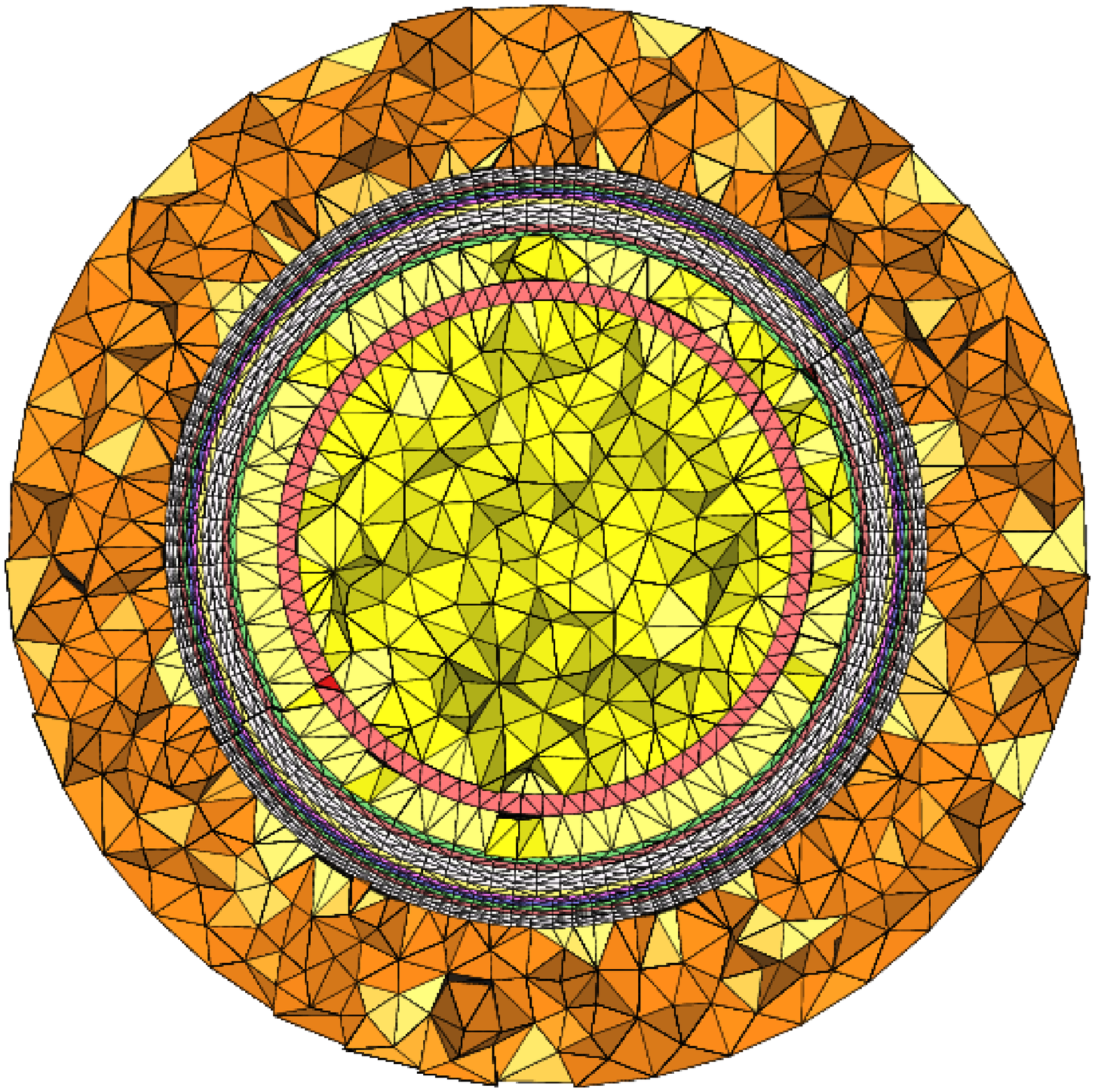}
  \end{minipage}
  \caption{Partition of the 3D mesh in order to deal with parallel matrix assembly. Identification of different regions of the heterogeneous $\sigma$ and $\mu$ (on top) and new label for element according to the new partition of the mesh.}
  \label{meshpartition}
\end{figure}	
	The variational formulation, namely the sesquilinear form defined at Eq.~\eqref{sesquil}, is discretized in order to assemble the finite element matrices. Assuming the above techniques, we have thus two global matrices to deal with during the probing process; one that takes into account the presence of the deposit and the second that disregards it and considers the vacuum. Even the scan process is reduced in term of computational time, because the time consuming of the factorization, it is important to propose parallelization across the assembly in order to accelerate the resolution. However, particular attention must be taken when the problem is non-homogeneous, in the sense of the change of the conductivities and the permeability in the domain. It's good practice to declare those variable (i.e. $\sigma$ and $\mu$ as P0-Lagrange finite elements variables) before the assembly of matrices. This task is done using one graph partitioner e.g. scotch~\cite{Pellegrini01scotchand} or metis~\cite{Karypis95metis} . After partitioning the mesh, elements change their labels as the ranks of the used group of processors. It is therefore accurate to define the P0-Lagrange non-homogeneous domain variable on the non-partitioned mesh and then include them in the variational formulation that admits the partitioning. We present in Fig.~\ref{meshpartition} the same 3d mesh before and after the partition. The top Figures provides partition into 8 sub-regions of the 3d-mesh using scotch graph partitioner, whereas the Figure on the bottom presents the initial mesh where regions identify the insulator, tube and deposit parts. 
	
\begin{quotation}
In order to limit the cash-memory usage, the validation of the direct solver considers some reductions of the computational domain. 
	 We therefore calibrate as it happens the electromagnetic parameters: the permeability and the conductivity of the deposit, while almost test cases keep the conductivity and permeability of Tube as described in Tab.~\ref{industparams}. Thus, we can uses coarse triangulation of the computational domain that include the vacuum and deposits. 
\end{quotation}

  \begin{figure}[!htbp] 
  \begin{minipage}[c] {0.45\linewidth}
   \centering
\includegraphics[width=6cm,height=4.6cm] {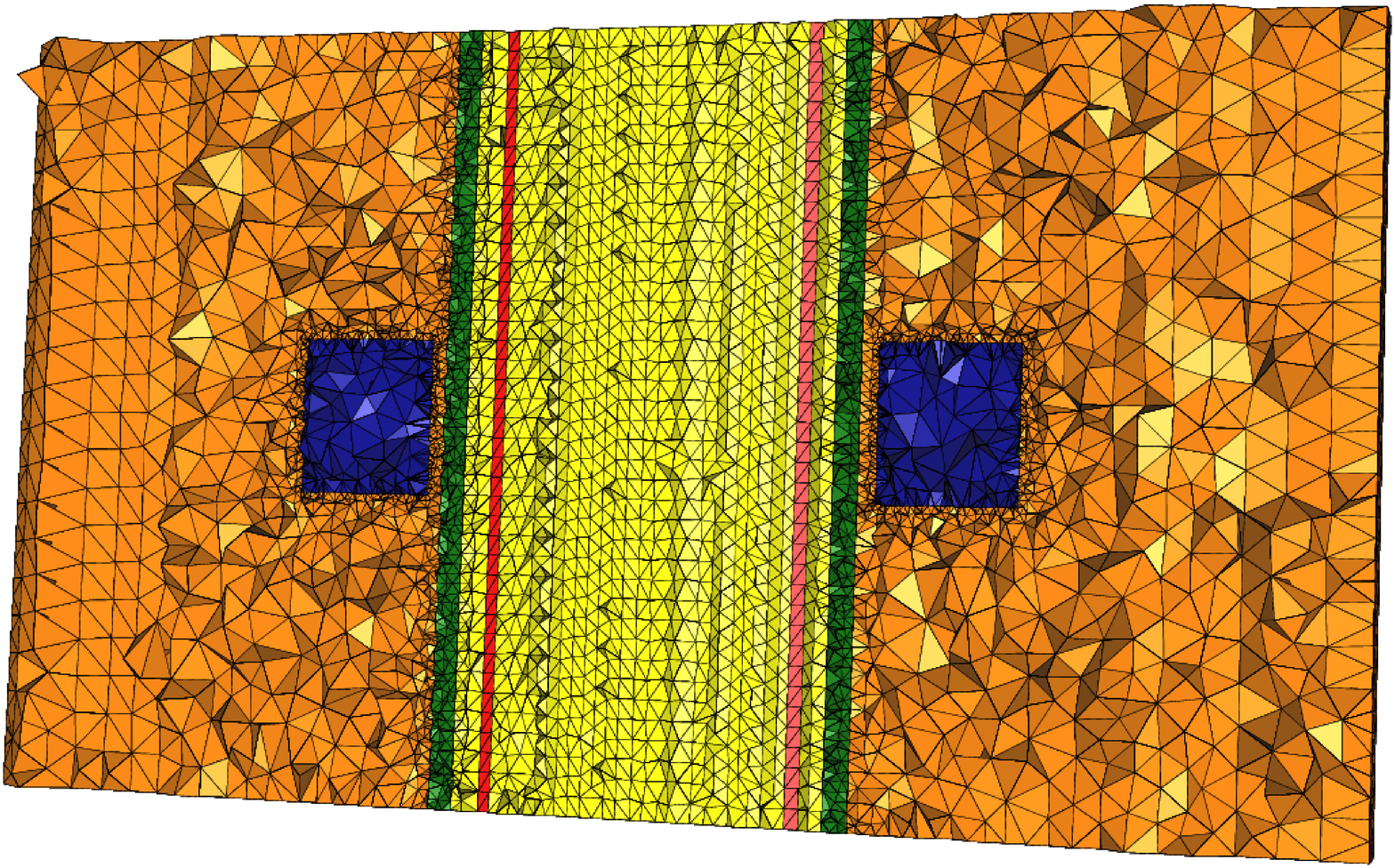}
  \end{minipage}
  \begin{minipage}[c] {0.45\linewidth}
   \centering
\includegraphics[width=3cm,height=4cm] {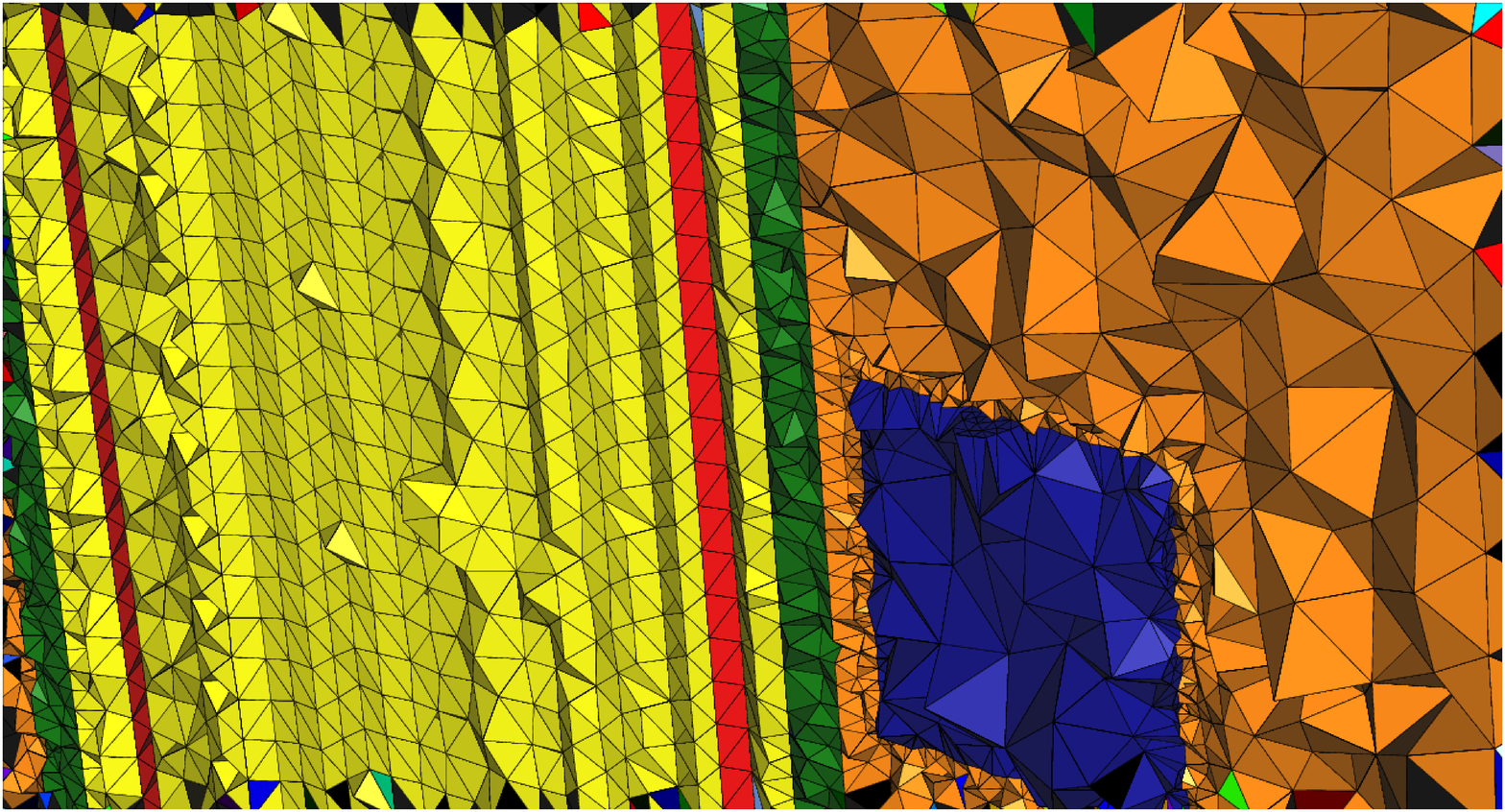}
  \end{minipage}
  \begin{minipage}[c] {0.45\linewidth}
   \centering
\includegraphics[width=6cm,height=4.6cm] {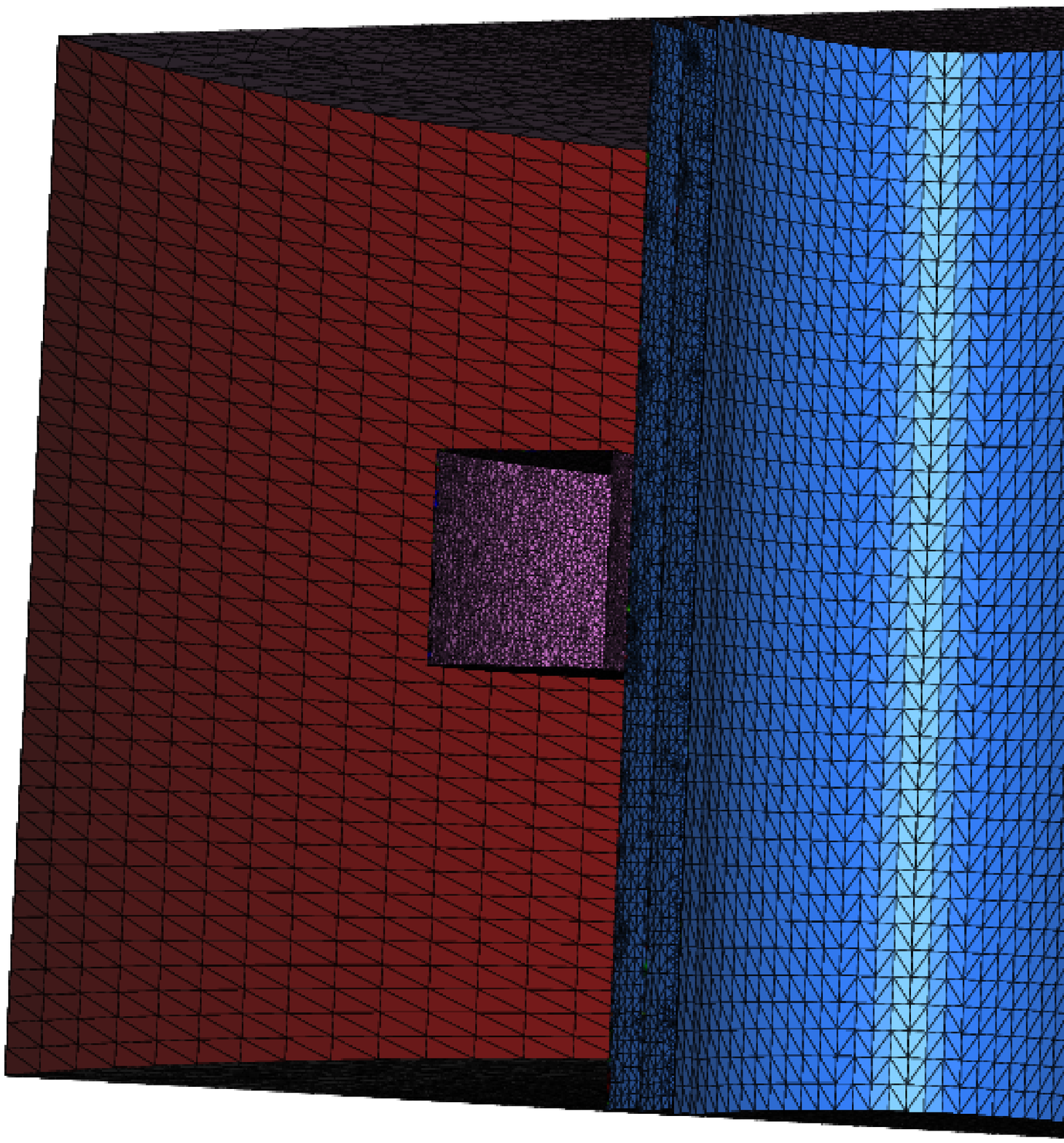}
  \end{minipage}
  \begin{minipage}[c] {0.45\linewidth}
   \centering
\includegraphics[width=3cm,height=4cm] {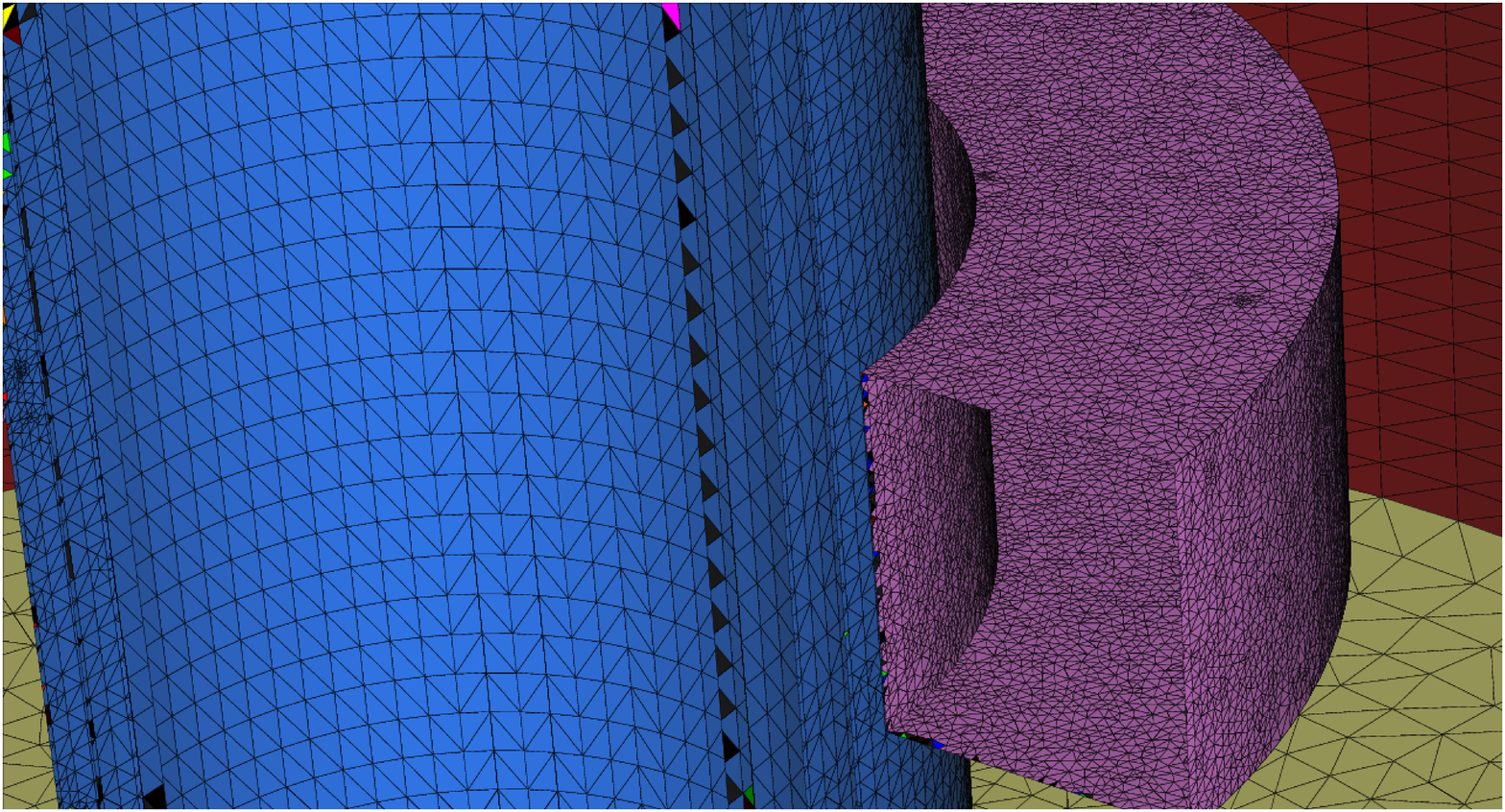}
  \end{minipage}
  \caption{Anisotropic local surface refinement for a high conductive axisymmetric default}
  \label{surfraffine}
\end{figure}	
In the case of a high conductivity of the deposit, we have to refine the mesh locally over this deposit. Anisotropic metric is therefore needed to ensure local refinement. Furthermore if the deposit skin depth is located around the boundary surface of the deposit, we no longer need to refine all the volume of the deposit, but we just need surface refinement. 
	Technically, we implement anisotropic surface 3D mesh refinement using P1-finite element function defined around the boundary of the deposit. That function is therefore used as metric for the mesh generator .e.g "tetgen", which is incorporated in FreeFem++. Fig.~\ref{surfraffine} shows the resulting 3D mesh after the local surface refinement. 
	

\subsection{3D vs 2D validation of the direct solver (the axisymmetric case)}  
Theoretically the axisymmetric solution is independent of the problem dimension. In fact, in the vectorial form of maxwell equation $\E\equiv(\E_x,\E_y,\E_z)^T$ the projection of the solution $\E$ on the plan .e.g $\vec y,\vec z$ gives only $\E_x$, which is the axisymmetric solution. 
 \begin{figure}[htpb] 
  \begin{minipage}[c] {0.48\linewidth}
   \centering
\includegraphics[origin=c,width=6cm,height=6cm] {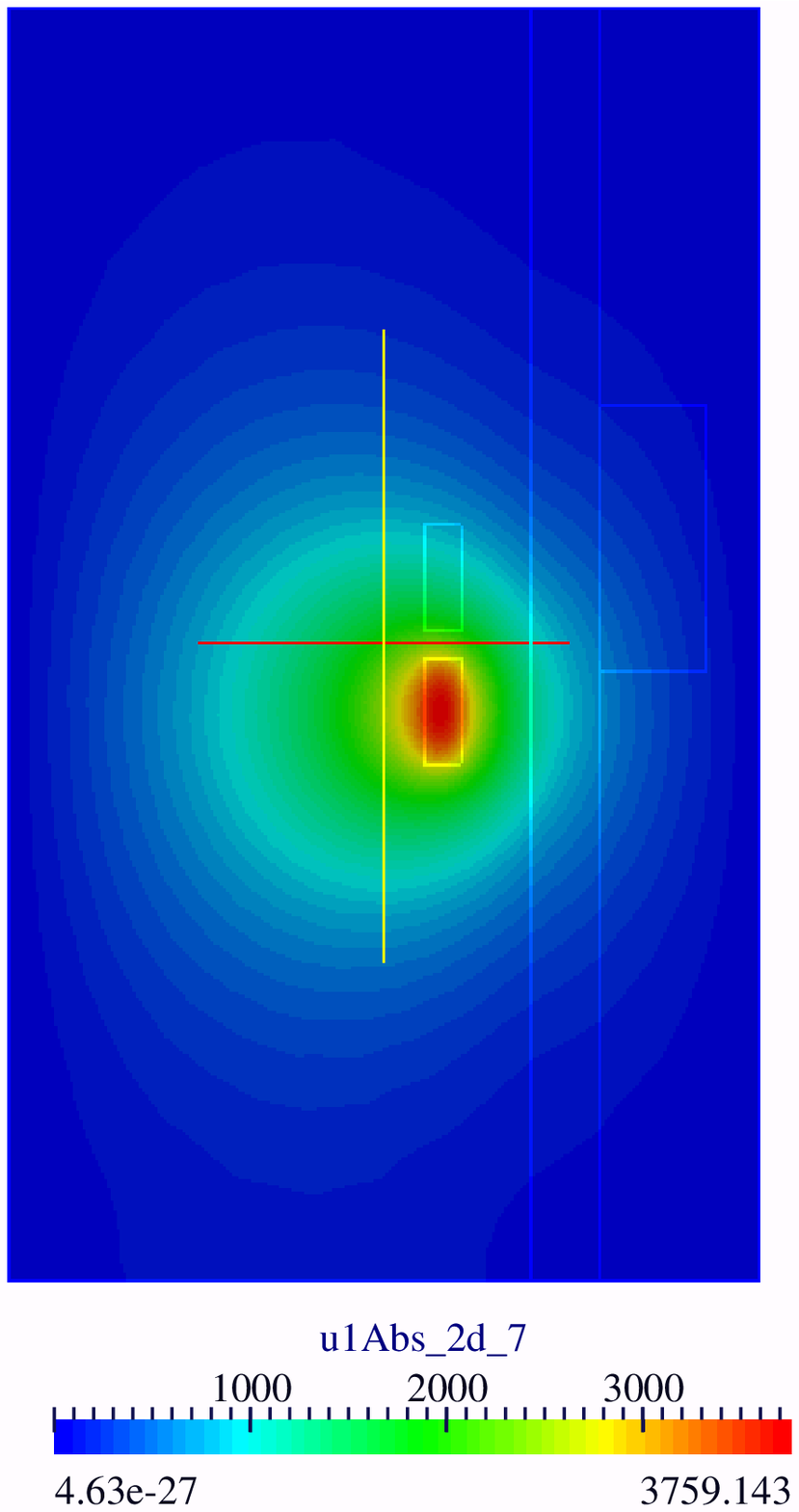}
  \end{minipage}
  \begin{minipage}[c] {0.48\linewidth}
   \centering
\includegraphics[origin=c,width=6cm,height=7cm] {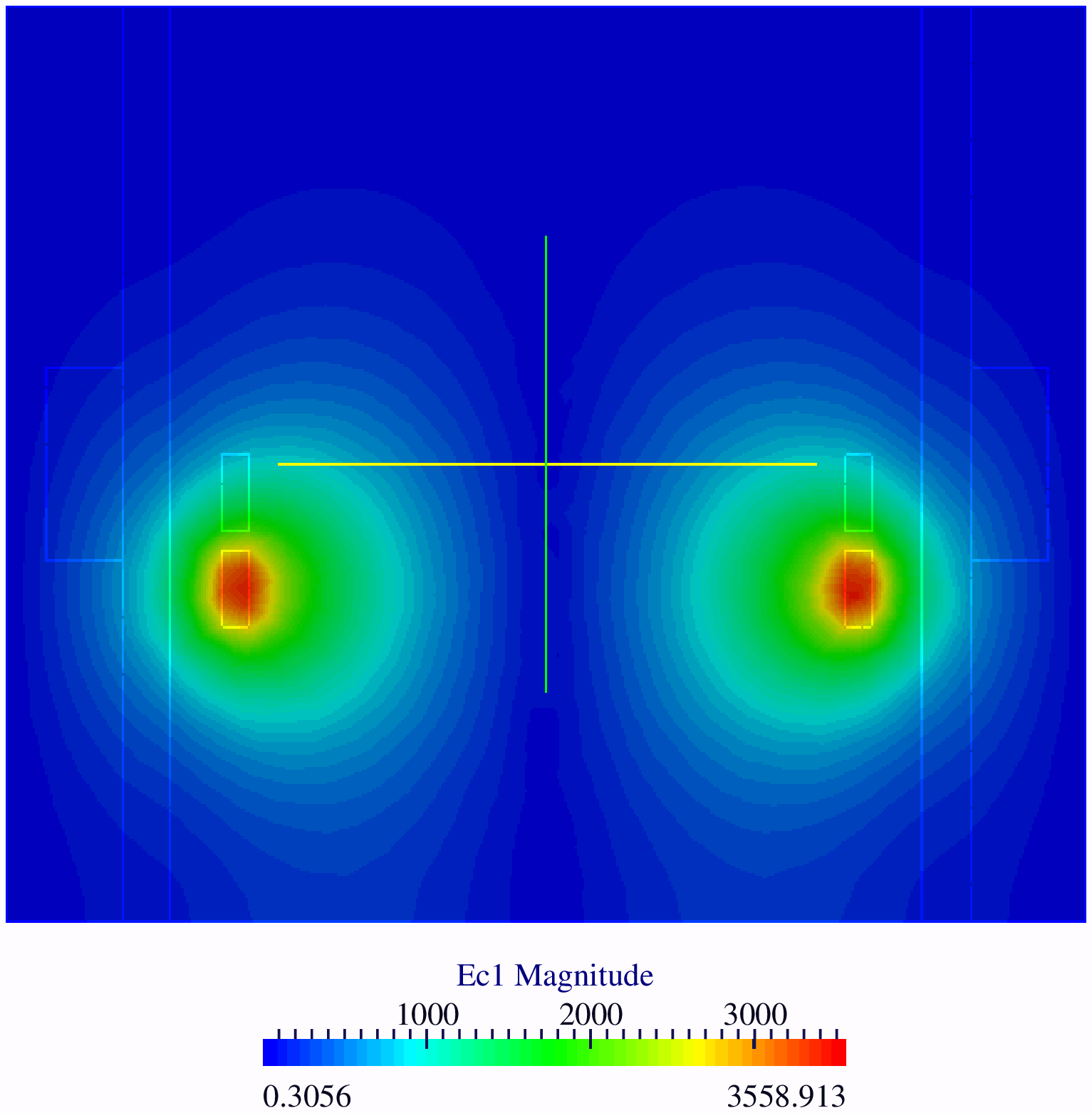}
  \end{minipage}
  \caption{Propagation of the axisymmetric solution (magnetic field $\E$) in the computational domain 2D (left) and 3D (magnetic field $\E=i\omega\A+\V$) (right). Calculus corresponds to $\sigma_d=1e3$.}\label{crownFirstTEST}
\end{figure}
	
	This paragraph concerns the 3D-extension of the axisymmetric simulation~\cite{ectod2d}. We give a comparative results in order to validate the 3D direct solver. 
	It is worth noting that, in order to have a good approximation of the electric fields $\E$ in vacuum we consider the closer part of the deposit as a conductor part, where we put $\sigma_{effictive}=0.1$. 
	In this test case, we consider a crown surrounding the tube as a deposit, Fig.\ref{crownFirstTEST} shows the cross section cut for the whole geometry. As we are concerned with a perfect continuity of the electromagnetic field, we thus consider the tangential component of the electric field vanishing i.e. we put $\E\times\bnu=0$ at the boundary. The comparison results are presented in Fig.~\ref{compareComlexplanValidation} with a complex plan representation of the absolute signal mode $FA$ as well as the differential signal mode $F3$.
	
\begin{figure}[htpb] 
\begin{tabular}{cc}
\includegraphics[angle=-90,origin=c,width=6cm,height=6cm] {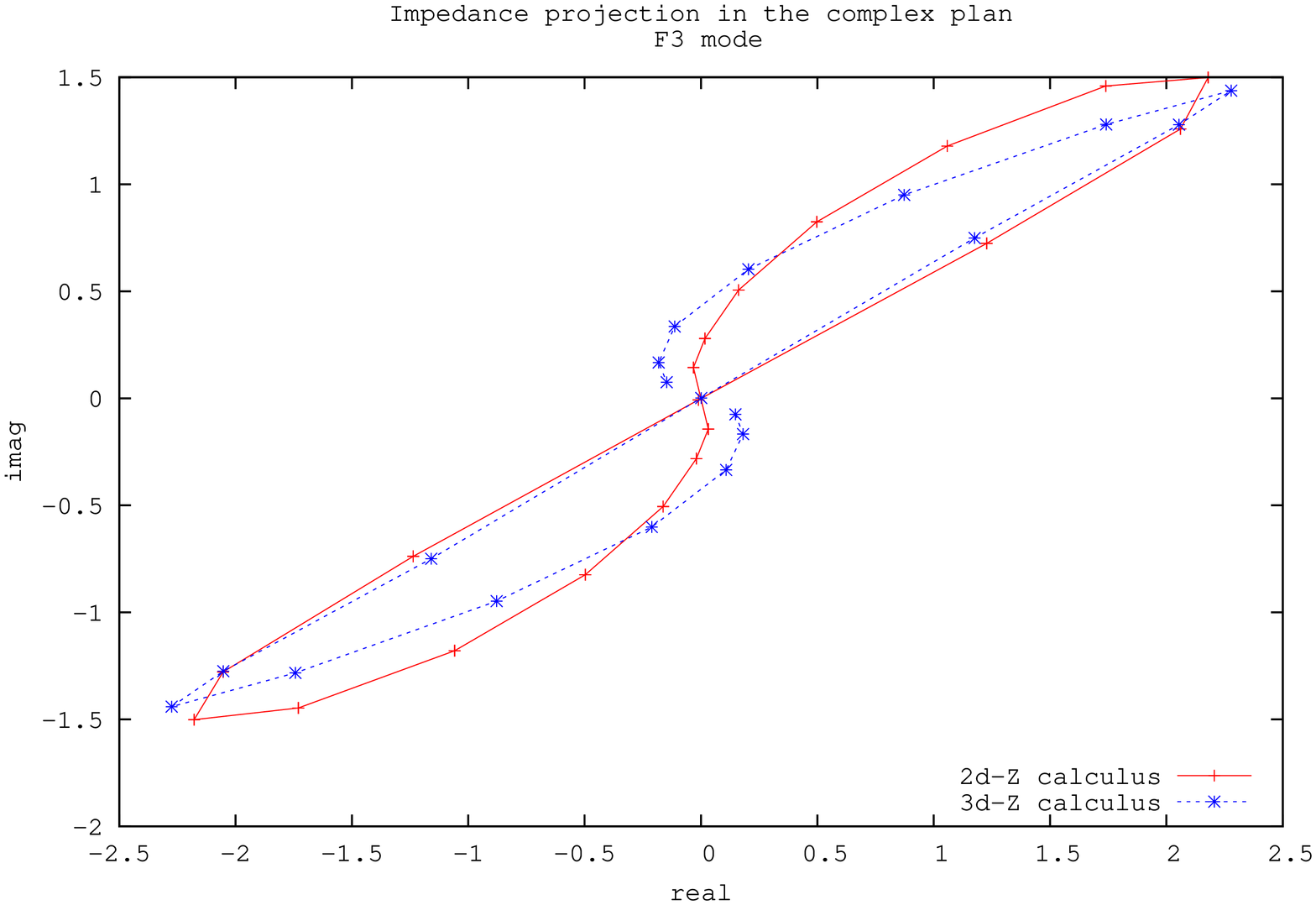} & 
\includegraphics[angle=-90,origin=c,width=6cm,height=6cm] {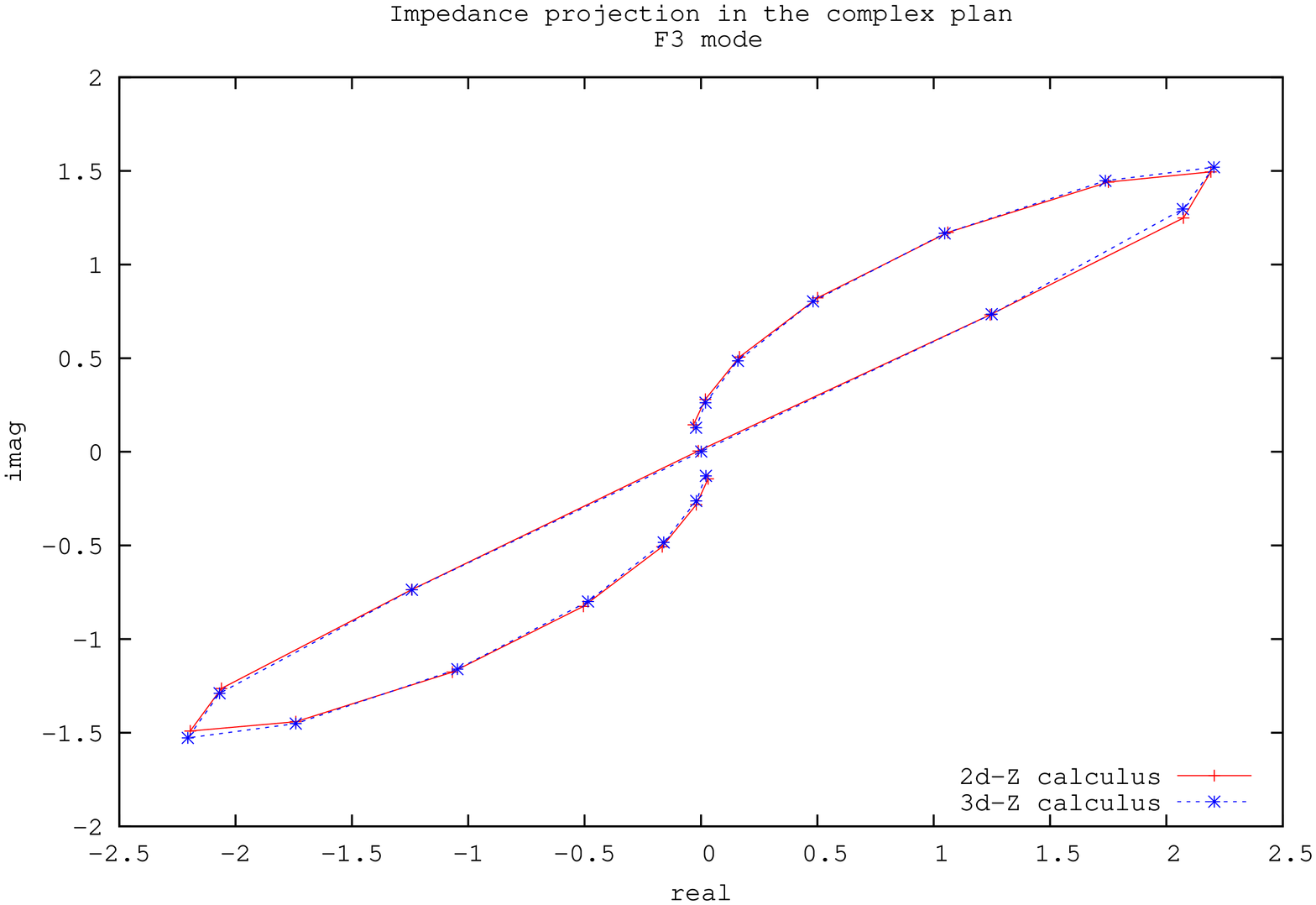} \\
\includegraphics[angle=-90,origin=c,width=6cm,height=6cm] {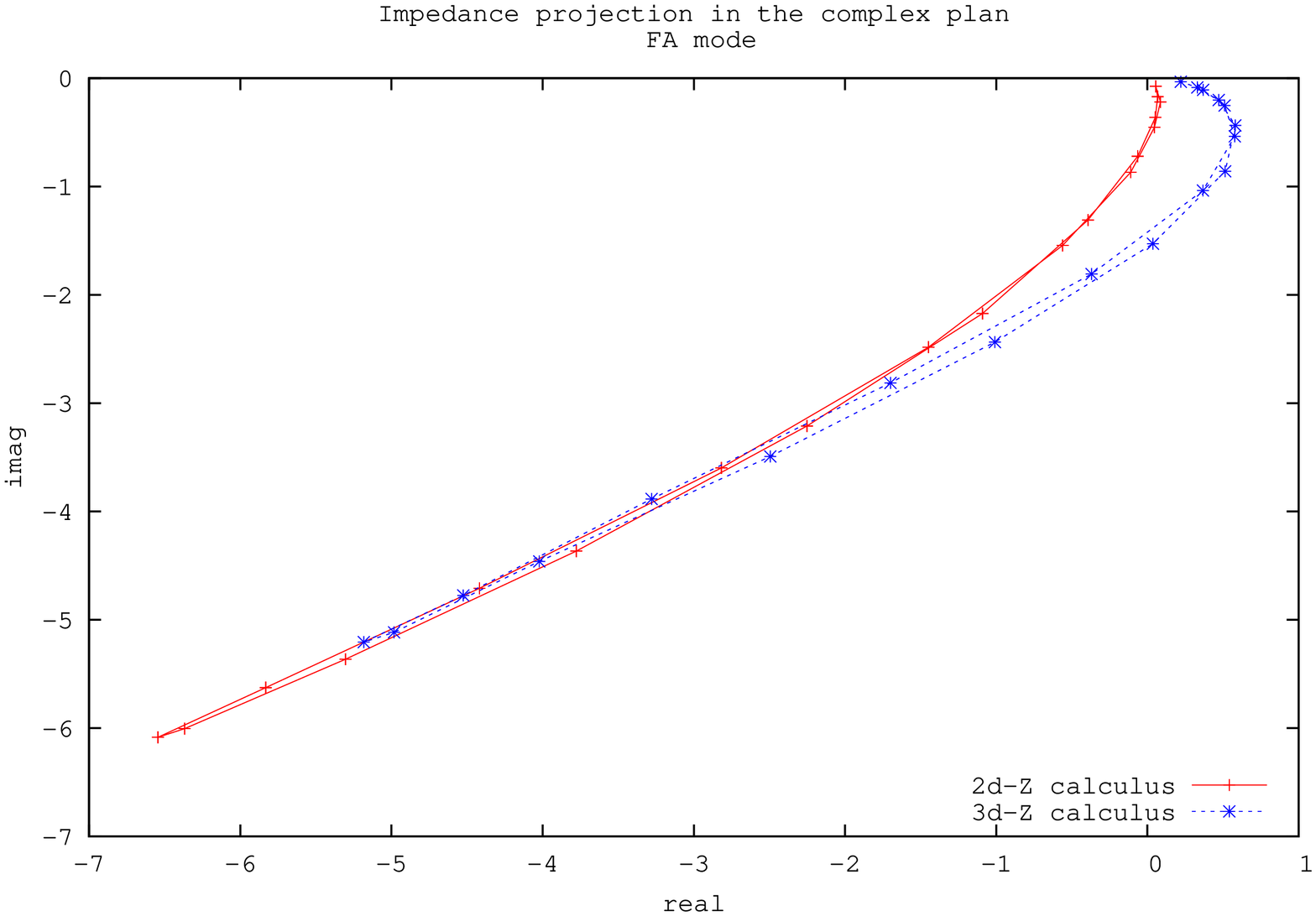} & 
\includegraphics[angle=-90,origin=c,width=6cm,height=6cm] {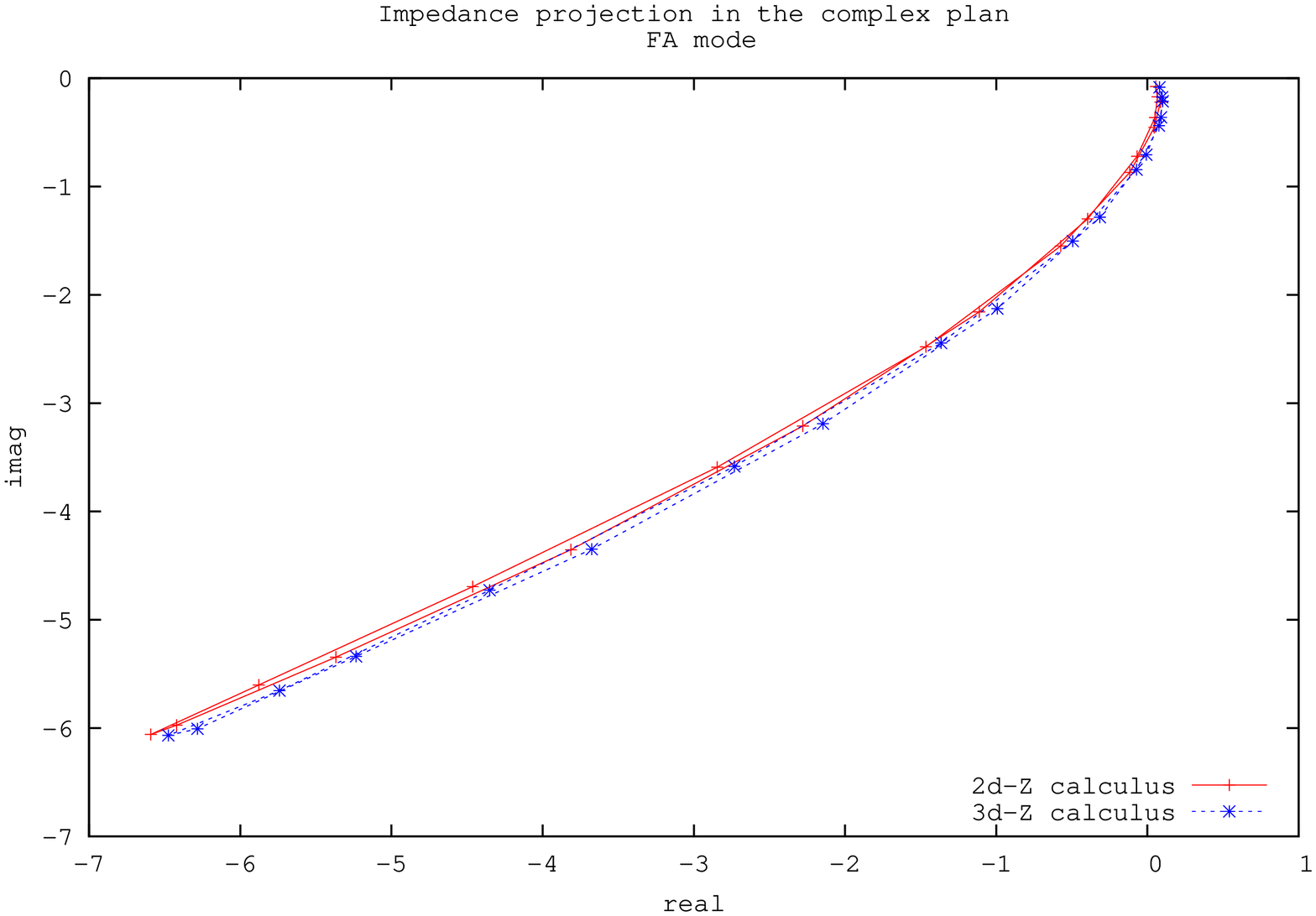} 
\end{tabular}
\caption{Complex-plan comparison of the Impedance signal 3D-vs-2D axisymmetric case: $\bf Z$-F3 (first line) $\bf Z$-FA (second line) where $\sigma_d$ is taken as $1.e4$ (first column) and $1.e3$ (second column). The configuration takes the deposit as a crown surrounding the Tube see Fig.~\ref{crownFirstTEST}.}\label{compareComlexplanValidation}
\end{figure}
\begin{figure}[!htb] 
\begin{tabular}{cc}
      \includegraphics[width=12cm,height=5cm] {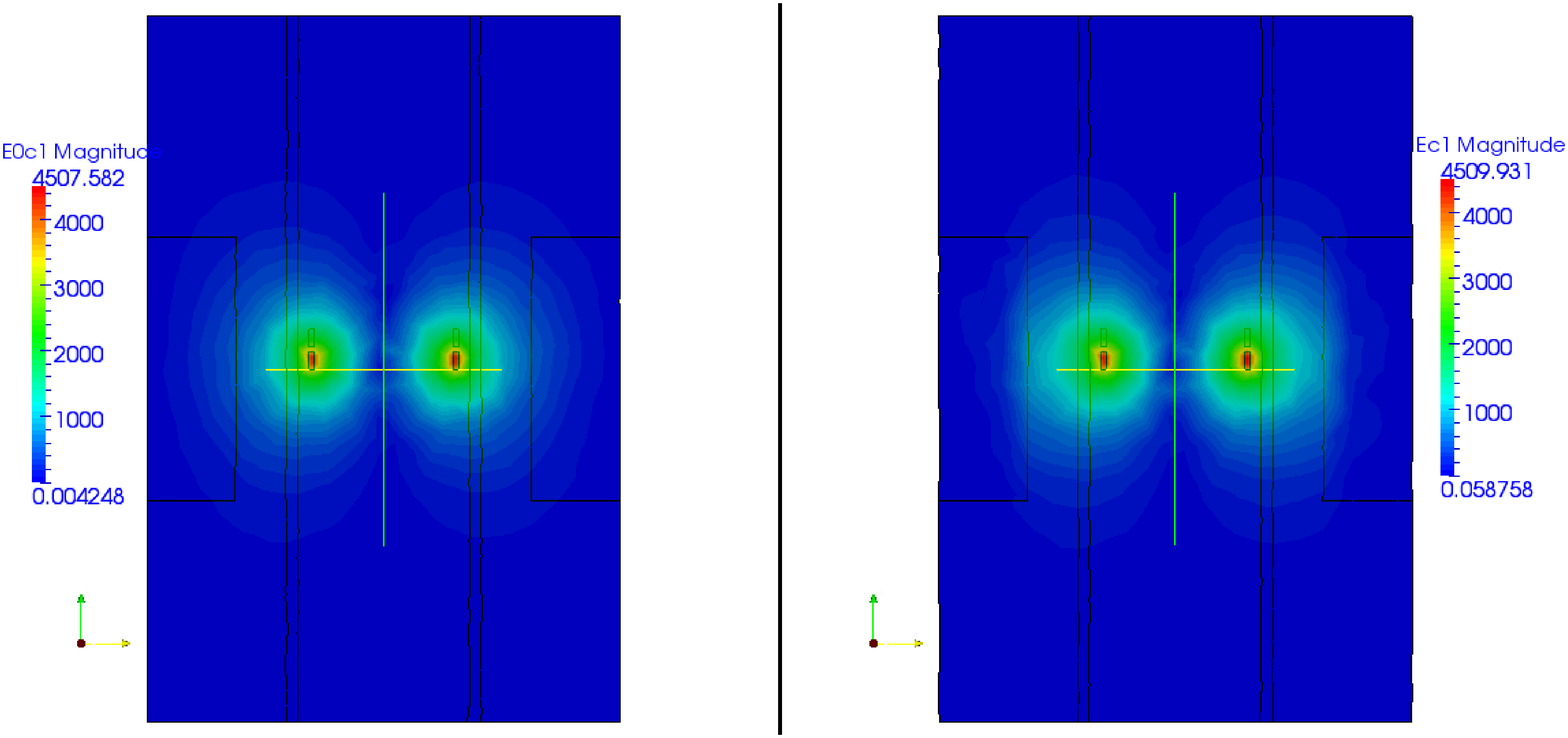}&\\
      \includegraphics[width=12cm,height=5cm] {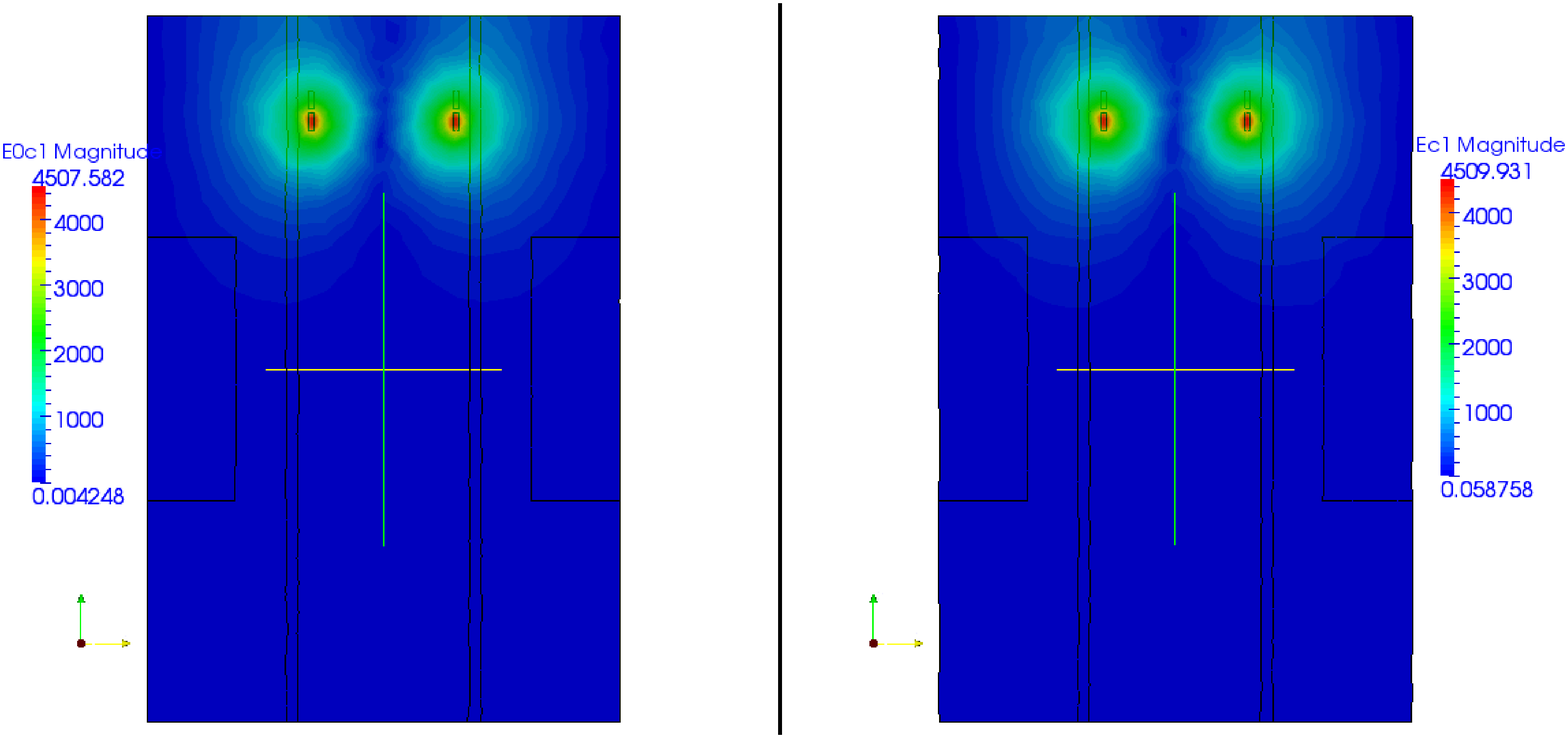}
\end{tabular}
\caption{Distribution of the magnetic field $\E=\i\omega\A+\V$; Solution to be compared with the case that take into account the IBC.}
\end{figure}

\subsection{2D and 3D validations of IBC implementation for TSP model}
The impedance boundary condition is very useful to model highly conductive part, where traditionally one has to deal with a mesh excessively refined in order to well approximate the exponentially decay of solution due to the skin depth of the material (proportional to its conductivity). The impedance boundary condition will therefore reduce the volume of the conductor to its surface. This reduction is cruelly shown at Fig.~\ref{IBC2dSOL} where we present the 2D-axisymmetric solution. The left graphic do not use IBC in the high conductor (modeled with a rectangle in 2D cross section), one can easily remark that the solution is very thin at the surface and practically vanishing due to the skin depth of the TSP. The right graphic shows that we can disregard the volume of  TSP and replace the solution inside by a surface solution taking into account the approximation Eq.~\eqref{IBC}. The comparative impedance signal results are shown at Fig.~\ref{IBC2d}.
\begin{figure}[!htbp] 
   \begin{minipage}[c] {.46\linewidth}
\centering      \includegraphics[width=6cm,height=8cm] {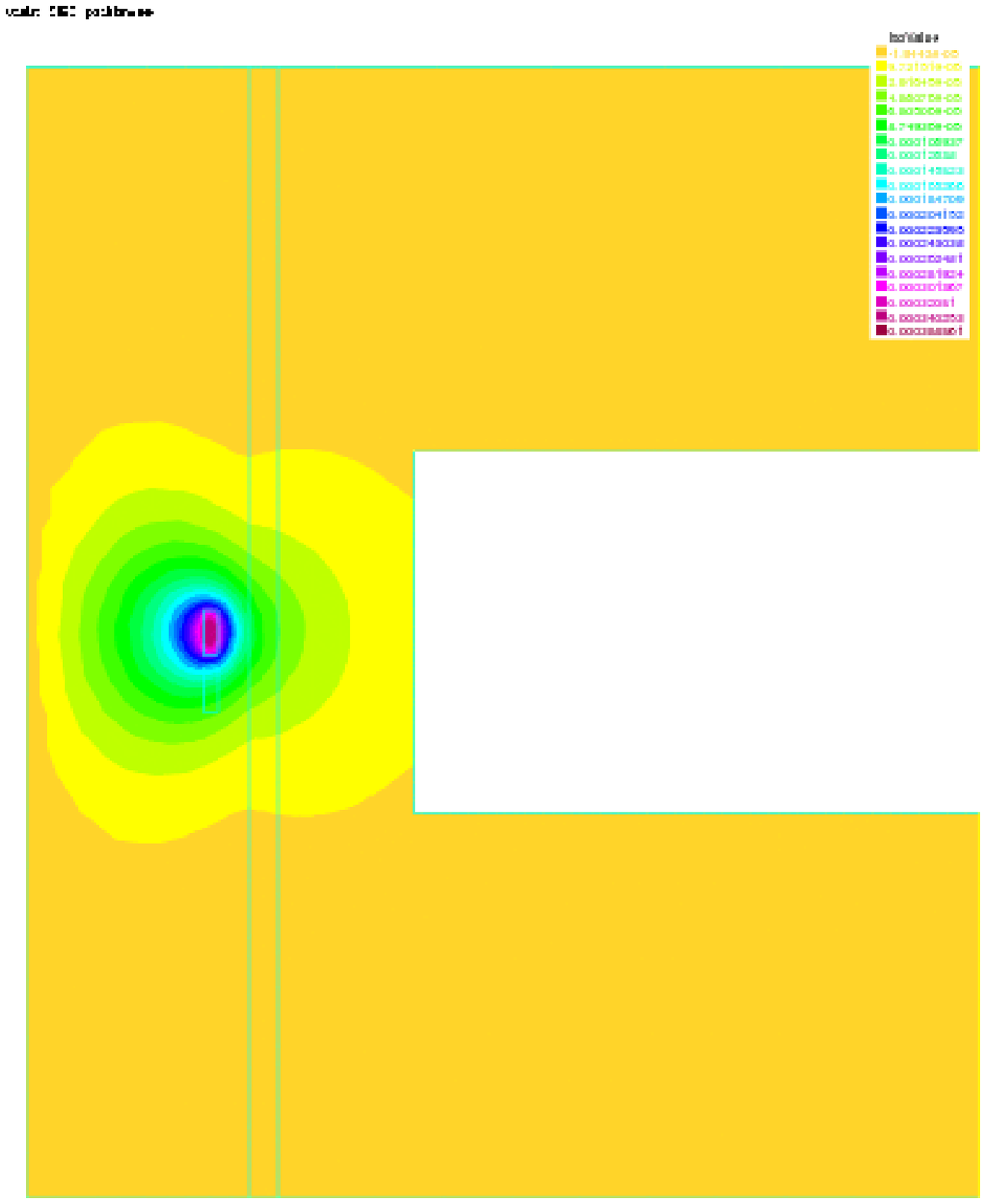}
   \end{minipage} 
   \begin{minipage}[c] {.46\linewidth}
\centering      \includegraphics[width=6cm,height=8cm] {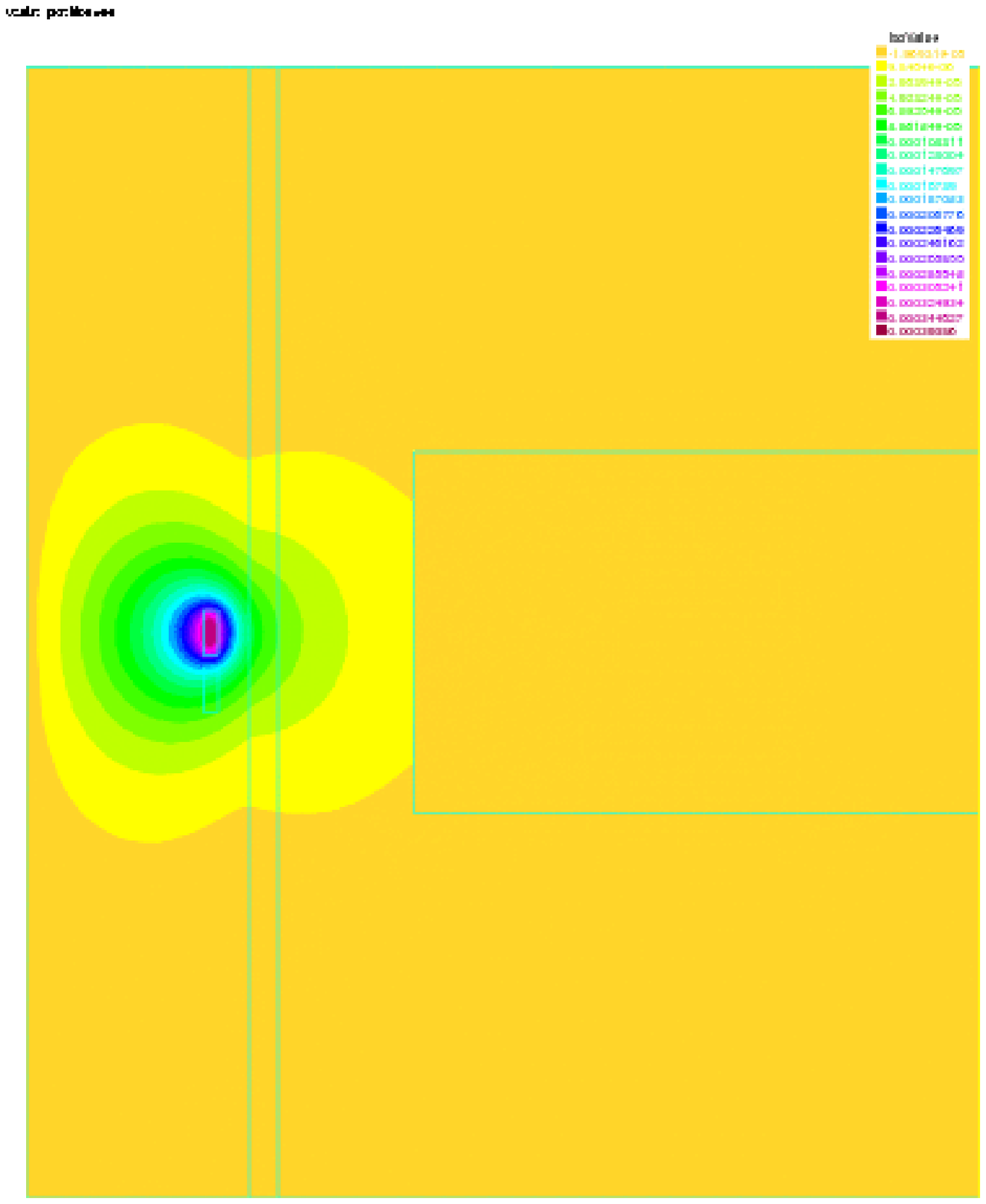}
   \end{minipage}
   \caption{2D solution with (left) and without (right) the impedance boundary condition for a highly conductive part (crown) modeling the TSP}\label{IBC2dSOL}
\end{figure}

\begin{figure}[!htbp] 
      \includegraphics[width=12cm,height=12cm] {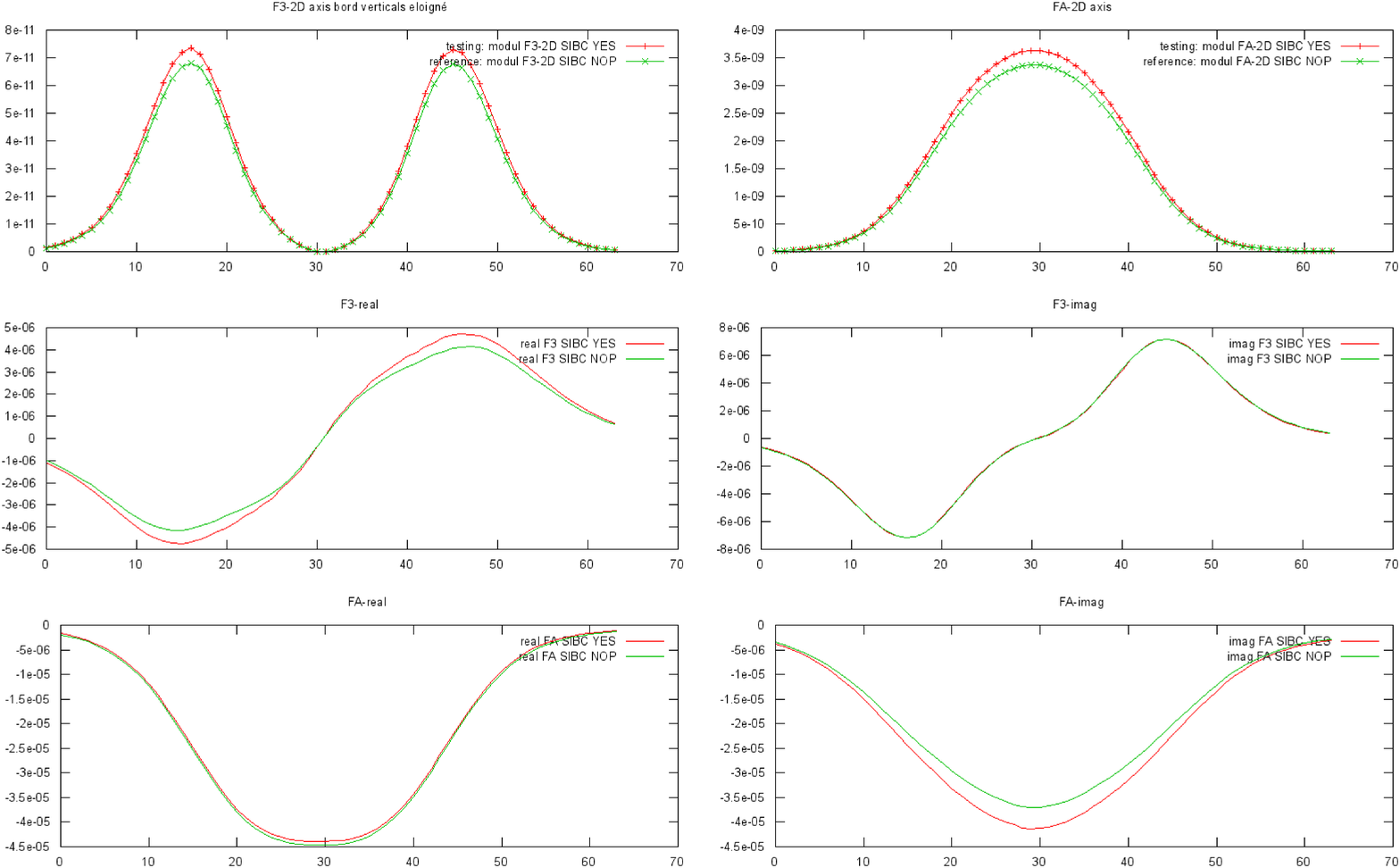}
   \caption{Validation of the IBC in the 2d case.The permeability and conductivity parameters are taken for the vacuum $(\mu_0=4\pi e-07,\sigma_0=0.)$, for the tube  $(\mu_t=1.01\mu_0,\sigma_t=.97e06)$, for the TSP $(\mu_p=70\mu_0,\sigma_p=1.75e06$). This test case corresponds to the configuration presented at Fig.\ref{IBC2dSOL}.}
   \label{IBC2d}
\end{figure}

	In what concern the 3D simulation, it is a little bit different from the 2D because the use of 3D simplex (tetrahedron) in order to approximate surface finite element by penalization~\footnote{The actual version 20.2 of FreeFem++ (september 2013), don't recognize finite element defined at a 3D surface.}. In fact, we disregard element defined on the TSP volume and we penalize the nodes in the closer the complement of the TSP. Fig.~\ref{IBC3dSOL} present the distribution of the solution and Fig.~\ref{IBC3d} compare the impedance signal produced with the 3D simulation against the 2D simulation.  
\begin{figure}[!htbp] 
   \begin{minipage}[c] {.46\linewidth}\centering  
       \includegraphics[width=8cm,height=8cm] {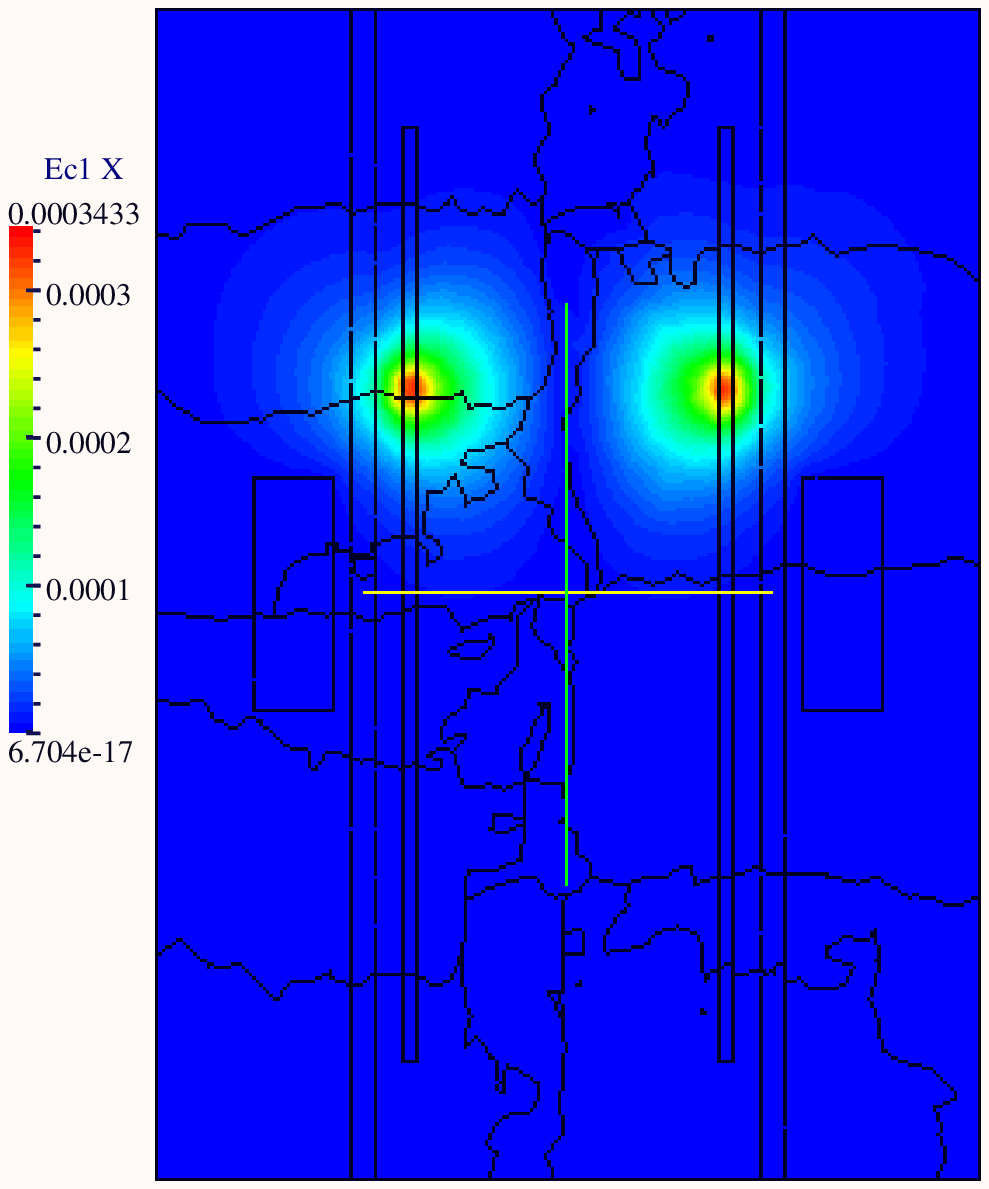}
          \end{minipage}
      \begin{minipage}[c] {.46\linewidth}\centering  
       \includegraphics[width=8cm,height=8cm] {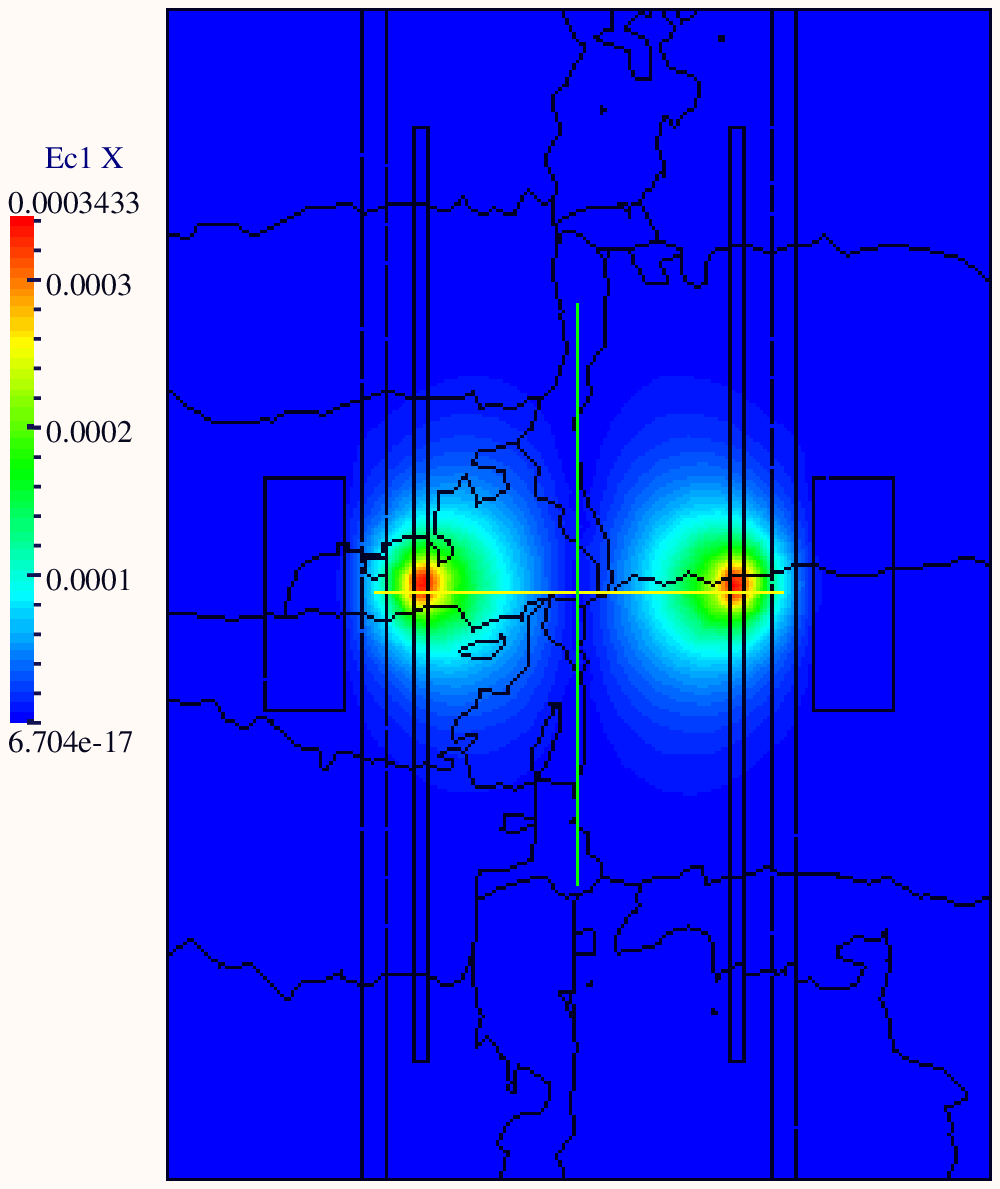}   \end{minipage}
   \caption{3D solution with the impedance boundary condition, which produces the impedance signal as presented in Fig.~\ref{IBC3d}. The deposit in this example is presented by the rectangular (cross section). We show in this figure that interior point don't contribute to the 3D solution because of the penalization techniques.}\label{IBC3dSOL}
\end{figure}
\begin{figure}[!htbp] 
\begin{tabular}{cc}
      \includegraphics [width=14cm,height=4cm]{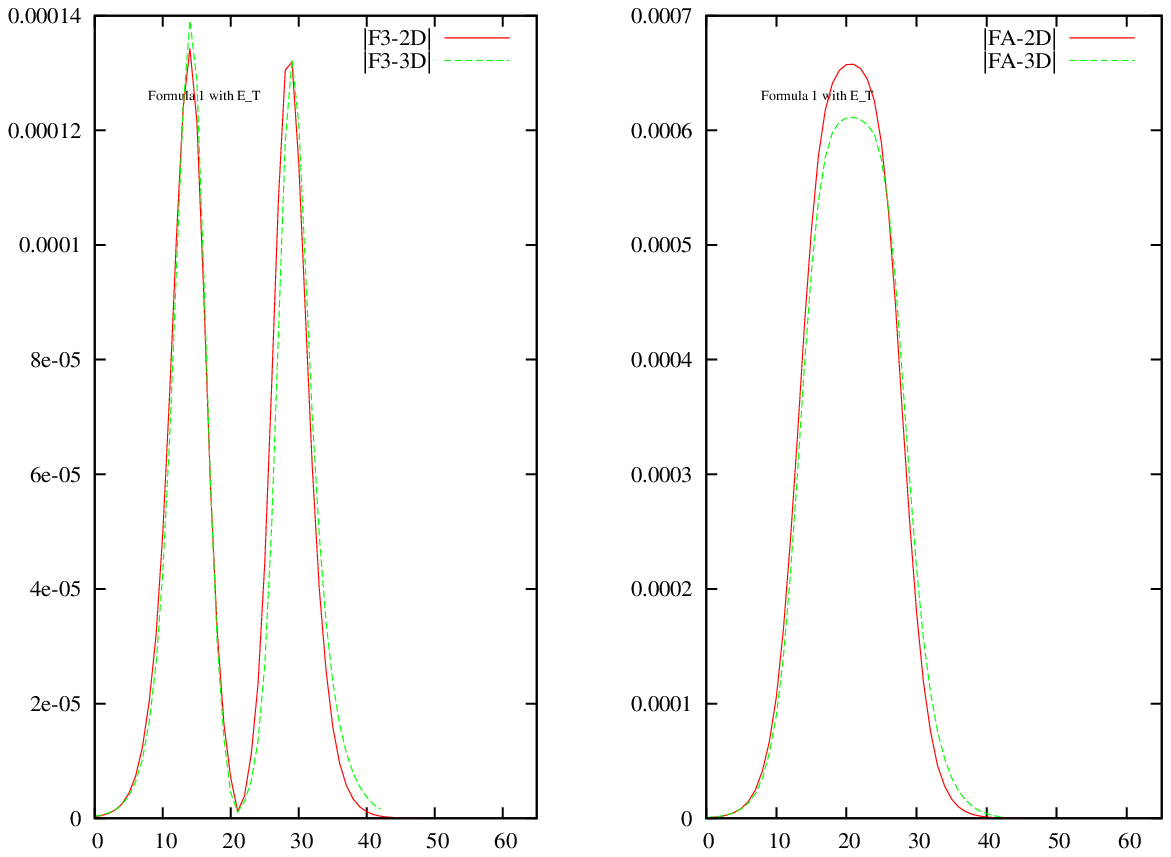}\\
      \includegraphics [width=14cm,height=4cm]{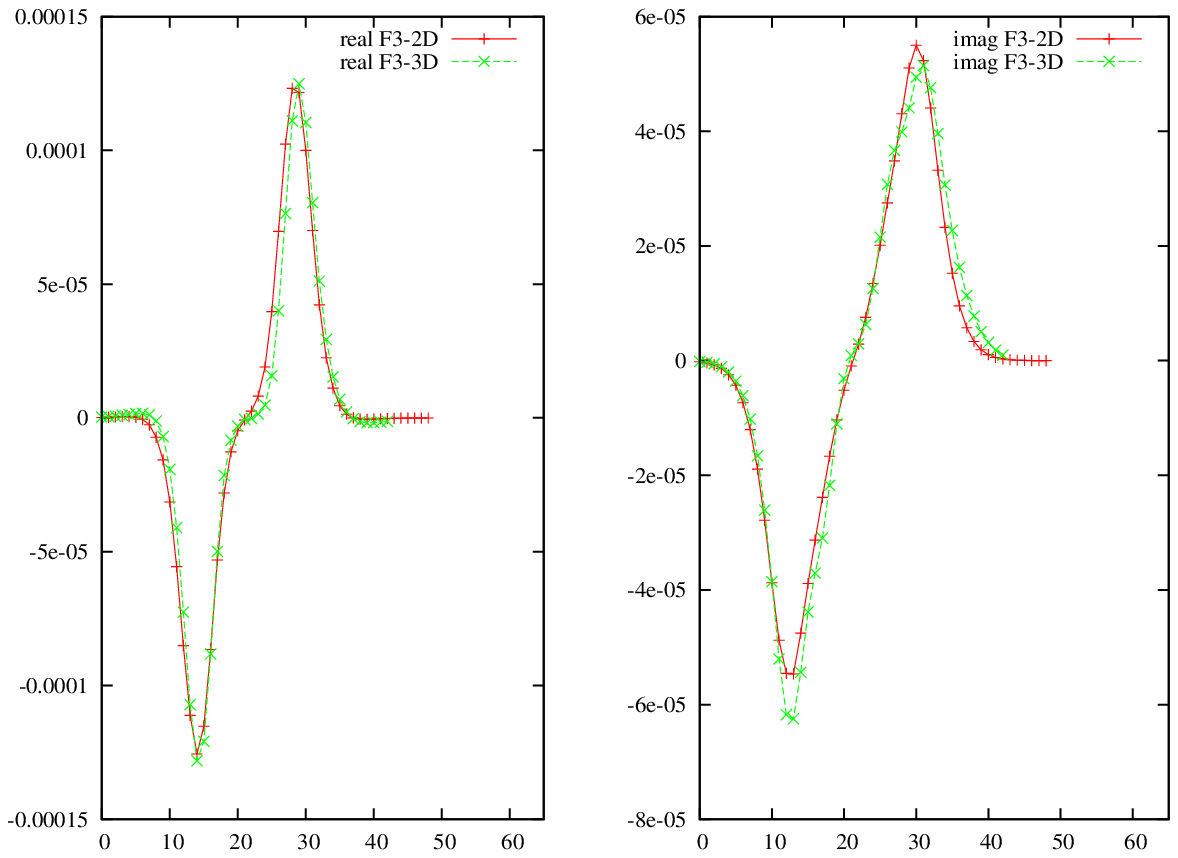}\\
      \includegraphics [width=14cm,height=4cm]{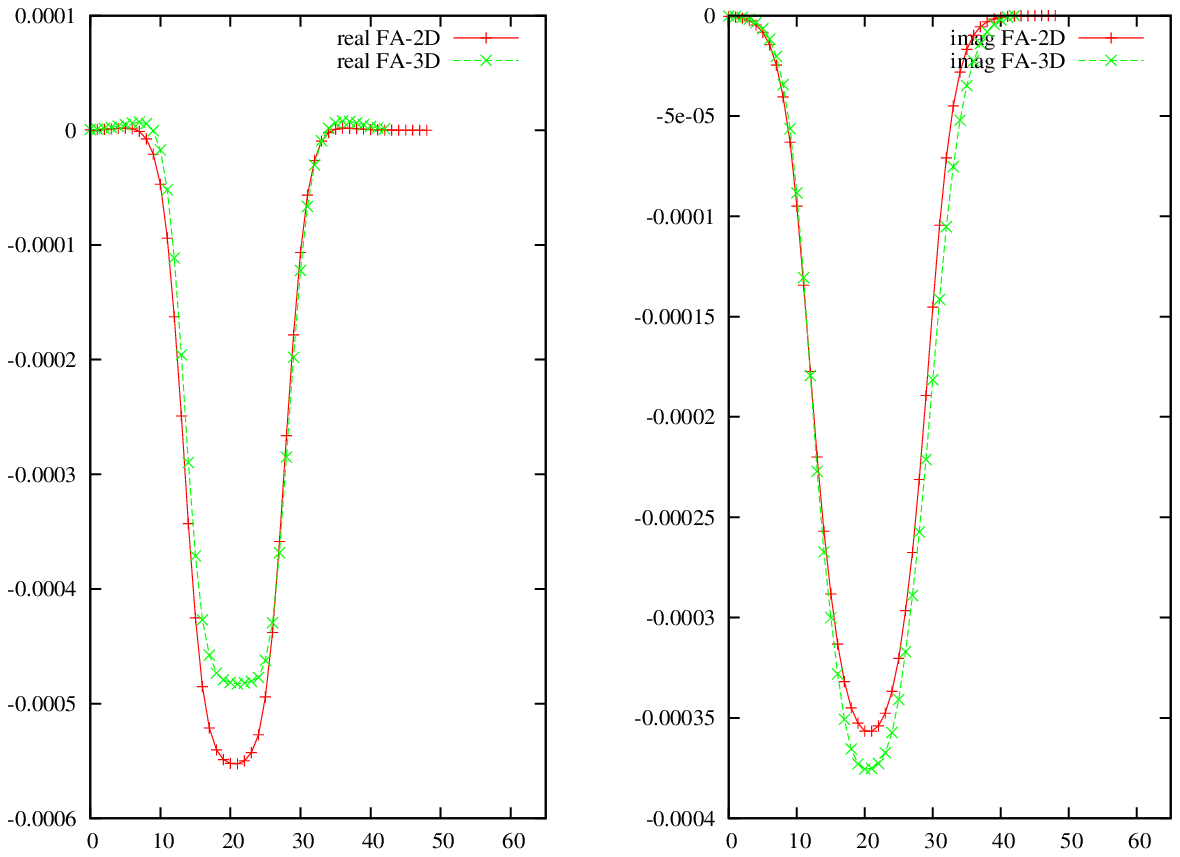}      
\end{tabular}
   \caption{Validation of IBC for the 3D simulation. The permeability and conductivity parameters are taken for the vacuum $(\mu_0=4\pi e-07,\sigma_0=0.)$, for the tube  $(\mu_t=1.01\mu_0,\sigma_t=.97e06)$, for the TSP $(\mu_p=\mu_0,\sigma_p=.97e06$).This test case corresponds to the configuration presented at Fig.\ref{IBC3dSOL}}\label{IBC3d}
\end{figure}
In practice, the incorporation of IBC in the 3D simulation, is achieved by penalizing interior vortex of the TSP by manipulating lines of the main matrix. In fact, lines corresponding to interior points will be eliminated except the diagonal term of the matrix where will be replaced by the value "1".

\subsection{Inversion by gradient descent algorithm}

	The adjoint problem is governed by the conjugate sesquilinear form of the direct problem. Since we use sparse-parallel solver that provides LU factorization of the matrix, we would reuse that factorization in the resolution of the adjoint state. A simple operation consists in dividing by $i\omega$ all terms in variational formulation enables the matrix $M$ of the direct problem to be symmetric. Hence the adjoint problem matrix is nothing but the transpose of $M$ ( with conjugation of complex entries). 
	Since $\overline{M}w=b$ is equivalent to $\overline{\overline{M}w=b}$ that writes $M\overline{w}=\overline{b}$.

	We present in this section, two types of geometrical parameters (width size) inversion for the reconstruction of a one layer deposit and 3 layers deposit. 

	 We consider first a validation of the gradient (shape gradient) of the cost function by comparing the shape gradient to the central difference approximation of the first derivative  of the cost function. The central difference is made with a small shape perturbation in the direction of the outward normal. We validate first the gradient, then we proceed to the numerical inversion that consists of the reconstruction of one layered deposit. Fig.~\ref{J_L}  presents the function decay with the increasing number of iterations of the inversion algorithm. Several snapshots (see Fig~\ref{animcheckL}) are also presented to show the convergence of the shape deformation to the deposit, with which we picked up impedance measurements. 
\begin{figure}[!htbp] 
   \begin{minipage}[c] {.46\linewidth}
      \includegraphics[scale=.6] {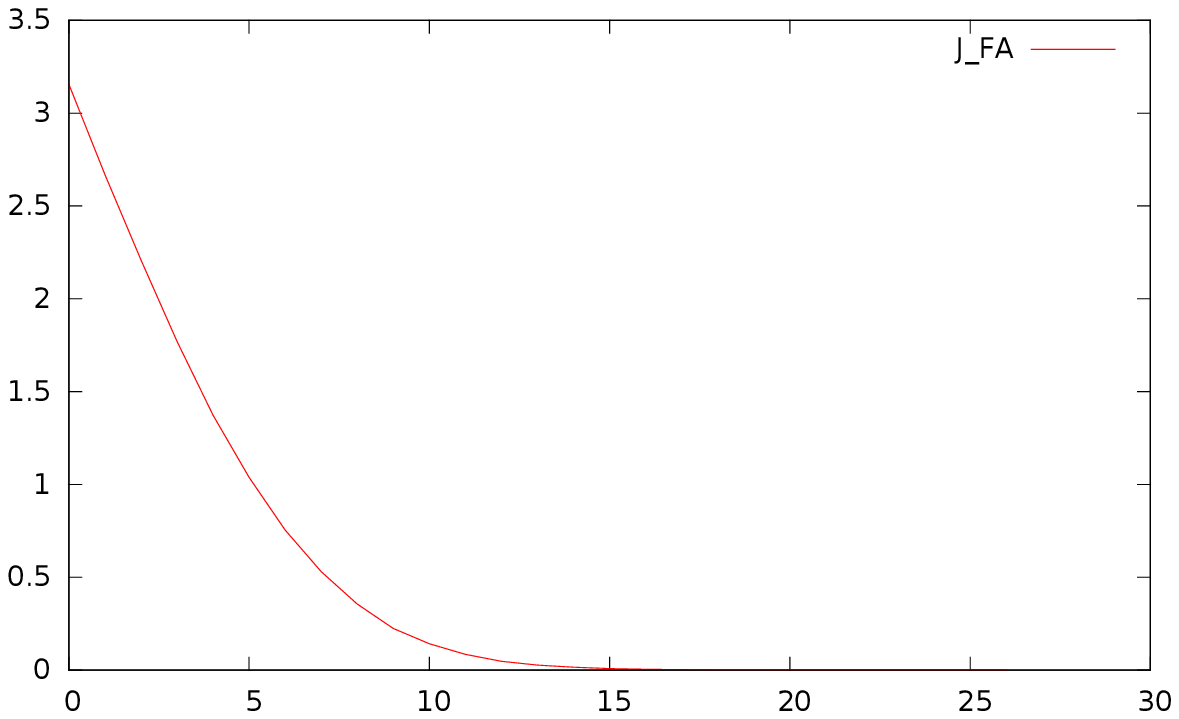}
   \end{minipage} \hfill
   \begin{minipage}[c] {.46\linewidth}
      \includegraphics[scale=.6] {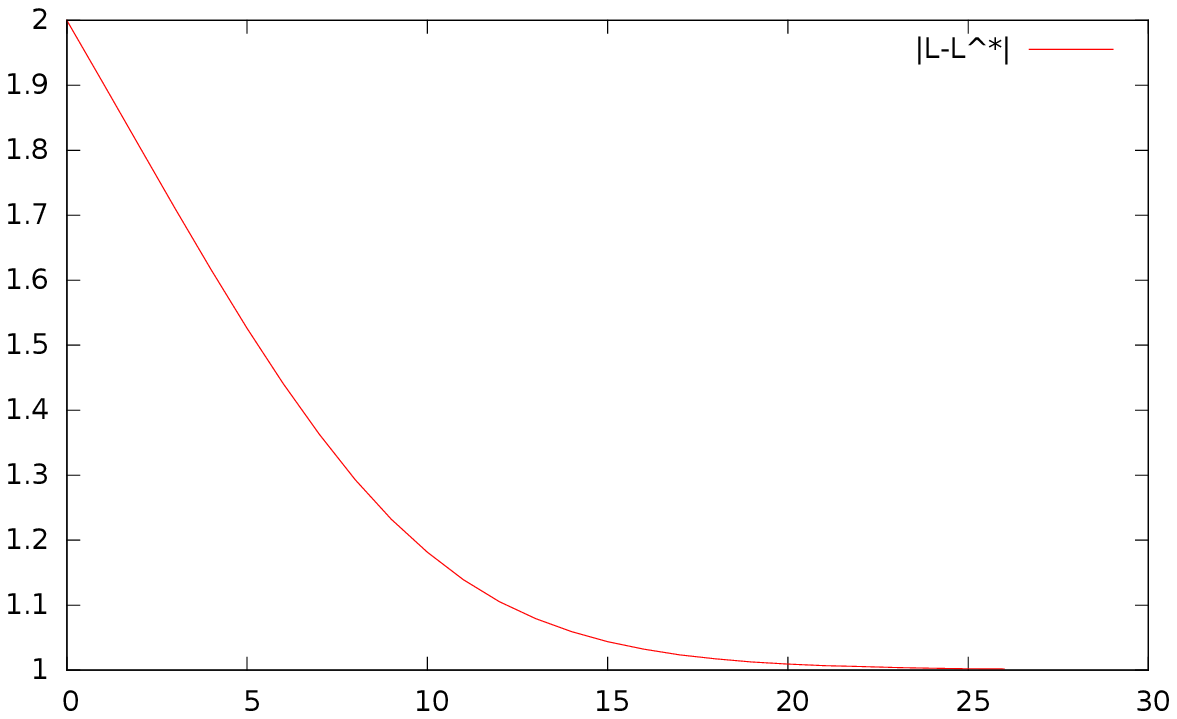}
   \end{minipage}
   \caption{Cost functional decaying with the iteration of the steepest descent algorithm. The parameters $\mu$ and $\sigma$ are taken for the vacuum as $(\mu_0=4\pi e-07,\sigma_0=0)$, for the tube as $(\mu_t=1.01\mu_0,\sigma_t=.97e6)$ and for the deposit $\mu_d=\mu_0,\sigma_d=3e03$}
   \label{J_L}
\end{figure}
	The test case presented in Fig.~\ref{J_L} consists in retrieving the width of an axis-symmetric deposit as crown surrounding the tube without the presence of the TSP.  
We present in Fig.~\ref{animcheckL} some plots of the solution (the iterate mesh) that converge.

\begin{figure}[!htbp] \centering
\begin{tabular}{lr}
\includegraphics[angle=0,origin=c,width=9cm,height=10cm] {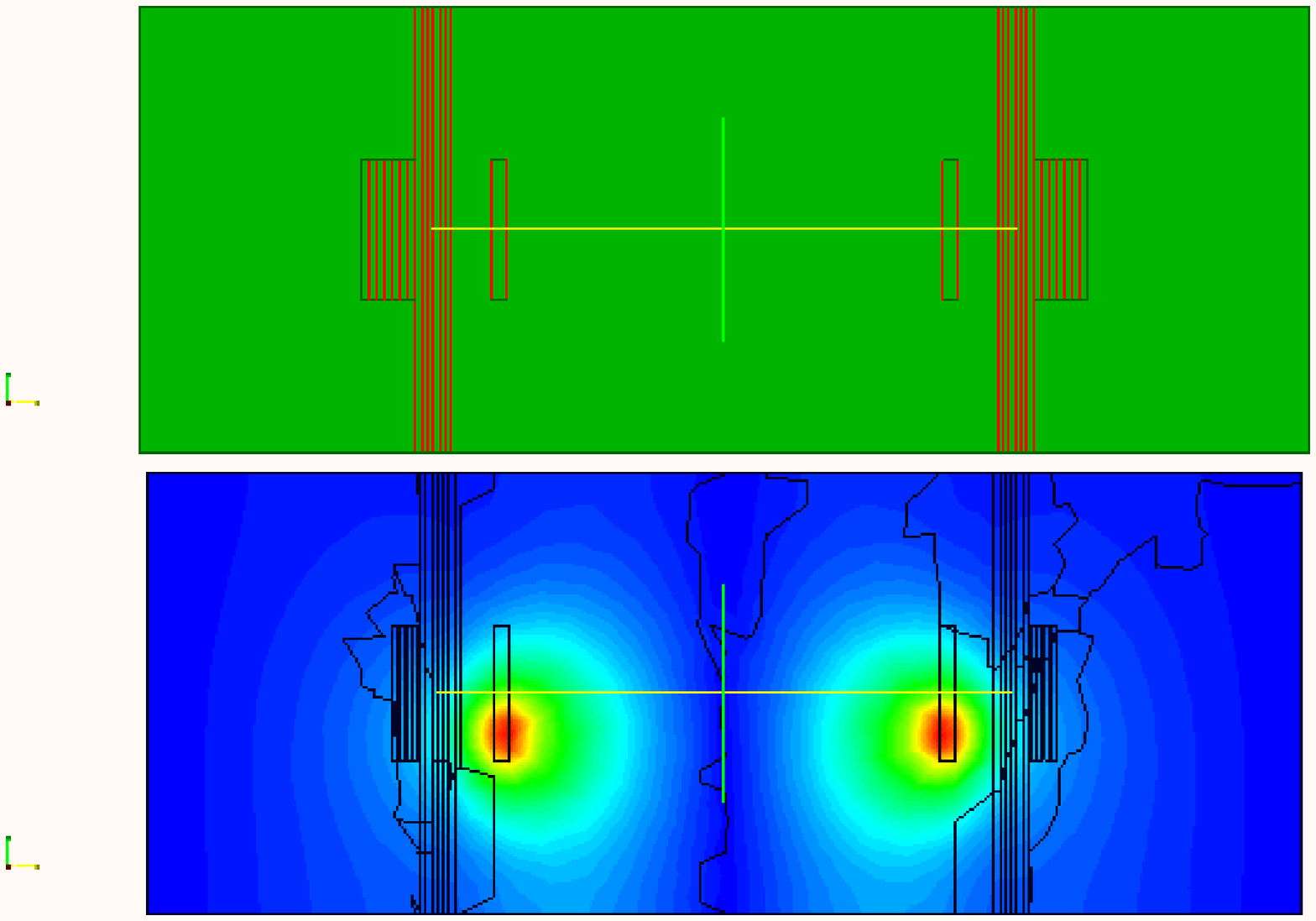} & 
\includegraphics[angle=0,origin=c,width=9cm,height=10cm] {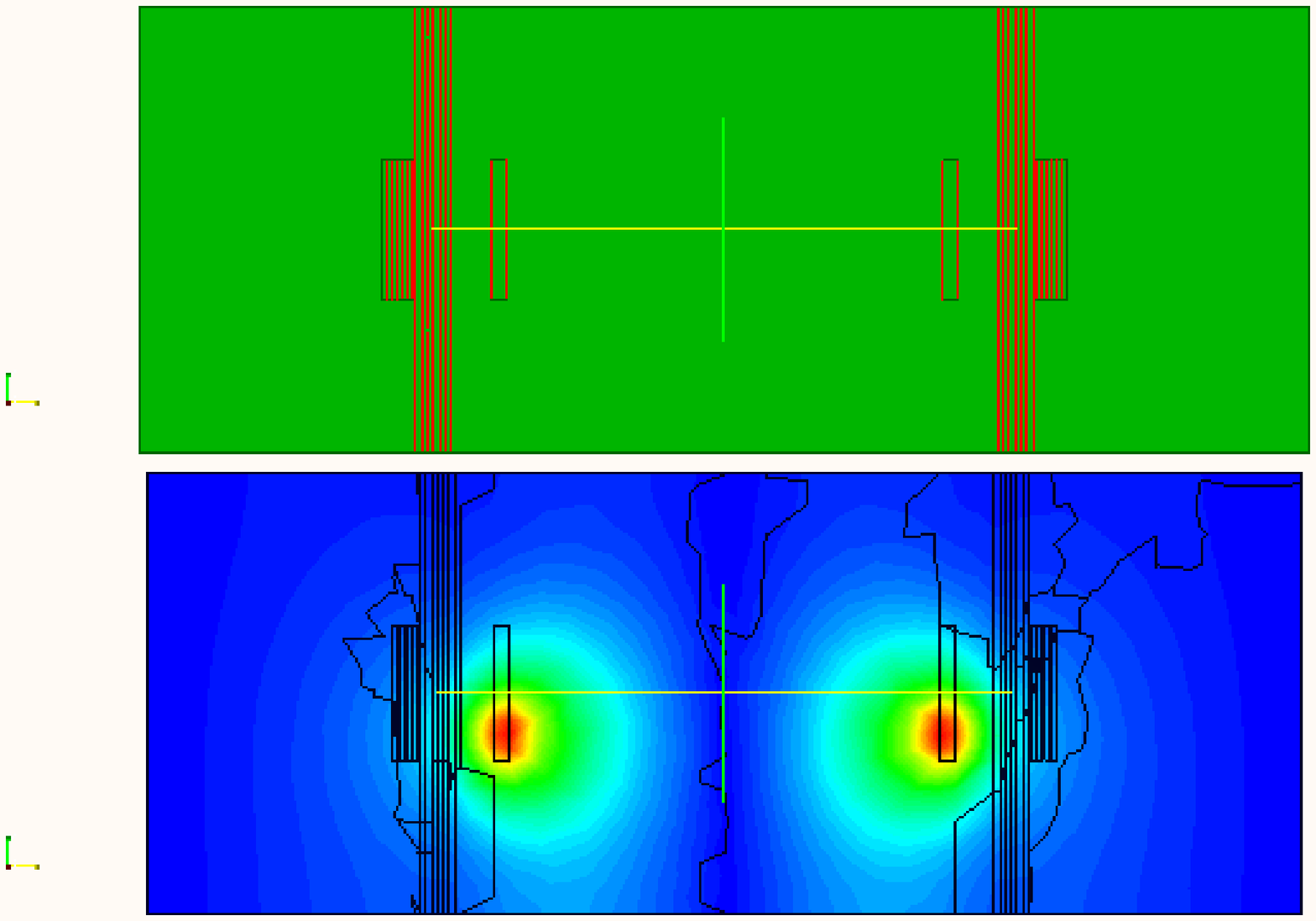} \\
\includegraphics[angle=0,origin=c,width=9cm,height=10cm] {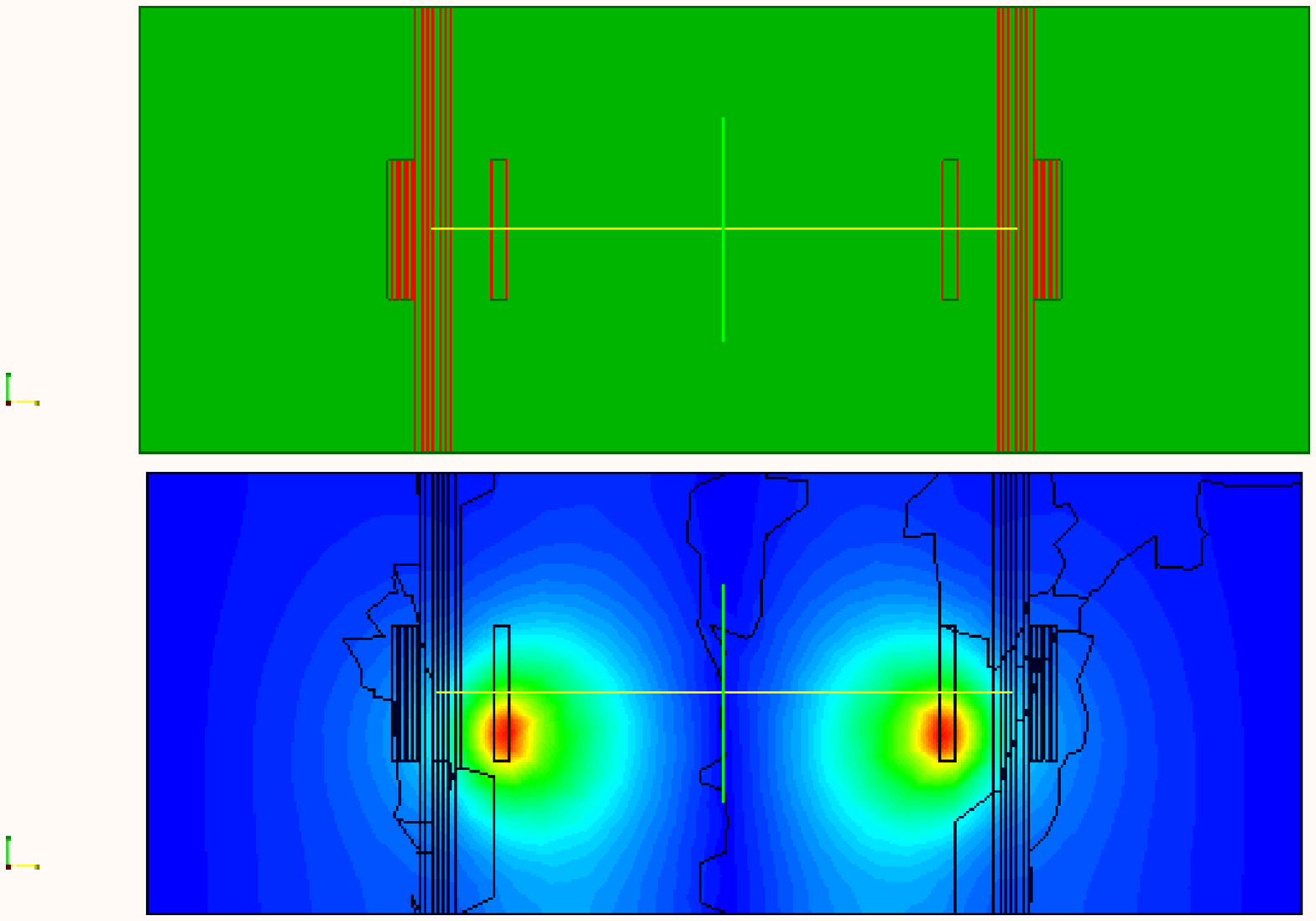} & 
\includegraphics[angle=0,origin=c,width=9cm,height=10cm] {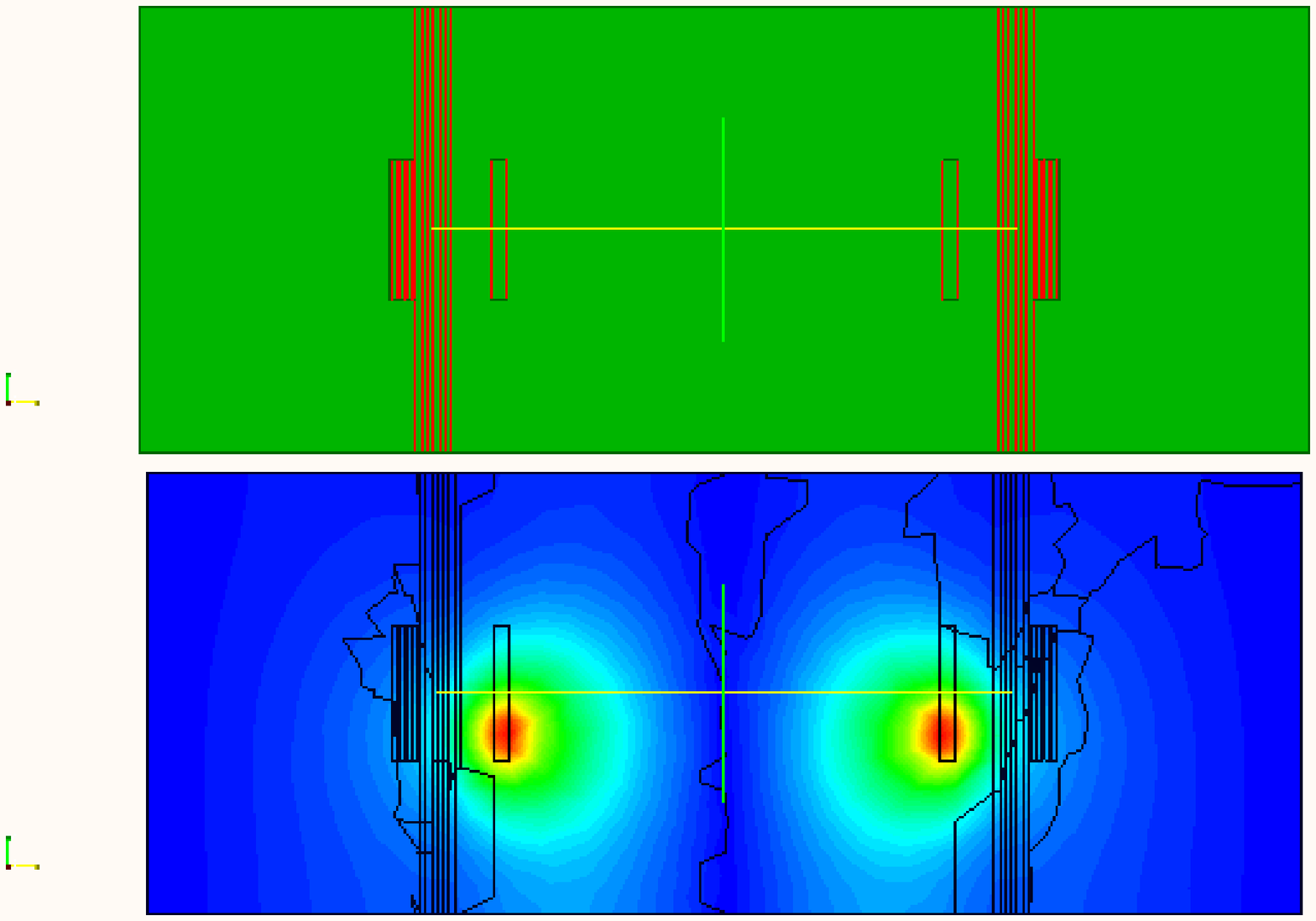} 
\end{tabular}
\caption{Snapshot of the solution for the inverse problem that considers one layere deposits.}\label{animcheckL}
\end{figure}
After the validation of the gradient's calculus by one layer reconstruction, we complicate the problem by considering a 3 layers deposit. When we have to deal with three geometrical parameters, namely $L1,L2$ and $L3$ as widths of a 3 different deposits. Fig.~\ref{J_L123} presents the function decay that considers both FA and F3 impedance signals. 
\begin{figure}[!htp] 
   \begin{minipage}[c] {.46\linewidth}
      \includegraphics[scale=.3,angle=-90] {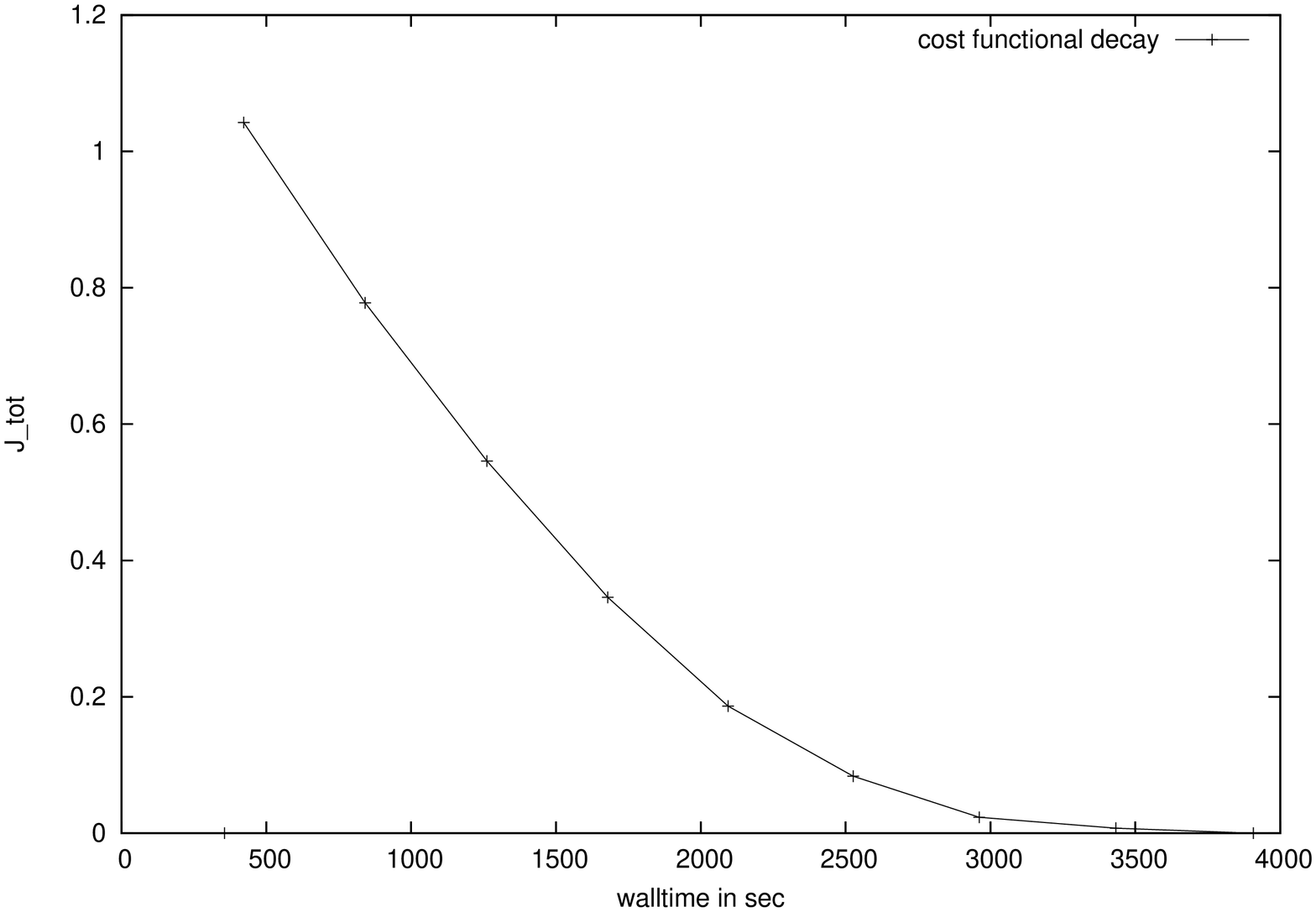}
   \end{minipage} \hfill
   \begin{minipage}[c] {.46\linewidth}
      \includegraphics[scale=.3,angle=-90] {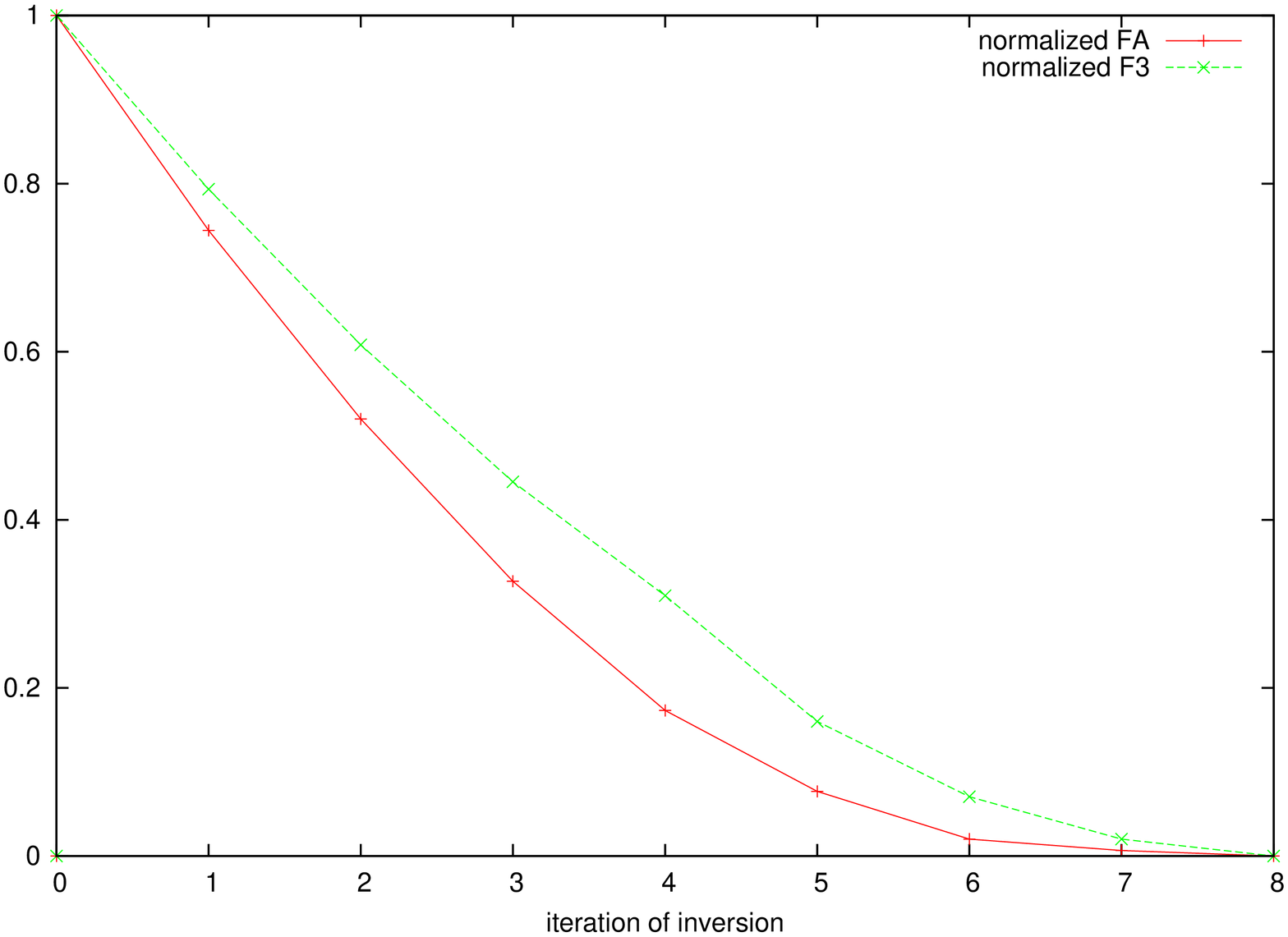}
   \end{minipage}
   \caption{Cost functional decaying with the iteration of the steepest descent algorithm. The minimization procedure takes as cost both FA and F3 signal for the reconstruction of deposit by manipulating 3 layers with different width. The parameters $\mu$ and $\sigma$ are taken for the vacuum as $(\mu_0=4\pi e-07,\sigma_0=0)$, for the tube as $(\mu_t=1.01\mu_0,\sigma_t=.97e6)$ and for the deposit $\mu_d=\mu_0,\sigma_d=3e03$}
   \label{J_L123}
\end{figure}
We present in Fig.~\ref{L123shapeconvg} the step of the inverse solver that reconstruct each layer of the deposit simultaneously. For completeness we add in Fig.~\ref{L123shapeconvg} snapshots of the solution of the mesh transform with the iteration of the inverse problem increase. 

\begin{figure}
   \begin{minipage}[c] {.46\linewidth}
      \includegraphics[scale=.3,angle=-90] {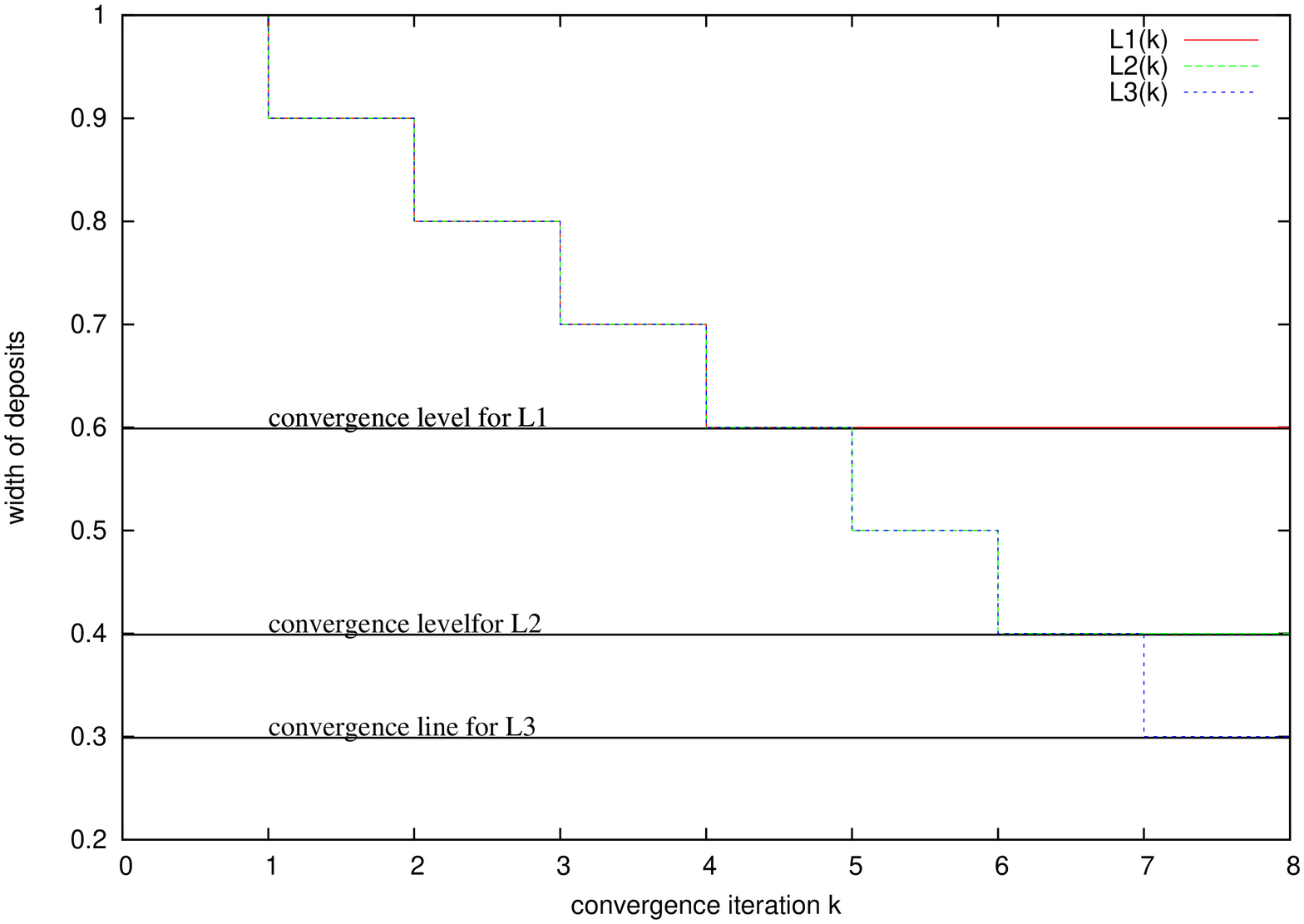}
   \end{minipage} \hfill
   \begin{minipage}[c] {.46\linewidth}
      \includegraphics[scale=.3,angle=0] {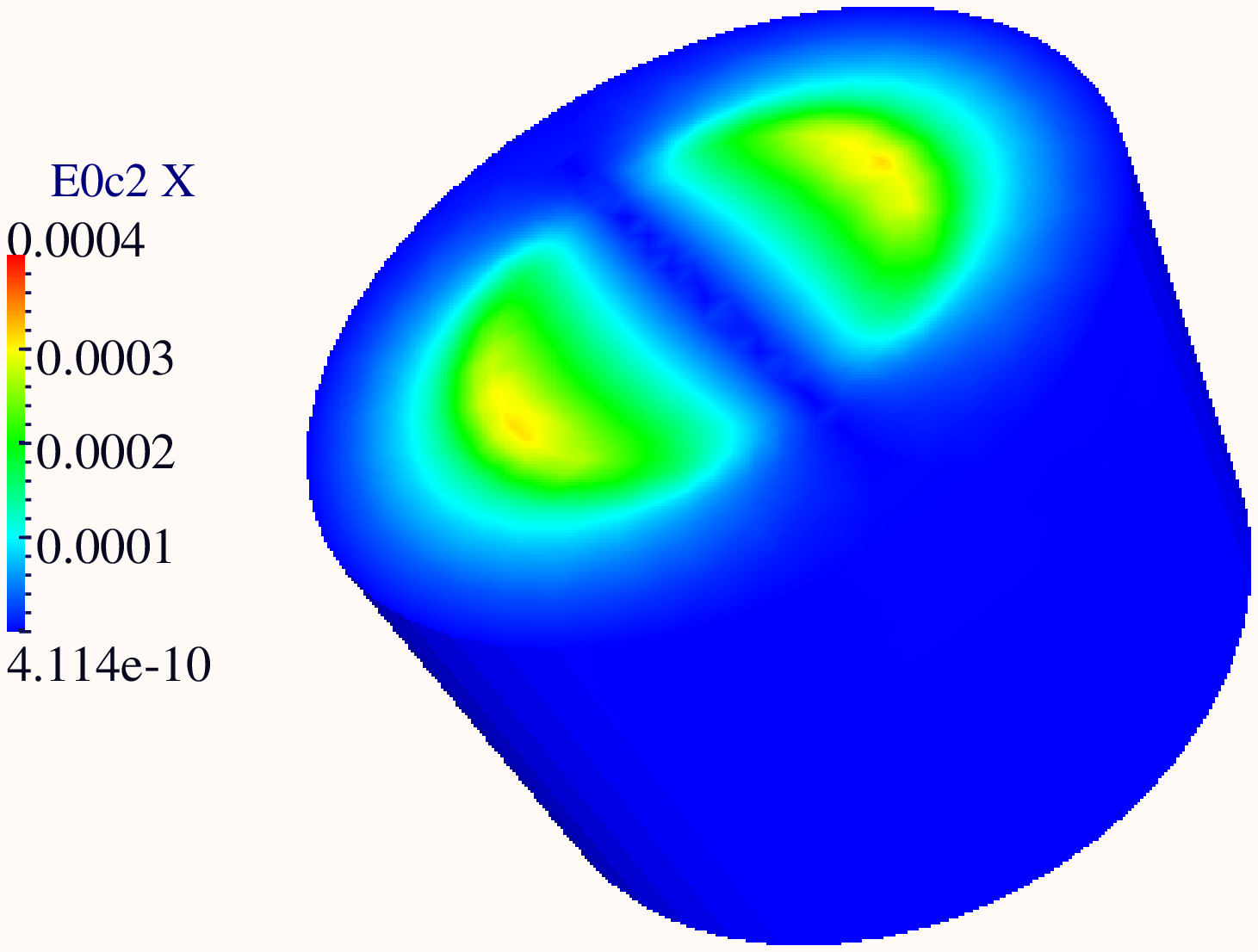}
   \end{minipage}
   \caption{Cost functional decay with the iteration of the steepest descent algorithm. We show that width of layers converges to the right width of the measured deposit. The graph at the right presents the appearance of the 3D solution.}
   \label{L123shapeconvg}
\end{figure}

\begin{figure}[!htbp] 
\begin{tabular}{cc}
\includegraphics[angle=0,origin=c,width=6cm,height=5cm] {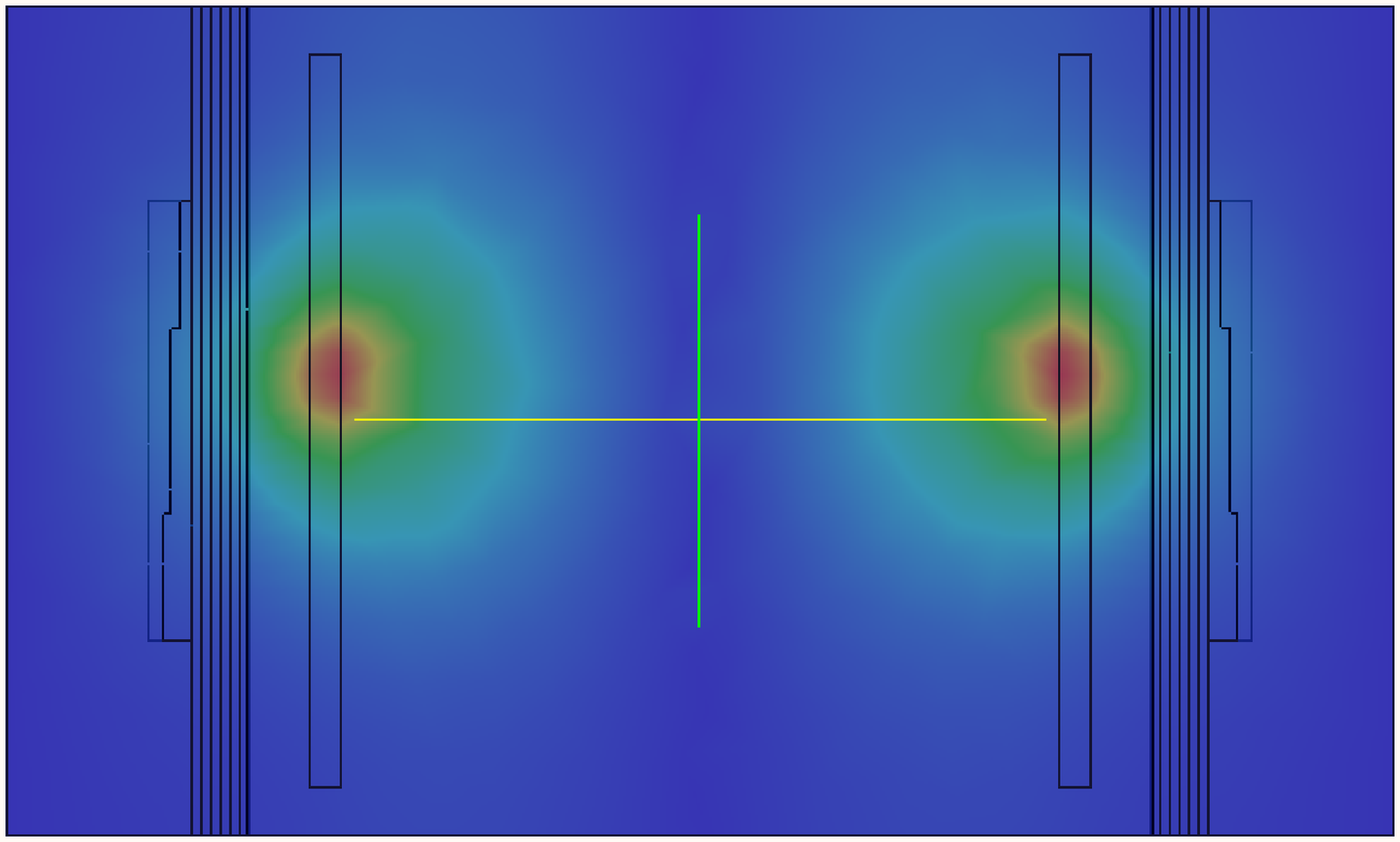} & 
\includegraphics[angle=0,origin=c,width=6cm,height=5cm] {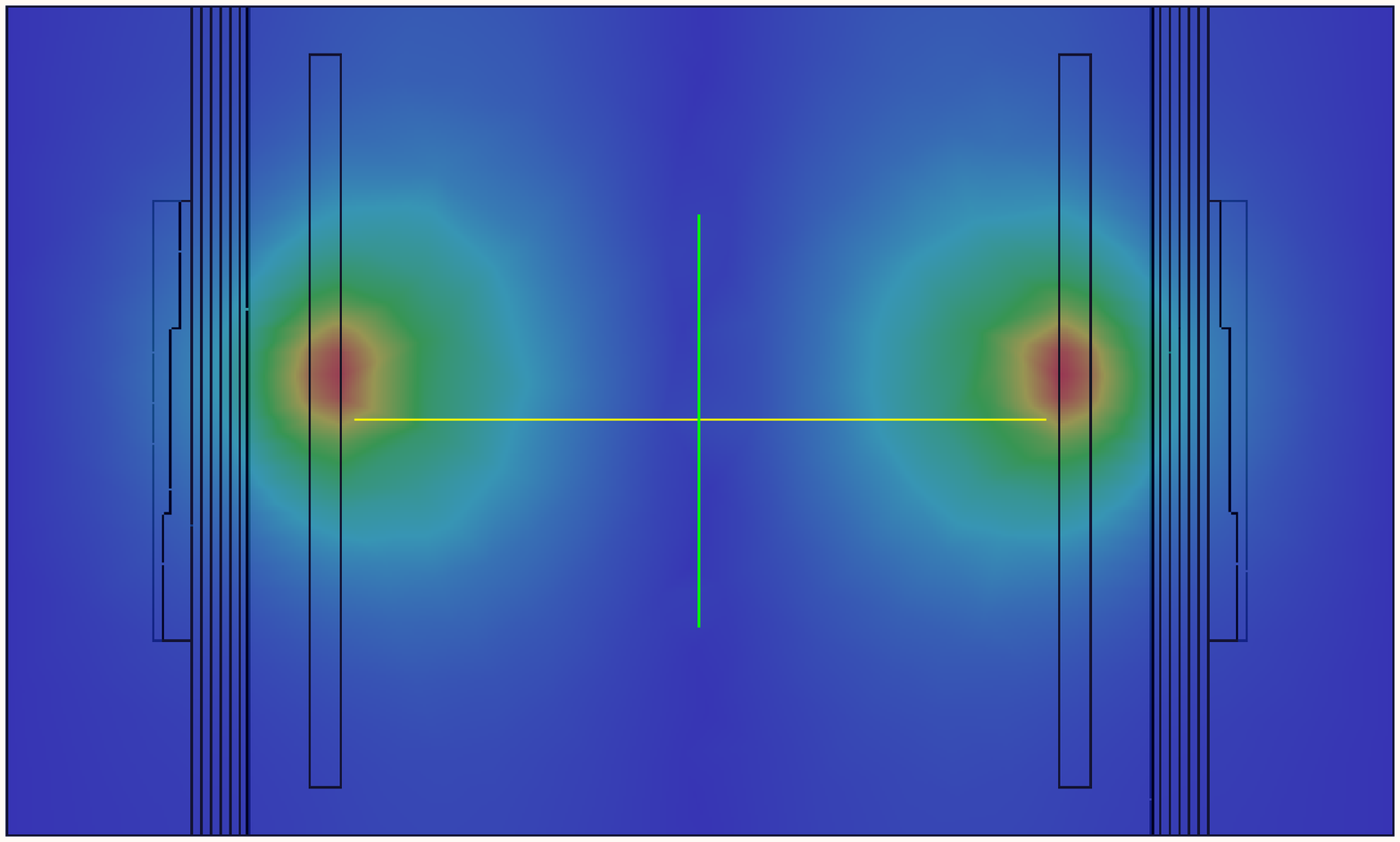} \\
\includegraphics[angle=0,origin=c,width=6cm,height=5cm] {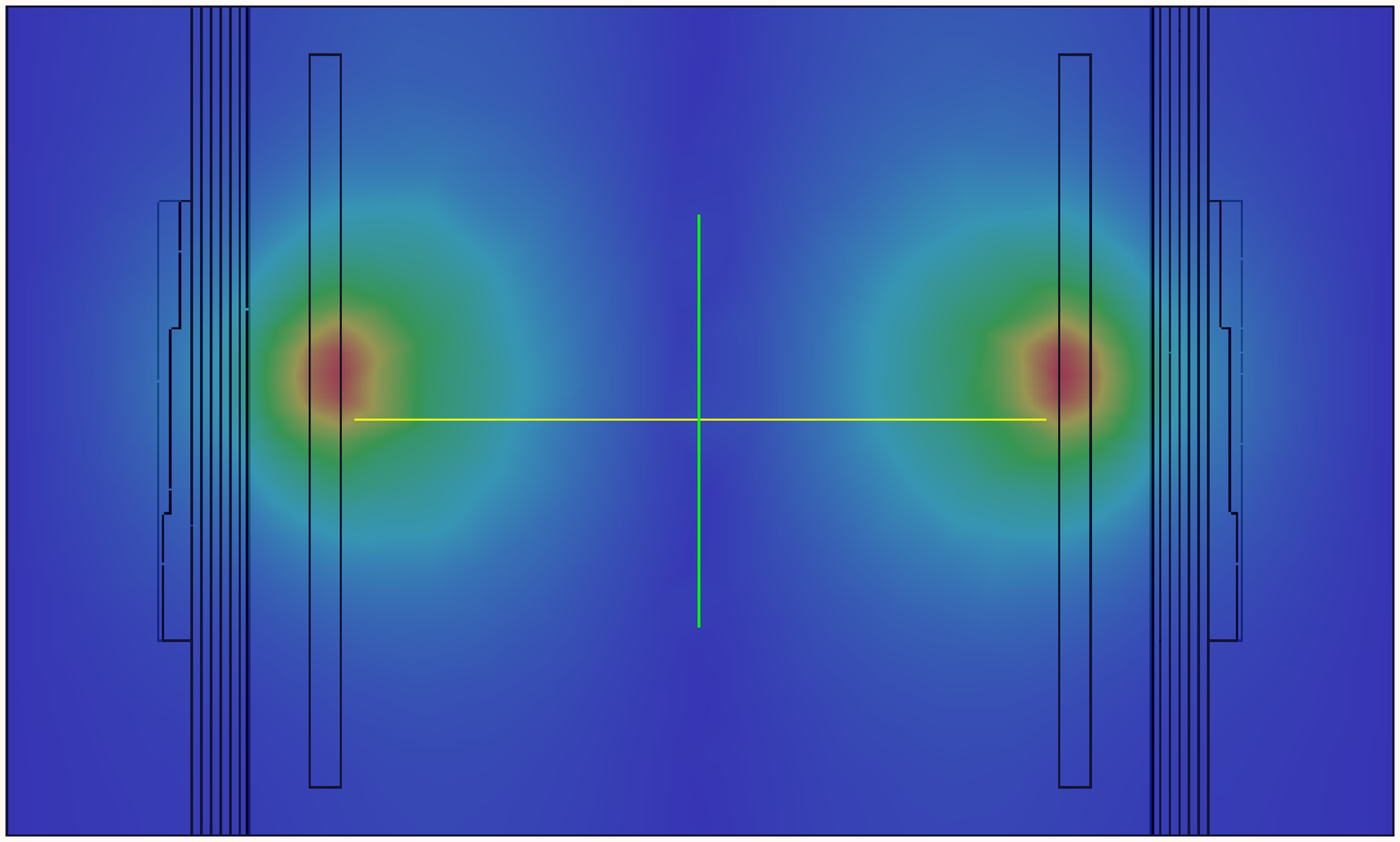} & 
\includegraphics[angle=0,origin=c,width=6cm,height=5cm] {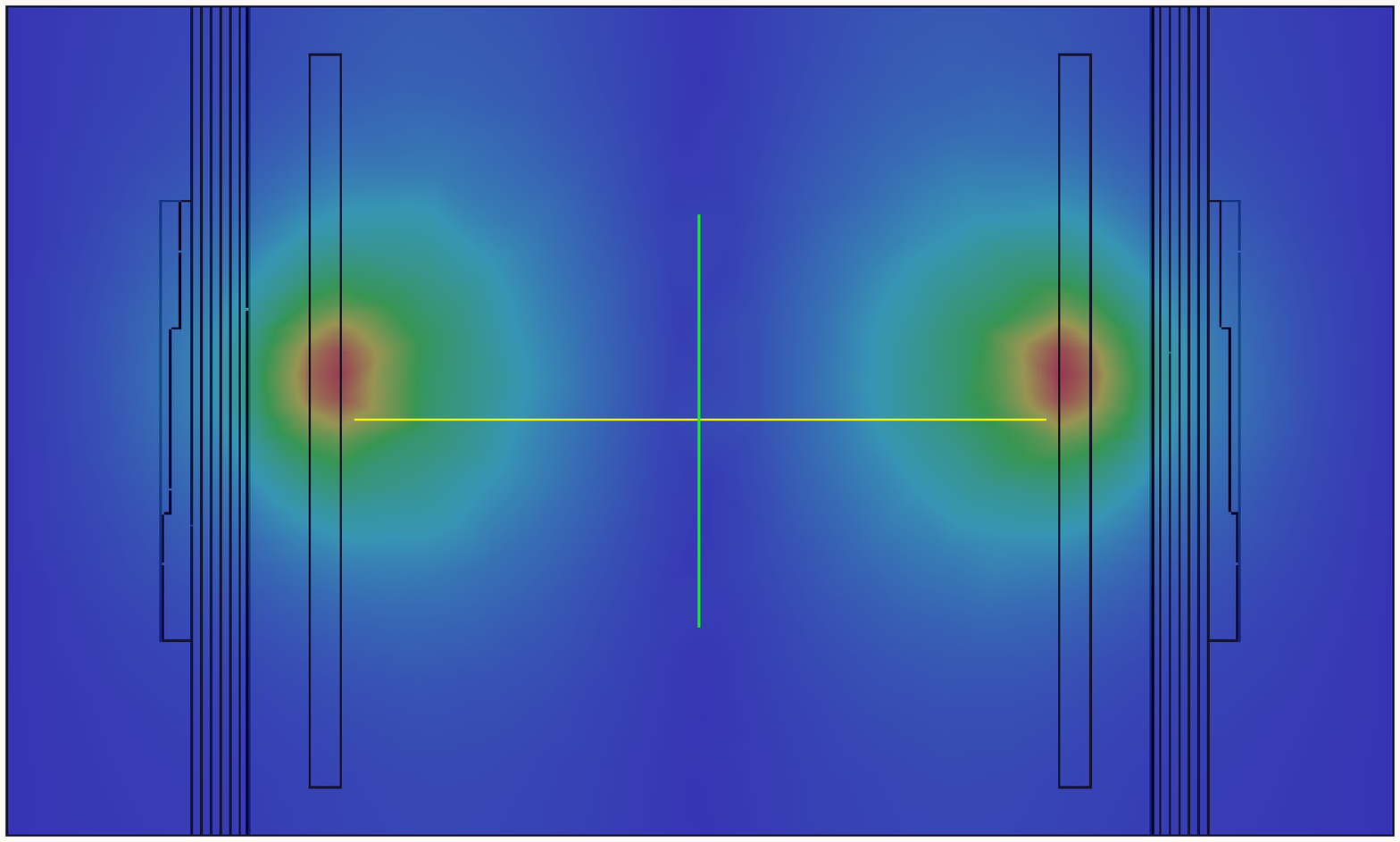} \\
\includegraphics[angle=0,origin=c,width=6cm,height=5cm] {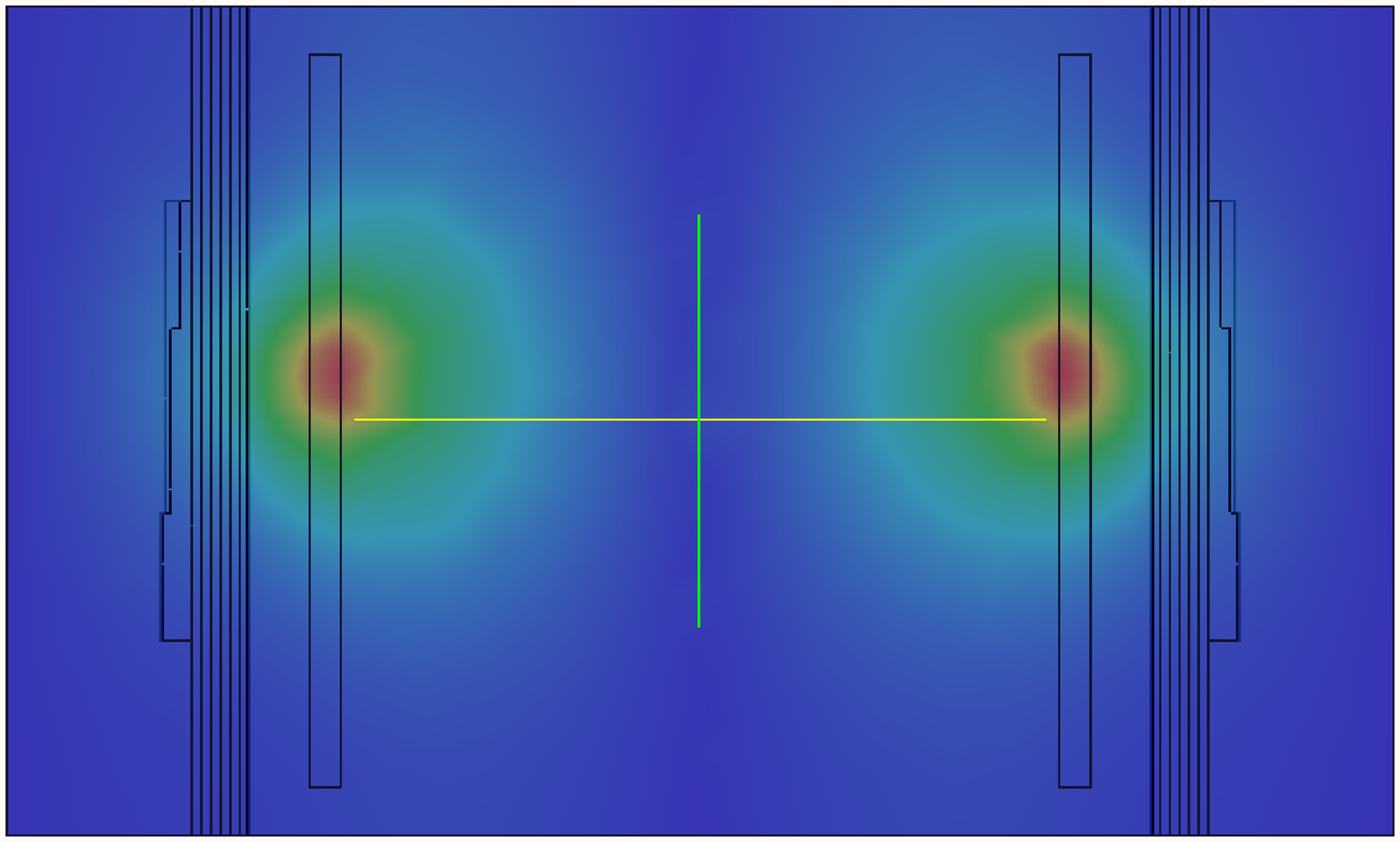} & 
\includegraphics[angle=0,origin=c,width=6cm,height=5cm] {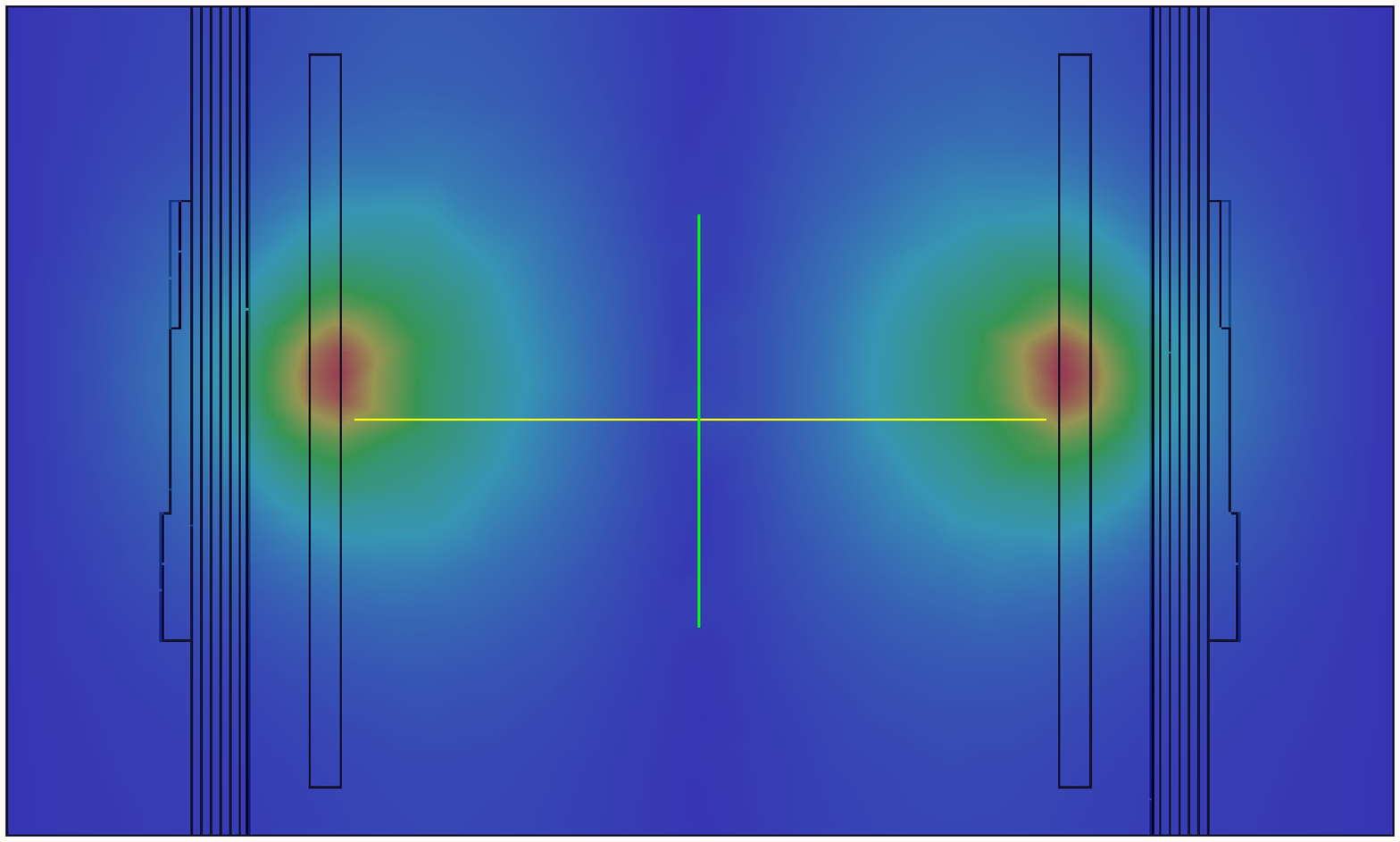} \\
\includegraphics[angle=0,origin=c,width=6cm,height=5cm] {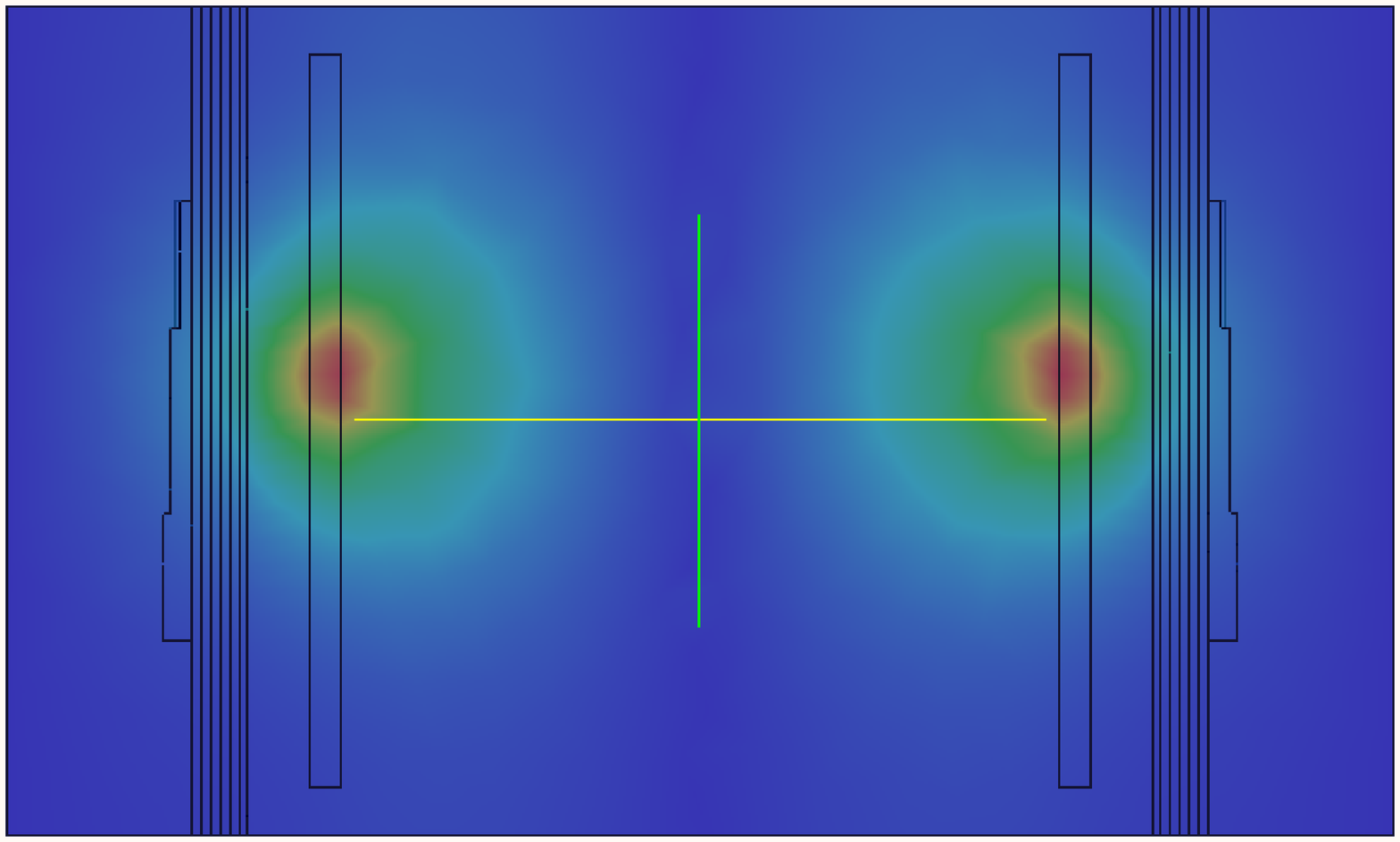} & 
\includegraphics[angle=0,origin=c,width=6cm,height=5cm] {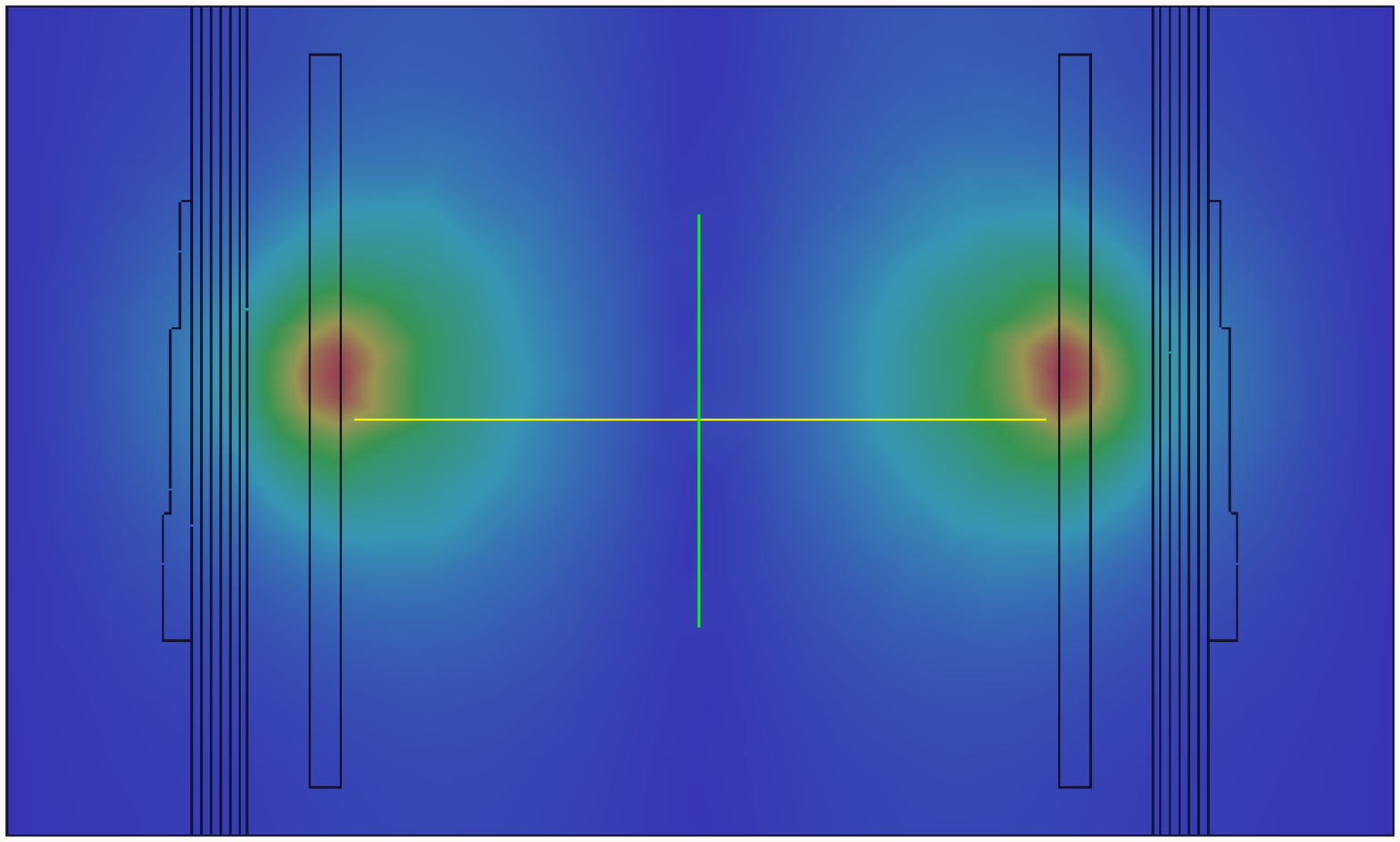} 
\end{tabular}
\caption{( line first presentation of ) Snapshot of the mesh transformation in order to reconstruct the deposit modeled by 3 layers, each graph corresponds to an iteration of inversion. Those solutions are related to results in Fig.~\ref{J_L123}.} 
\label{snapshotsL123}
\end{figure}

\section*{Aknowledgement}
The authors gratefully thanks EDF-R\&D STEP team for their helpful discussions and remarks, who also finances this project, in the postdoctoral framework for the second author. 
\section{Conclusion}

	We set in this report a direct and inverse 3D solver for eddy current probing of deposits. Numerical tests are provided to demonstrate the efficiency of the 3D solver that we validate with comparison to the 2D axisymmetric case. The inverse algorithm in most cases achieves a good convergence error of order 1\%.

	On going work, concerns the regularization of the descent direction of the inverse problem and consider more complicate real life problem, namely the reconstruction of any shaped deposit using non healthy signal du to the presence of the high conductive TSP. 

\clearpage\newpage

\bibliographystyle{alpha}
\bibliography{biblio_eddy_shape}

\end{document}